\documentclass[10pt]{article}

\usepackage[authoryear,round]{natbib}                   % need for \citet*
\usepackage{fixltx2e,amsmath}                                                       % need fixltx2e for \MakeRobust command (used below)
\numberwithin{equation}{section}
\usepackage{amsfonts}                                                       % need for subequations
\usepackage{graphicx}                                                       % need for figures
\usepackage{subfigure}                                                      % need for subfigures
\usepackage{amssymb}                                                        % gives you \mathbb{} font
\usepackage[mathscr]{eucal}                                             % gives you \mathscr font
\usepackage{cancel}                                                             % gives you the ability to visibly cross out terms in equations}
\usepackage[normalem]{ulem}                                                                 % gives you \sout
\usepackage{pstricks}
\usepackage{rotating}
\usepackage{mdwlist}                                                            % need for \suspend{enumerate} text \resume{enumerate} and enumerate*
\usepackage[paperwidth=8.5in,paperheight=11in,top=1.25in, bottom=1.25in, left=1.00in, right=1.00in]{geometry}
\usepackage{mathtools}                                                      % need for show only references'
\mathtoolsset{showonlyrefs=true}                                    % only equations which are labeled AND referenced will be numbered.
\usepackage{amsthm}                                                             % need for theorem-proof environment
\theoremstyle{plain}
\newtheorem{theorem}{Theorem}
\numberwithin{theorem}{section}
\newtheorem{lemma}[theorem]{Lemma}          % [theorem] ==> theorems and lemmas will share a counter
\newtheorem{proposition}[theorem]{Proposition}
\newtheorem{corollary}[theorem]{Corollary}
\theoremstyle{definition}
\newtheorem{definition}[theorem]{Definition}
\newtheorem{example}[theorem]{Example}
\newtheorem{remark}[theorem]{Remark}
\newtheorem{assumption}[theorem]{Assumption}

\MakeRobust{\eqref}                                                                     % allows you to use \eqref{} in \caption{}...
\usepackage{hyperref}

\newcommand{\G}{\Gamma}

\def \R  {{\mathbb {R}}}
\def \N  {{\mathbb {N}}}
\def \x {{\xi}}
\def \g {{\gamma}}

\def \t {{\tau}}

\def \m {{\mu}}

\def \th {{\theta}}

\def \p {{\partial}}
\def \a {{\alpha}}

\def \k {{\kappa}}

\newcommand{\<}{\langle}
\renewcommand{\>}{\rangle}
\renewcommand{\(}{\left(}
\renewcommand{\)}{\right)}
\renewcommand{\[}{\left[}
\renewcommand{\]}{\right]}

\newcommand\Eb{\mathbb{E}}
\newcommand\Pb{\mathbb{Q}}
\newcommand\Rb{\mathbb{R}}
\newcommand\Ib{\mathbb{I}}

\newcommand\Ac{\mathscr{A}}

\newcommand\Ec{\mathscr{E}}
\newcommand\Fc{\mathscr{F}}
\newcommand\Gc{\mathscr{G}}

\newcommand\Lc{\mathscr{L}}
\newcommand\Oc{\mathscr{O}}

\newcommand\Nc{\mathscr{N}}
\newcommand\Mc{\mathscr{M}}

\newcommand\mf{\mathfrak{m}}

\newcommand\eps{\varepsilon}
\newcommand\om{\omega}
\newcommand\Om{\Omega}
\newcommand\sig{\sigma}

\newcommand\gam{\gamma}
\newcommand\Gam{\Gamma}
\newcommand\lam{\lambda}
\newcommand\del{\delta}

\newcommand\nub{\bar{\nu}}

\newcommand\xb{\bar{x}}
\def \l {{\ell}}

\newcommand\hh{\hat{h}}

\newcommand\Hv{\mathbf{H}}
\newcommand\Cv{\mathbf{C}}
\newcommand\mv{\mathbf{m}}

\newcommand\Psiv{\mathbf{\Psi}}
\newcommand\Phiv{\mathbf{\Phi}}

\renewcommand\d{\partial}

\begin{document}

\title{
Pricing approximations and error estimates for local {L}\'evy-type models with default}

\author{
Matthew Lorig
\thanks{ORFE Department, Princeton University, Princeton, USA.  Work partially supported by NSF grant DMS-0739195}
\and
Stefano Pagliarani
\thanks{Centre de Math\'ematiques Appliqu\'ees, Ecole Polytechnique and CNRS, Route de Saclay, 91128 Palaiseau Cedex, France. Email: {\tt  pagliarani@cmap.polytechnique.fr}. Work supported by the Chair {\it Financial Risks} of the {\it Risk Foundation}.}
\and
Andrea Pascucci
\thanks{Dipartimento di Matematica,
Universit\`a di Bologna, Bologna, Italy}
}

\date{\today}

\maketitle

\begin{abstract}
We find approximate solutions of partial integro-differential equations, which arise in financial
models when defaultable assets are described by general scalar L\'evy-type stochastic processes.
We derive rigorous error bounds for the approximate solutions.  We also provide numerical examples illustrating the usefulness and versatility of our methods in a variety of
financial settings.
\end{abstract}


%%%%%%%%%%%%%%%%%%%%%%%%%%%%%%%%%%%%%%%
%
%       SECTION: Introduction
%
%%%%%%%%%%%%%%%%%%%%%%%%%%%%%%%%%%%%%%%

\section{Introduction}
\label{sec:intro} It is now clear from empirical examinations of option prices and high-frequency
data that asset prices exhibit jumps (see, e.g., \cite{Ait-Sahalia,eraker} and references
therein).  From a modeling perspective, the above evidence supports the use of exponential L\'evy
models, which are able to incorporate jumps in the price process through a Poisson random measure.
Moreover, exponential L\'evy models are convenient for option pricing since, for a wide variety of
L\'evy measures, the characteristic function of L\'evy processes are known in closed-form,
allowing for fast computation of option prices via generalized Fourier transforms (see
\cite{lewis2001simple,lipton2002,levendorskiibook,contbook,Almendral2005}).  However, a major
disadvantage of exponential L\'evy models is that they are spatially homogeneous; neither the
drift, volatility nor the jump-intensity have any local dependence.  Thus, exponential L\'evy
models are not able to exhibit volatility clustering or capture the leverage effect, both of which
are well-known features of equity markets.
\par
In addressing the above shortcomings, it is natural to allow the drift, diffusion and L\'evy
measure of a L\'evy process to depend locally on the value of the underlying process.  Compared to
their L\'evy counterparts, \emph{local L\'evy} models (also known as \emph{scalar L\'evy-type}
models) are able to more accurately mimic the real-world dynamics of assets.  However, the
increased realism of local L\'evy models is matched by an increased computational complexity; very
few local L\'evy models allow for efficiently computable exact option prices (the notable
exception being the L\'evy-subordinated diffusions considered in \cite{carr}). Since an option price can directly related to the solution of a partial integro-differential equation (Kolmogorov backward equation) by means of the Feynman-Kac formula, other classical numerical approaches, such as finite difference or Monte Carlo methods, can be employed.  However, such approaches are by no means free of drawbacks (see, for instance, \cite{andersen2000jump,Forsyth2005}).
\par
Recently, there have been a number of methods proposed for finding approximate option prices in
local L\'evy settings.  We mention in particular the work of \cite{gobet-smart}, who use Malliavin
calculus methods to derive analytic approximations for options prices in a setting that includes
local volatility and Poisson jumps.  We also mention the work of \cite{lorig-jacquier-1}, who use
regular perturbation methods to derive option price and implied volatility approximations in a
local-L\'evy setting. {Another polynomial operator expansion technique was proposed in
\cite{PagliaraniPascucci2013} and \cite{RigaPagliaraniPascucci} to compute option prices in
stochastic-local-L\'evy volatility models.}
\par
More recently, \citet*{lpp-1} illustrate how to obtain a family of asymptotic approximations for
the transition density of the full class of scalar L\'evy-type process (including infinite
activity L\'evy-type processes).  The methods developed in \cite{lpp-1} can be briefly described
as follows.  First, one considers the infinitesimal generator of a general scalar L\'evy-type
process.  One expands the drift, volatility and killing coefficients as well as the L\'evy kernel
as an infinite series of analytic basis functions.  The infinitesimal generator can then be
formally written as an infinite series, with each term in the series corresponding to a different
basis function.  Inserting the expansion for the generator into the Kolmogorov backward equation,
one obtains a sequence of nested Cauchy problems for the density of the L\'evy-type process.
\par
The polynomial expansion technique described in \cite{lpp-1} has also been applied in multi-dimensional settings.  In particular, \cite{lorig-pagliarani-pascucci-2} derive explicit approximations and error bounds for implied volatilities for a general class of $d$-dimensional diffusions.  \cite{LPP4} derive error bounds for transition densities and option prices in a general $d$-dimensional diffusion setting.  However, in neither of these papers do the authors consider processes with jumps.  For $d$-dimensional models with jumps, \cite{lorig-pagliarani-pascucci-5} derive explicit approximations for transition densities and option prices.  However, the results are only formal, as no rigorous error bounds are established for the approximation.  %In this paper, we shall provide rigorous results for models with jumps.
{The main contribution of this paper is a rigorous proof of short-time error estimates on transition densities and option prices, under Local L\'evy models with Gaussian jumps. Furthermore, the proof, which is based on a non-trivial generalization of the standard \emph{parametrix} method, paves the road for further extensions in order to include more general choices of L\'evy measures.}
\par
The main contributions of this paper are as follows.  First, we analytically solve the sequence of nested Cauchy problems mentioned above and thereby derive an explicit expression of the approximate option price (i.e., solution of the integro-differential equation) to arbitrarily high order.  Second, we provide a rigorous and detailed proof of some pointwise error estimates for the approximation.  These estimates were announced, \emph{without} proof, in \cite{lpp-1}. Lastly, we illustrate how to implement our approximation formulas in Mathematica, Wolfram's symbolic computation software.  In particular, we provide numerical examples for transition densities, Call and Put prices, implied volatilities, bond prices and credit spreads.  For the readers' convenience, example Mathematica code is also made freely available on the authors' websites.  The numerical tests in this manuscript and the authors' websites clearly demonstrate the versatility and accuracy of the method.
\par
%The purpose of this paper is three-fold:
%\begin{enumerate}
%\item First, we provide an explicit representation of any term in the density (and price) expansion,
%given as an integro-differential operator acting on a L\'evy-type density. This result (Theorem
%\ref{th:un_general_repres}) is an extension of \cite{lpp-1}, where only a recursive representation
%for the Fourier transform of each term in the expansion is provided. Additionally, it
%extends \cite[Theorem 3.8]{LPP4}, where such an explicit representation was given for the purely
%diffusion case.
%\item Second, we provide a rigorous and detailed proof of the pointwise error estimates
%stated in \cite{lpp-1}. The main results are given by Theorem \ref{t1} and Corollary \ref{cor2},
%where global error bounds are stated for both the approximations of densities and prices. These
%estimates are interesting from the theoretical point of view, as they imply some
%non-classical upper bounds for the fundamental solution of a certain class of integro-differential
%operators with variable coefficients.
%\ora{(Andrea, please feel free to add more here if you like.)}
%\end{enumerate}
The rest of this paper proceeds as follows: in Section \ref{sec:model} we describe a financial market in which a defaultable asset evolves as an exponential L\'evy-type
process.  We then relate the problem of the pricing of a European-style option to the solution of a partial integro-differential equation (PIDE).  Next, in Section \ref{sec:formal},  we introduce a family of asymptotic solutions of the pricing PIDE.  The main results are given in Sections \ref{errbou}, where global error bounds are proved for both the density and price approximations (Theorem \ref{t1} and Corollary \ref{cor2}. respectively). In addition to their practical use, these estimates are interesting from the theoretical point, as they imply some non-classical upper bounds for the fundamental solution of a certain class of integro-differential operators with variable coefficients.  The proof of Theorem \ref{t1} is postponed to Section \ref{sec:proof}.  Before proving the theorem, we provide in Section \ref{sec:examples} a number of numerical examples, which are relevant for financial applications.

%%%%%%%%%%%%%%%%%%%%%%%%%%%%%%%%%%%%%%%
%
%       SECTION: Model
%
%%%%%%%%%%%%%%%%%%%%%%%%%%%%%%%%%%%%%%%

\section{Market model and option pricing}
\label{sec:model} For simplicity, we assume a frictionless market, no arbitrage, zero interest
rates and no dividends.  Our results can easily be extended to include locally dependent interest
rates and dividends.  We take, as given, an equivalent martingale measure $\Pb$, chosen by the
market on a complete filtered probability space \mbox{$(\Om,\Fc,\{\Fc_t,t\geq0\},\Pb)$} satisfying
the usual hypotheses. % of completeness and right continuity.
%The filtration $\Fc_t$ represents the history of the market.
All stochastic processes defined below live on this probability space and all
expectations are taken with respect to $\Pb$.  We consider a defaultable asset $S$ whose
risk-neutral dynamics are given by
\begin{align}
\begin{cases}
 S_t=  \Ib_{\{ \zeta > t\}} e^{X_t}, \\
 dX_t =  \mu(t,X_t) dt + \sig(t,X_t) dW_t + \int_\Rb  z \tilde{N}(dt,X_{t-},dz) ,  \\
 \tilde{N}(dt,X_{t-},dz)=  N(dt,X_{t-},dz) - \nu(t,X_{t-},dz) dt, \\
 \zeta=  \inf \left\{ t \geq 0 : \int_0^t \gam(s,X_s) ds \geq \Ec \right\}.
\end{cases}\label{eq:dX}
\end{align}
Here, $X$ is a L\'evy-type process with local drift function $\mu(t,x)$, local volatility function
$\sig(t,x) \geq 0$ and state-dependent Poisson random and L\'evy measures $N(dt,x,dz)$ and $\nu(t,x,dz)$
respectively. %We shall denote by $\Fc_t^X$ the filtration generated by $X$.
The random variable $\Ec \sim \text{Exp}(1)$ has an exponential distribution and is independent of
$X$.  Note that $\zeta$, which represents the default time of $S$, is defined here through the
so-called \textit{canonical construction} (see \cite{bielecki2001credit}). %and is the first
%arrival time of a doubly stochastic Poisson process with local intensity function $\gam(t,x) \geq
%0$.
This way of modeling default is {also} considered in a local volatility setting in
\citet*{JDCEV,linetsky2006bankruptcy}, and for exponential L\'evy models in \cite{capponi}. Notice
that the drift coefficient $\mu$ is fixed by $\sig$, $\nu$ and $\gam$ in order to satisfy the
martingale condition:
%\footnote{
%We provide a derivation of the martingale condition in Section \ref{sec:pricing} Remark \ref{rr11}
%below. }
\begin{align}
\mu(t,x)
    &=  \gam(t,x) - a(t,x) - \int_\Rb \nu(t,x,dz) (e^z-1-z), &
a(t,x) &:= \frac{1}{2} \sig^2(t,x). \label{eq:drift}
\end{align}
%\bigskip
%\ora{I added a comment to Example 3, the orthogonal basis example.  Shall we delete what you have in blue here?  If so, please delete.}\\
%{\andrea The condition $a(t,\cdot), \gam(t,\cdot), \nu(t,\cdot,dz) \in C^{\infty}(\Rb)$
%\footnote{This notation for $\nu$ means: $\nu(t,\cdot,B) \in C^{\infty}(\Rb)$ for any Borel set
%$B$.} could be too strong: for the Taylor expansion of order $n$ we only need the $C^{n}$
%regularity. For the Hermite expansion it suffices that the coefficients are in $L^2$: this is a
%great advantage of the Hermite expansion, at least if models with irregular coefficients will ever
%be used in finance. That could be the case of non-linear models. Perhaps we might want to add a
%picture with the approximation of a step function instead of the exp function by Hermite
%polynomials: this is not strictly necessary also because I guess that several terms are needed to
%get a decent approximation.}
%\bigskip
\par
We assume that the coefficients are measurable in $t$ and suitably smooth in $x$ so as to ensure the
existence of a strong solution to \eqref{eq:dX}  (see, for instance, \citet*{oksendal2}, Theorem 1.19).
We also assume that
\begin{align}
  \nub(dz)    &:= \sup_{(t,x) \in \Rb^+ \times \Rb} \nu(t,x,dz) ,
\end{align}
satisfies the following three boundedness conditions
\begin{align}
\int_\Rb  \min\{|z|,z^2\}\, \nub(dz)
        &<      \infty, &
\int_{|z| \geq 1}  e^z\, \nub(dz)
        &<      \infty, \label{eq:conditions}
\end{align}
which is rather standard assumption for financial applications.
We will relax some of these assumptions for the numerical examples provided in Section
\ref{sec:examples}.  Even without the above assumptions in force, our numerical tests indicate
that our approximation techniques gives very accurate results.

%%%%%%%%%%%%%%%%%%%%%%%%%%%%%%%%%%%%%%%
%
%       SECTION: Pricing
%
%%%%%%%%%%%%%%%%%%%%%%%%%%%%%%%%%%%%%%%

\label{sec:pricing} We consider a European derivative expiring at time $T$ with payoff $H(S_T)$
and we denote by $V$ its price process.
we introduce
\begin{align}
h(x)
    &:= H(e^x)\qquad\text{and}\qquad K:= H(0).
\end{align}
Then, by no-arbitrage arguments (see, for instance,
\citet[Section 2.2]{linetsky2006bankruptcy}) the price of the option at time $t<T$ is given by
%
%For convenience, we introduce
%\begin{align}
%h(x)
%    &:= H(e^x)\qquad\text{and}\qquad K:= H(0).
%\end{align}
%\begin{proposition}\label{p1}
%The price $V_t$ is given by
\begin{align}\label{e1}
 V_{t}
    &=      K+\Ib_{\{\zeta>t\}} \Eb \[e^{-\int_t^T \gam(s,X_s) ds} \left(h(X_T)-K\right)  | X_t \], &
 t
        &\le    T.
\end{align}
%\end{proposition}
%\begin{proof}
%The proof can be found in Section 2.2 of \citet{linetsky2006bankruptcy}.
%\end{proof}
%\noindent
From \eqref{e1} we see that, in order to compute the price of an option, we must
evaluate functions of the form\footnote{Note: we can accommodate stochastic interest rates and
dividends of the form $r_t = r(t,X_t)$ and $q_t=q(t,X_t)$ by simply making the change: $\gam(t,x)
\to \gam(t,x) + r(t,x)$ and $\mu(t,x) \to \mu(t,x) + r(t,x) - q(t,x)$ in PIDE \eqref{eq:v.pide}.}
\begin{align}\label{expectation}
u(t,x)
    &:= \Eb \[e^{-\int_t^T \gam(s,X_s) ds}h(X_T)| X_t = x \] .
\end{align}
By a direct application of the Feynman-Kac representation theorem (see, for instance, Theorem
14.50 in \cite{Pascucci2011}) the classical solution of the following Cauchy problem,
%the unique classical solution of the following Cauchy problem, when it exists, is equal to $v(t,x)$
%we have that $v(t,x)$ is
%equal, when it exists, to the unique classical solution of the Cauchy problem
\begin{align}\label{eq:v.pide}
&\begin{cases}
 (\d_t + \Ac) u(t,x)=0,\qquad & t\in[0,T[,\ x\in\mathbb{R}, \\
 u(T,x) =  h(x),& x \in\mathbb{R},
\end{cases}
\end{align}
when it exists, is equal to the function $u$ defined in \eqref{expectation}. Here $\Ac\equiv \Ac(t,x)$ is the
integro-differential operator associated with the SDE \eqref{eq:dX}
and defined explicitly as %\begin{align}
%\Ac(t) f(x) &=  \gam(t,x) ( \d_{x}f(x) - f(x) ) + a(t,x) ( \d_{xx} f(x) - \d_{x}f(x) )\\
%           &\quad - \int_\Rb \nu(t,x,dz)(e^z-1-z)  \d_{x}f(x) + \int_\Rb \nu(t,x,dz)( f(x+z) - f(x) - z \d_{x}f(x) ) , \label{eq:A}
%\end{align}
\begin{align}
\Ac(t,x) f(x) &=   a(t,x)\d_{xx} f(x) +\m(t,x) \d_{x}f(x) -\gam(t,x) f(x) \\ & \qquad
    + \int_\Rb\left(f(x+z) - f(x) - z \d_{x}f(x)\right)\nu(t,x,dz) \label{eq:A}
\end{align}
with $\m$ and $a$ as in \eqref{eq:drift}.
%For greater convenience, with a slight abuse of notation,
We say that $\Ac$ is the {\it characteristic operator}\footnote{More precisely, $\Ac+\gam$ would be the
characteristic operator of $X_{t}$.} of $X_{t}$. %Moreover, in order to shorten the notation, in
%the sequel we will suppress the explicit dependence on $t$ in $\Ac(t)$ when it is clear from the
%context.
\par
Sufficient conditions for the existence and uniqueness of a classical solution of a second order elliptic
integro-differential equations of the form \eqref{eq:v.pide} are given in Theorem II.3.1 of \citet*{GarroniMenaldi}.
%We denote by $p(t,x;T,y)$ the fundamental solution of the operator $(\d_t + \Ac)$, which is
%defined as the solution of $\eqref{eq:v.pide}$ with $h=\del_y$. Note that $p(t,x;T,y)$ represents
%also the transition density of $\log S$ \footnote{Here with $\log S$ we denote the process
%$X_t\Ib_{\{ \zeta > t\}} -\infty\,\Ib_{\{ \zeta \leq t\}}$.  }
%\begin{align}
% p(t,x;T,y )dy =  \Pb [ \log S_T \in dy | \log S_t = x ], \qquad  x,y\in\Rb, \qquad t < T.
%\end{align}
%Note also that $p(t,x;T,y)$ is not a probability density since (due to the possibility that $S_T=0$) we have
%\begin{align}
%\int_\Rb p(t,x;T,y) dy
%    \leq 1.
%\end{align}
%Given the existence of the fundamental solution of $(\d_t + \Ac)$, we have that for any $h$ that
%is integrable with respect to the density $p(t,x;T,\cdot)$, the Cauchy problem \eqref{eq:v.pide}
%has a classical solution that can be represented as
%\begin{align}
%u(t,x)
%    &= \int_\Rb h(y) p(t,x;T,y) dy. \label{eq:v.def1}
%\end{align}
%In the sequel we will make systematic use of the classical Fourier transform
%\begin{align}\label{eq:def_Fourier}
%\Fc_x f(\cdot)(\xi):=\int_{\mathbb{R}}e^{i x\xi}f(x) d x,\qquad \xi\in\mathbb{R},\quad f\in L^{1}(\mathbb{R}).
%\end{align}
In particular, given the existence of the fundamental solution $p(t,x;T,y)$ of $\left(\d_t +
\Ac\right)$, we have that for any integrable datum $h$, the Cauchy problem \eqref{eq:v.pide} has a
classical solution that can be represented as
\begin{align}
 u(t,x) &= \int_\Rb h(y) p(t,x;T,y) dy. \label{eq:v.def1}
\end{align}
Notice that $p(t,x;T,y)$ is a ``defective'' probability density since (due to the possibility that
$S_T=0$) we have
\begin{align}
 \int_\Rb p(t,x;T,y) dy \leq 1.
\end{align}

\section{Approximate densities and option prices via polynomial expansions}
\label{sec:formal} In this section we describe the approximation methodology and define the notation that will be needed in subsequent sections.
%for the full proof of the results stated in \citet*{lpp-1} which will be given in the next section.
%Our goal is to construct an approximate solution of Cauchy problem \eqref{eq:v.pide}. Throughout

\begin{definition}
\label{def:A}
For any $n\leq N\in\mathbb{N}_0$, let $a_{n}=a_{n}(t,x)$, $\gamma_{n}=\gamma_{n}(t,x)$
and $\nu_n=\nu_n(t,x,d z)$ be such that the following hold:
\begin{enumerate}
%\item[i)] the first $N$ coefficients in the sequence $(a_{\alpha,n}(t,x))_{n \geq 0}$ are continuous functions that depend polynomially on $x$ with $a_{\alpha,0}(t,x) \equiv a_{\alpha,0}(t)$,
\item[(i)] For any $t \in [0,T]$, the functions $a_{n}(t,\cdot)$, $\gamma_{n}(t,\cdot)$ are polynomials with $a_{0}(t,x) \equiv a_{0}(t)$, $\gamma_{0}(t,x) \equiv \gamma_{0}(t)$, and for any $x\in\Rb$ the functions $a_{n}(\cdot,x)$, $\gamma_{n}(\cdot,x)$ belong to
$L^{\infty}([0,T])$.
\item[ii)] %for every $(t,x)$, the sequence $(\nu_n(t,x,dz))_{n \geq 0}$ is a sequence of real measures (not necessarily non-negative) on $\mathbb{R}^d$, the first $N$ of which are continuous in $(t,x)$, depend polynomially on $x$
For any $t\in [0,T]$, $x\in\mathbb{R}$, we have
 \begin{align}\label{e34}
  \nu_n(t,x,d z)
    &=\sum_{m=0}^{M_{n}}x^{m}  \nu_{n,m}(t,d z), & M_n &\in \mathbb{N}_{0} ,
\end{align}
where each $\nu_{n,m}(t,d z)$ satisfies condition \eqref{eq:conditions}.  Moreover, $M_{0}=0$, $\nu_{0}\ge0
$ and
\begin{align}\label{e21aa}
   \int_{|z|\geq 1}e^{\lambda |z|}\nu_{0}(t,d z)<\infty,\qquad t\in[0,T],
\end{align}
for some positive $\lambda$.
%\item[iii)] we have convergence:
%\begin{align}
%\nu
%    &=  \sum_{n=0}^\infty \nu_n , &
%a_\alpha
%    &=  \sum_{n=0}^\infty a_{\alpha,n}, &
%|\alpha|
%    &\leq 2, \label{eq:a.sum}
%\end{align}
%in some sense (pointwise or in norm).
%\item[(ii)] for any $t \in [0,T]$ and $B \in \Bc(\Rb_0^d)$ the functions $\nu_n(t,\cdot,B)$ and $a_{\alpha,n}(t,\cdot)$ are polynomials with
%\begin{align}
%\nu_0(t,\cdot,B)
%    &=\nu_0(t,B) &
%    &\text{and}&
%a_{\alpha,0}(t,\cdot)
%    &=a_{\alpha,0}(t),
%\end{align}
%\item[(iii)] for any $x\in\R^{d}$ the functions $a_{\alpha,n}(\cdot,x)$ belong to
%$L^{\infty}([0,T])$.
\end{enumerate}
Then we say that $\( \Ac_n(t) \)_{0\le n \le N}$, defined by
%\begin{align}
% \Ac_n(t,x)f(x)=\Ac_n(t)f(x) &= \sum_{|\beta|\le M_{n}}x^{m} \int_{\mathbb{R}} \nu_{n,m}(t,d z) \left( e^{z\, \partial_x} - 1 - z\, \partial_x \right)f(x)
%            + \sum_{|\alpha |\leq 2}  a_{\alpha,n}(t,x) D^{\alpha}_x f(x) \\
%            &\equiv  \int_{\mathbb{R}} \nu_n(t,x,d z) \left( e^{\<z,\nabla_x \>} - 1 - \<z,\nabla_x \>\right)f(x)
%            + \sum_{|\alpha |\leq 2}  a_{\alpha,n}(t,x) D^{\alpha}_x f(x),      \label{eq:An}
%\end{align}
\begin{align}
 \Ac_n(t,x)f(x)&=%\Ac_n(t)f(x) &=
 a_n(t,x) ( \d_{xx} f(x) - \d_{x}f(x))+\gam_n(t,x)\left(\d_{x}f(x)-f(x)\right)\\
           &\quad +\sum_{m=1}^{M_{n}}x^{m}\left( - \int_\Rb (e^z-1-z) \, \nu_{n,m}(t,dz)\,  \d_{x}f(x) + \int_{\mathbb{R}} \left( e^{z\, \partial_x} - 1 - z\, \partial_x \right)f(x)\, \nu_{n,m}(t,d z) \right) \label{eq:A_n}
 \\
            &\equiv   a_n(t,x) ( \d_{xx} f(x) - \d_{x}f(x))+\gam_n(t,x)\left(\d_{x}f(x)-f(x)\right)\\
           &\quad - \int_\Rb (e^z-1-z) \, \nu_n(t,x,dz)\,  \d_{x}f(x) +
           \int_\Rb ( f(x+z) - f(x) - z \d_{x}f(x) )\,\nu_n(t,x,dz),      \label{eq:An}
\end{align}
is an \emph{$N$-th order polynomial expansion} of $\Ac(t)$.
\end{definition}
Definition \ref{def:A} allows for very general polynomial specifications.  The idea is to choose
an expansion $(\Ac_n(t))$ that closely approximates $\Ac(t)$, i.e. formally one has
\begin{align}\label{eq:A.expand}
 \Ac(t,x)=
    %{\andrea \phi(t,x,\Dc)= \sum_{n=0}^\infty B_{n}(x)\phi_n(t,\Dc)=}
    \sum_{n=0}^\infty \Ac_{n}(t,x).
\end{align}
The precise sense of this approximation will depend on the application.  Below, we present three polynomial expansions. The first two expansion schemes provide an accurate approximation $\Ac(t,x)$ in a pointwise local sense, under the assumption of smooth coefficients. The last expansion scheme approximates $\Ac(t,x)$ in a global sense and can be applied even in the case of discontinuous coefficients.
%\noindent
%For a fixed polynomial expansion basis $(a_n,\gamma_n,\nu_n)_{n\ge 0}$, the operator $\Ac$ can be formally written as
%\begin{align}\label{eq:A.expand}
% \Ac(t,x)=
%    \sum_{n=0}^\infty \Ac_{n}(t,x),
%\end{align}
%where the operators $\Ac_n=\Ac_n(t,x)$ act as
%\begin{align}
% \Ac_n(t,x) f(x) &=  a_n(t,x) ( \d_{xx} f(x) - \d_{x}f(x))+\gam_n(t,x)\left(\d_{x}f(x)-f(x)\right)\\
%           &\quad - \int_\Rb (e^z-1-z) \, \nu_n(t,x,dz)\,  \d_{x}f(x) +
%           \int_\Rb ( f(x+z) - f(x) - z \d_{x}f(x) )\,\nu_n(t,x,dz). \label{eq:A_n}
%\end{align}
%Let us describe some useful choices of approximating sequences.
\begin{example}\label{example:Taylor}(Taylor polynomial expansion)\\
Assume the coefficients $a(t,\cdot),\gamma(t,\cdot)\in C^{N}(\mathbb{R})$ and that the compensator
$\nu$ takes the form
 $$\nu(t,x,d z)=h(t,x,z)\bar{\nu}(d z)$$
where $h(t,\cdot,z)\in C^{N}(\mathbb{R})$ with $h\geq 0$, and $\bar{\nu}$ is a L\'evy measure.
%is
%for any Borel set $B \in \Bc(\Rb^d)$ the function $\nu(t,\cdot,B)$ to be differentiable, up to order $N$.
Then, for any fixed $\bar{x}\in\R$ and $n\le N$, we define $a_{n}$, $\gamma_{n}$ and $\nu_n$ as the
$n$th order term of the Taylor expansions of $a$, $\gamma$ and $\nu$ respectively in the spatial
variables $x$ around the point $\bar{x}$.  That is, we set
\begin{equation}\label{eq:taylor_coefficients}
       a_{n}(t,x) = \frac{\partial_x^n a(t,\bar{x})}{n!}(x-\bar{x})^{n},\qquad \gamma_{n}(t,x) = \frac{\partial_x^n \gamma(t,\bar{x})}{n!}(x-\bar{x})^{n},\qquad  \nu_n(t,x,d z)
  =       \frac{\partial_{x}^{n}h(t,\bar{x},z)}{n!}(x-\bar{x})^{n}\bar{\nu}(d z). %& 0    &\leq n\leq N,
\end{equation}
%where as usual {$\beta!=\beta_{1}!\cdots\beta_{d}!$ and $x^\beta = x_1^{\beta_1} \cdots x_d^{\beta_d}$}.
The expansion proposed in \cite{lorig-pagliarani-pascucci-2} and
\cite{lorig-pagliarani-pascucci-3} is the particular case when $\nu\equiv 0$.
\end{example}
\begin{example}\label{example:TimeTaylor}(Time-dependent Taylor polynomial expansion)\\
% Assume the coefficients $a_{\alpha}(t,\cdot)\in C^{N}(\mathbb{R}^d)$, and that $\nu(t,x,d z)=h(t,x,z)m(d z)$ with $h(t,\cdot,z)\in C^{N}(\mathbb{R}^d)$.
Under the assumptions of Example \ref{example:Taylor},
fix a trajectory $\xb:\Rb_+ \to \Rb$.
%let the basis point of the
%expansion to depend on time: $\xbar=\xbar(t)$.
We then define $a_{n}$, $\gamma_{n}$ and $\nu_n$ as the
$n$th order term of the Taylor expansions of $a$, $\gamma$ and $\nu$ respectively around $\bar{x}(t)$. %More precisely, we set
%\begin{align}
% \nu_n(t,x,d z)
%    &=\sum_{|\beta|=n}\frac{D_x^{\beta}h(t,\bar{x}(t), z)}{\beta!}(x-\bar{x}(t))^{\beta}\bar{\nu}(d z),\\
%     a_{\alpha,n}(t,x)
%    &=\sum_{|\beta|=n}\frac{D_x^{\beta}a_{\alpha}(t,\bar{x}(t))}{\beta!}(x-\bar{x}(t))^{\beta}, &|\alpha| &\leq 2 .
%\end{align}
This expansion for the coefficients allows the expansion point $\bar{x}$ of the Taylor
series to evolve in time according to the evolution of the underlying process $X_t$. For instance, one could choose $\bar{x}(t)=\mathbb{E}[X_t]$. In \cite{lorig-pagliarani-pascucci-2} this choice results in a highly accurate approximation for option prices and implied volatility in the \cite{heston} model, recently included in the open-source financial library QuantLib.
\end{example}
%\ora{******** END CHANGES HERE ************}

\begin{example}\label{example:Hilbert}(Hermite polynomial expansion)\\
Hermite expansions can be useful when the diffusion coefficients are discontinuous.  A remarkable
example in financial mathematics is given by the Dupire's local volatility formula for models with
jumps (see \cite{frizyor2013}). In some cases, e.g., the well-known Variance-Gamma model, the
fundamental solution (i.e., the transition density of the underlying stochastic model) has
singularities.  In such cases, it is natural to approximate it in some $L^{p}$ norm rather than in
the pointwise sense. For the Hermite expansion centered at $\xb$, one sets
\begin{align}
a_{n}(t,x)
    &=
                \< \Hv_n(\cdot-\bar{x}) , a(t,\cdot) \>_{\Gamma} \Hv_n(x-\bar{x}),\qquad   \gamma_{n}(t,x)            =
                \< \Hv_n(\cdot-\bar{x}) , \gamma(t,\cdot) \>_{\Gamma} \Hv_n(x-\bar{x}),\\
       \nu_{n}(t,x,d z)
    &=
                \< \Hv_n(\cdot-\bar{x}) , \nu(t,\cdot,d z) \>_{\Gam} \Hv_n(x-\bar{x}),
                \end{align}
where the inner product $\<\cdot,\cdot\>_\Gam$ is an integral over $\Rb$ with a Gaussian weighting
centered at $\xb$ and $\Hv_n(x)$ % = H_{n}(x)$ where $H_n$
is the $n$-th one-dimensional Hermite polynomial (properly normalized so that $\< \Hv_m,
\Hv_n \>_\Gam = \delta_{m,n}$ with $\delta_{m,n}$ being the Kronecker's delta function).
%\ora{Why do we have two notations for the Hermite polynomials?  We should stick with eitehr $H_n$
%or $\Hv_n$.}
\end{example}
\begin{remark}
\label{rmk:order}
%\blu{(I added this remark)}\\
Although in each of the above examples, $a_n$, $\gam_n$ and $\nu_n$ are polynomials in $x$ of degree $n$, this is not a requirement of our expansion method.  The degree of $a_n$, $\gam_n$ and $\nu_n$ may be greater than, equal to, or less than $n$.
\end{remark}
\par
We now return to Cauchy problem \eqref{eq:v.pide}.
Following the classical perturbation approach, we expand the solution $u$ as an infinite sum
\begin{align}
  u  &=  \sum_{n=0}^\infty u_n. \label{eq:v.expand}
\end{align}
Inserting \eqref{eq:A.expand} and \eqref{eq:v.expand} into \eqref{eq:v.pide} we find that the
functions $(u_{n})_{n\geq 0}$ satisfy the following sequence of nested Cauchy problems
\begin{align}\label{eq:v.0.pide}
&\begin{cases}
 (\d_t + \Ac_0) u_0(t,x) =  0,\qquad & t\in[0,T[,\ x\in\mathbb{R}, \\
 u_0(T,x) =  h(x),& x \in\mathbb{R},
\end{cases}
\intertext{and} \label{eq:v.n.pide} &\begin{cases}
 (\d_t + \Ac_0) u_n(t,x) =  - \sum\limits_{k=1}^{n} \Ac_k(t,x) u_{n-k}(t,x),\qquad &  t\in[0,T[,\ x\in\mathbb{R}, \\
 u_n(T,x) =  0, &    x \in\mathbb{R}.
\end{cases}
\end{align}
\begin{remark}
\label{rmk:choice}
%\blu{(I added this remark)}\\
In fact, the nested sequence of Cauchy problems \eqref{eq:v.0.pide}-\eqref{eq:v.n.pide} satisfied by the sequence of functions $(u_n)$ is a particular choice.  This choice can be motivated by considering a family of Cauchy problems, indexed by a small parameter $\eps$
\begin{align}
(\d_t + \Ac^\eps) u^\eps
    &=  0 , &
u^\eps(T,x)
    &=  h(x) , &
\Ac^\eps
    &=  \sum_{n=0}^\infty \eps^n \Ac_n , &
\eps
    &\in [0,1] . \label{eq:matt}
\end{align}
Note that, by \eqref{eq:A.expand}, we formally have $\Ac^\eps|_{\eps=1}=\Ac$.  If one seeks a solution to \eqref{eq:matt} of the form $u^\eps = \sum_{n=0}^\infty \eps^n u_n$, then, collecting terms of like powers of $\eps$ one finds that $u_0$ and $u_n$ satisfy \eqref{eq:v.n.pide} and \eqref{eq:v.n.pide}, respectively.
\end{remark}
\subsection{Expression for $u_0$}
Notice that
%\begin{align}
$\Ac_{0}=\Ac_{0}(t,x)$
%\end{align}
is the characteristic operator of the following additive process
\begin{equation}
dX^{0}_t =  \left(   \gam_0(t) - a_0(t) - \int_\Rb (e^z-1-z) \, \nu_0(t,dz)  \right) dt +
\sqrt{2a_0(t)} dW_t + \int_\Rb  z\left(N^{(0)}_t(dt,dz) - \nu_0(t,dz) dt\right), \label{eq:dX_0}
\end{equation}
whose characteristic function $\hat{p}_0(t,x;T,\xi)$ is given explicitly by %{\blue\bf[Please check! I think that the formula was wrong.]}
\begin{align}
 \hat{p}_0(t,x;T,\xi)
    &:= \Eb[ e^{i \xi X_T^0} | X_t^0 = x ] %\label{eq:char.E} \\
  %&
  =\exp\left(i \xi x +\Phiv_0(t,T,\xi) \right),
    \label{eq:charact_0}
\end{align}
where
\begin{equation}\label{eq:Phi_0}
\Phiv_0(t,T,\xi)=\left(i \xi\, \mv(t,T) -\frac{1}{2}\Cv(t,T)\xi^2+\Psiv(t,T,\xi)-\int_t^T \g_0(s) d s\right),
\end{equation}
and with $\mv(t,T)$, $\Cv(t,T)$ and $\Psiv(t,T,\xi)$ being defined as
\begin{align}\label{eq:op_m_C}
  \mv(t,T)&:=\int_t^T \left( \gam_0(s) - a_0(s)
       - \int_\Rb (e^z-1-z)\, \nu_0(s,dz)\right)ds,\\
  \Cv(t,T)&:=\int_t^T 2 a_0(s)ds,\\ \label{eq:op_psi}
  \Psiv(t,T,\xi) &:=\int_t^T \int_\Rb  (e^{i z\xi}-1-iz \xi)\, \nu_0(s,dz)ds.
\end{align}
Note, the additive process $X^0$ in \eqref{eq:dX_0} is assumed to be defined on an appropriate probability space.
It is well-known that additive processes can be constructed as time-changed L\'evy processes (see \cite[Chapter 14]{contbook}).
%The time-dependent jump-measure $N^{(0)}$ can be constructed as a time-changed Poisson random measure $N^{(0)}_t([0,t],dz) = P_{\Lam(t)}(dz)$ where $P$ is a Poisson random measure and $\Lam(t) := \int_0^t \nu_0(s,dz) ds$.  We refer the reader to \cite[Chapter 14]{contbook} for details.
The fundamental solution $p_{0}$ of $\left(\p_{t}+\Ac_{0}\right)$, {which exists if $a_0>0$
\cite[Proposition 28.3]{sato1999levy},} can be recovered by Fourier inversion since, by the first
equality in \eqref{eq:charact_0}, we have
\begin{align}\label{eq:def_Fourier}
 \hat{p}_0(t,x;T,\xi)=\Fc_y p_{0}(t,x,;T,\cdot)(\xi):=\int_{\mathbb{R}}e^{i y\xi}p_{0}(t,x,;T,y) d y ,
\end{align}
and therefore
\begin{align}\label{eq:p0.ft}
p_{0}(t,x,;T,y)
    &=  \Fc_y^{-1} \hat{p}_0(t,x;T,\cdot)(y) := \frac{1}{2\pi}\int_{\mathbb{R}}e^{-i y\xi}\hat{p}_0(t,x;T,\xi) d \xi .
\end{align}
Given the fundamental solution $p_{0}(t,x,;T,y)$, we have the representation for the solution $u_{0}$ of problem \eqref{eq:v.0.pide}
\begin{align}\label{u0expl}
 u_0(t,x) &=   \int\limits_{\mathbb{R}} p_{0}(t,x;T,y) h(y)\, d y ,\quad t<T,\ x\in\mathbb{R} .
\end{align}
Assume that the payoff function $h$ and its Fourier transform $\hat{h}(y) \in L^1(\Rb,dy)$. Then, by inserting the expression \eqref{eq:p0.ft} for $p_{0}(t,x,;T,y)$ into \eqref{u0expl} and integrating with respect to $\xi$, we also have the following alternative representation
\begin{align}
 u_0(t,x) &= \frac{1}{2\pi}  \int\limits_{\mathbb{R}} \hat{p}_{0}(t,x;T,\xi) \hat{h}(-\xi)\, d \xi %, &
%\hat{h} &:=  \Fc^{-1}h , &
%h & \in L^1(\Rb,dy)
. \label{eq:u0.fourier}
\end{align}
%If $h(y) \notin L^1(\Rb,dy)$ but $h(y)e^{cy}  \in L^1(\Rb,dy)$ for some $c \in \Rb$ (which is the case for Call and Put payoffs), one can still use expression \eqref{eq:u0.fourier} by fixing an imaginary component of $\xi$.  This technique, known as a \emph{generalized Fourier transform}, is described in detail in \cite{lewis2000,lipton2002}.\\

%\begin{proposition}
%The leading term $u_0$ in the expansion \eqref{eq:v.expand} is explicitly given by where $p_{0}$
%is a L\'evy-type density whose Fourier transform is given by
%
%\end{proposition}
%\begin{proof}
%It is a direct consequence of Remark \ref{rem:A_0} combined with definition \eqref{eq:v.0.pide}, condition \eqref{parabolicity}, and the L\'evy-Khintchine representation.
%\end{proof}
%\begin{align}
% p_{0}(t,x;T,y) &=  \frac{1}{  \sqrt{(2\pi)^{d}|\Cv(t,T)|} }
%    \exp\left(-\frac{1}{2}\langle\Cv^{-1}(t,T) (y - x-\mv(t,T)),
%    (y -x-\mv(t,T))\rangle\right) \label{e22and}
%\end{align}
%with covariance matrix $\Cv(t,T)$ and mean vector $x+\mv(t,T)$ given by:
%\begin{align}\label{eq:covariance_mean}
% \Cv(t,T)=  \int_t^T C(s) d s,\qquad\qquad \mv(t,T)=   \int_t^T m(s) d s.
%\end{align}
\subsection{Expression for $u_n$}
The following theorem provides an explicit formula for $u_{n}(t,x)$ in \eqref{eq:v.n.pide} %the price $u(t,x)$,
%characteristic function $\hat{p}(t,x;T,\xi)$ of $X_{t}$,
expressed in terms of integro-differential operators applied to
%$\hat{p}_{0}(t,x;T,\xi)$
$u_{0}(t,x)$ in \eqref{u0expl}.
\begin{theorem}\label{th:un_general_repres}
Fix $N\in \mathbb{N}$ and let $\( \Ac_n(t) \)_{0\le n \le N}$ be an $N$th order polynomial
expansion of $\Ac$ as in Definition \ref{def:A}. %If problem \eqref{eq:v.n.pide} has a
%classical solution $u_{n}$
For any $1\leq n\leq N$, %let $u_{n}$ be as in \eqref{eq:v.n.pide}: then
%Then, %for any $n\geq 1$,
we have
\begin{align}\label{eq:un_tilde}
 u_n(t,x)%:=\Fc_y p_n(t,x;T,\cdot)(\xi)
  = \Lc^x_n(t,T) u_0(t,x),\quad t<T,\ x,\xi\in\mathbb{R},\qquad 1\leq n\leq N,
\end{align}
with $u_0$ as in \eqref{u0expl} and
\begin{align}\label{eq:def_Ln}
\Lc_n^x(s_0,T) :=  \sum_{h=1}^n \int_{s_0}^T d s_1 \int_{s_1}^T d s_2 \cdots \int_{s_{h-1}}^T
d s_h
      \sum_{i\in I_{n,h}}\Gc^x_{i_{1}}(s_0,s_1) \cdots \Gc^x_{i_{h}}(s_0,s_h) ,
\end{align}
where\footnote{ For instance, for $n=3$ we have $I_{3,3}=\{(1,1,1)\}$, $I_{3,2}=\{(1,2),(2,1)\}$
and $I_{3,1}=\{(3)\}$. }
\begin{align}\label{eq:def_Ln_bis}
I_{n,h}
    &=  \{i=(i_{1},\dots,i_{h})\in\mathbb{N}^{h} \mid i_{1}+\dots+i_{h}=n\} , &
1
    & \le h \le n ,
\end{align}
and $\Gc^x_{n}(t,s)$ is the operator (see Remark \ref{rem:and1} below)
\begin{align}\label{def_Gn}
 \Gc^x_{n}(t,s) &:=  \Ac_n\left(s,\Mc^x(t,s)\right),
% \Gc^x_{n}(t,s) &:=  B_{n}\left(\Mc^x(t,s)\right) \Ac_n(s),
\end{align}
%\begin{align}\label{def_Gn}
%\Gc_{n}(t,s)
%    &:= \sum_{|\alpha| \leq 2} a_{\alpha,n}\left(s,\Mc(t,s)\right) D^{\alpha}_{x},
%\end{align}
with $\Mc^x(t,s)$ acting as
\begin{align} \label{eq:M}
\Mc^x(t,s)f(x)=& \left(x+\mv(t,s)+\Cv(t,s)\, \partial_x\right) f(x)+\int_t^s \int_{\mathbb{R}}
\left(f(x+z)-f(x)\right)z\, \nu_0(r,dz)dr.
\end{align}
%In particular the solution $u_n$ of \eqref{eq:v.n.pide}, if it exists, is given by
%\begin{align}\label{eq:un}
%  u_n(t,x)   &=  \Lc^x_n(t,T) u_0(t,x),\quad t<T,\ x\in\mathbb{R}.
%\end{align}
\end{theorem}
\begin{remark}\label{rem:and1}
The operator in \eqref{def_Gn} can be written more explicitly as
\begin{align}
 \Ac_n\left(s,\Mc^x(t,s)\right) f(x) =&  a_n(t,\Mc^x(t,s)) ( \d_{xx} f(x) - \d_{x}f(x))+\gam_n(t,\Mc^x(t,s))
 \left(\d_{x}f(x)-f(x)\right)\\
 &- \sum_{m=1}^{M_{n}}\big(\Mc^x(t,s)\big)^{m}  \int_\Rb (e^z-1-z) \, \nu_{n,m}(t,dz)\,  \d_{x}f(x) \\
 & +  \sum_{m=1}^{M_{n}}\big(\Mc^x(t,s)\big)^{m} \int_\Rb \big(f(x+z) - f(x) - z \d_{x}f(x)\big)\nu_{n,m}(t,dz). \label{eq:A_n1}
\end{align}
\end{remark}
\begin{remark}\label{rem:LPP4}
Theorem \ref{th:un_general_repres} extends the novel representation given in \cite[Theorem
3.8]{LPP4}, which is given for the purely diffusion case.  When no jump component is present the
operator $\Mc^x$ in \eqref{eq:M} reduces to
\begin{align}
\Mc^x(t,s)= x+\mv(t,s)+\Cv(t,s)\, \partial_x.
\end{align}
\end{remark}
\begin{remark}
The expression for $u_n$ given in \eqref{eq:un_tilde} can be used in two ways.  First, if the fundamental solution $p_0(t,x;T,y)$ is explicitly available (this is always the case in the purely diffusive setting), then to obtain $u_n$ one can apply the operator $\Lc_n^x(t,T)$ directly to $p_0(t,x;T,y)$ in \eqref{u0expl}.  Second, if $p_0(t,x;T,y)$ is not available explicitly, then one can obtain a Fourier representation for $u_n$ by applying the operator $\Lc_n^x(t,T)$ directly to $\hat{p}_0(t,x;T,\xi)$ in \eqref{eq:u0.fourier}. The details of the latter approach will be shown in Subsection \ref{subsec:fourier}.
\end{remark}
\begin{proof}[Proof of Theorem \ref{th:un_general_repres}]
Let $p_{0}$ be formally defined by \eqref{eq:def_Fourier}. The proof of Theorem
\ref{th:un_general_repres} relies on the following symmetry properties:
%\begin{lemma}\label{lem:simmetry}
for any $t<s$ and $x,y\in\mathbb{R}$, we have
\begin{align}
%\int_{\mathbb{R}} p_{0}(t,x+z;s,y) \nu_n(s,dz)
%    &=\int_{\mathbb{R}} p_{0}(t,x;s,y-z) \nu_n(s,dz),
p_{0}(t,x;s,y)
    &= p_{0}(t,0;s,y-x),
\label{prop:convol}
\\
\partial_x p_{0}(t,x;s,y)
    &= -\partial_{y} p_{0}(t,x;s,y),\label{prop:deriv}
\intertext{and} \label{prop:polyn1}
 y\, p_{0}(t,x;s,y) &=  \Mc^{x}(t,s)p_{0}(t,x;s,y),\\ \label{prop:polyn2}
 x\, p_{0}(t,x;s,y) &=\bar{\Mc}^{y}(t,s)p_{0}(t,x;s,y),
\end{align}
with $\bar{\Mc}^y(t,s)$ acting as
\begin{align} \label{eq:M_bar}
\bar{\Mc}^y(t,s)f(y)=& \left(y-\mv(t,s)+\Cv(t,s)\partial_y\right) f(y)+\int_t^s \int_{\mathbb{R}}
\left(f(y+z)-f(y)\right)z \nu_0(r,-dz)dr.
\end{align}
%\end{lemma}
%\begin{proof}
Identities \eqref{prop:convol}-\eqref{prop:deriv} follow directly from the spatial-homogeneity of the
coefficients of $\Ac_{0}$.
%fundamental solution $p_0=p_0(t,x;s,y)$. Indeed, in light of Remark \ref{rem:A_0} we have
%\begin{align}\label{eq:space_homogeneity}
%p_0(t,x;s,y)=p(t,0;s,y-x)=p_0(t,x-y;s,0),\quad t<s,\ x,y\in\mathbb{R}.
%\end{align}
In order to prove \eqref{prop:polyn1}-\eqref{prop:polyn2}, we shall use some standard properties of the
Fourier transform. For any function $f$ in the Schwartz space we have
\begin{align}\label{prop:fourier1}
i \xi \Fc_{x}(f)=\Fc_{x}(-\partial_{x}f) , \qquad
 \Fc_{x}(x f)=-i\partial_{\xi}\Fc_{x}f,% \qquad \Fc_x \left(\int_{\mathbb{R}} f(x+z) \mf(dz)\right) (\xi)=\int_{\mathbb{R}}e^{i z \xi}\, \mf(-dz)\, \Fc_x f (\xi),
\end{align}
and for any L\'evy measure $\mf$ such that $\int_{|x|>1} |x| \, \mf(dx)<\infty$, we have
\begin{align}\label{prop:fourier2}
\Fc_x \left(\int_{\mathbb{R}}\left( f(x-z) -f(x) \right) z\, \mf(dz)\right)
(\xi)=\int_{\mathbb{R}}(e^{i z \xi}-1)z\, \mf(dz)\, \Fc_x f (\xi).
\end{align}
Thus, by \eqref{prop:fourier1} we obtain
\begin{align}
 &\Fc_{y}(y\, p_{0}(t,x;s,y))(\xi)\\
        &=-i\partial_{\xi}\Fc_{y}(p_{0}(t,x;s,y))(\xi) \\
    &= \left(x+\mv(t,s)+ \Cv(t,s)i \xi -i \partial_{\xi}\Psiv(t,s,\xi)\right)
    \Fc_{y}p_{0}(t,x;s,y)(\xi)&
        &\text{(by \eqref{eq:charact_0}-\eqref{eq:Phi_0})} \\
        &=  \left(x+\mv(t,s)+ \Cv(t,s)i \xi  + \int_t^s \int_\Rb  (e^{i z\xi}-1)\,z\, \nu_0(r,dz)dr\right) \Fc_{y}p_{0}(t,x;s,y) (\xi) &
        &\text{(by \eqref{eq:op_psi})}\\
        &=  \Fc_{y}\left(\left(x+\mv(t,s)- \Cv(t,s)\partial_{y}   \right)p_{0}(t,x;s,y)\right)(\xi)
  \\
        &\quad + \Fc_{y}\left( \int_t^s \int_\Rb \left( p_{0}(t,x;s,y-z)-p_{0}(t,x;s,y)  \right) \, \nu_0(r,dz)dr  \right)(\xi) &
        &\text{(by \eqref{prop:fourier1} and \eqref{prop:fourier2})}\\
        &=  \Fc_{y}\left(\Mc^{x}(t,s)p_{0}(t,x;s,y)\right)(\xi). &
        &\text{(by \eqref{prop:deriv}, \eqref{prop:convol} and \eqref{eq:M})}
\end{align}
The identity \eqref{prop:polyn2} arises from the same arguments and because, by the symmetry
property \eqref{prop:convol}, we have
\begin{align}\label{e22dd}
\Fc_x p_{0}(t,\cdot\,;T,y)(\xi)= \exp\left(i \xi(y -\mv(t,T)) -\frac{1}{2}\Cv(t,T)\xi^2+\Psiv(t,T,-\xi)-\int_t^T \g_0(s) d s\right).
\end{align}
As indicated in Remark \ref{rem:LPP4}, Theorem \ref{th:un_general_repres} reduces to
\cite[Thorem 3.8]{LPP4} in case of a null L\'evy measure $\nu(t,x,dz)\equiv 0$. The proof of the
\cite[Thorem 3.8]{LPP4} is based on a systematic use of symmetry properties of Gaussian densities
combined with some classical relations such as the Chapman-Kolmogorov equation and the Duhamel's
principle.  Using the same classical relations, the proof of Theorem \ref{th:un_general_repres} follows by replacing the Gaussian symmetry properties in \cite[Lemma 5.4]{LPP4} with the symmetries properties
\eqref{prop:convol}-\eqref{prop:deriv}-\eqref{prop:polyn1}-\eqref{prop:polyn2} outlined above for additive processes.
We refer to \cite[Section 5]{LPP4} for the details.
\end{proof}

\subsubsection{Fourier representation for $u_n$}
\label{subsec:fourier}
%The representation for the functions $u_n$ given in \eqref{eq:v.expand}-\eqref{eq:un} is useful
%for practical purposes only when the L\'evy-type density $p_0(t,x;T,y)$ and $u_{0}$ in
%\eqref{u0expl} are explicitly computable. In general, the approximation for the prices can be
%obtained combining standard Fourier inversion techniques (cf. \cite{lewis2001simple} and
%\cite{lipton2002}) with the $N$-th order expansion of the characteristic function
  %$$\sum_{n=0}^{N}\hat{p}_n(t,x;T,\xi)$$
%with $\hat{p}_n(t,x;T,\xi)$ given in \eqref{eq:charact_0}-\eqref{eq:un_tilde}.
Using \eqref{eq:charact_0}, \eqref{eq:u0.fourier} and \eqref{eq:un_tilde}, we obtain
\begin{align}
u_n(t,x) = \Lc^x_n(t,T) u_0(t,x)
    &=  \frac{1}{2\pi} \int_{\Rb} e^{\Phiv_0(t,T,\xi)}\big(\Lc^x_n(t,T)e^{i x \xi} \big)
            \hat{h}(-\xi) d \xi.
\end{align}
The term in parenthesis $\Lc^x_n(t,T) e^{i x \xi}$ can be computed explicitly.  However,
$\Lc^x_n(t,T)$ is, in general, an \emph{integro-differential} operator (when $X$ is a diffusion
$\Lc^x_n(t,T)$ is simply a differential operator).  Thus, for models with jumps, computing
$\Lc^x_n(t,T)e^{i x \xi}$ is a challenge.  Remarkably, we will show that there exists a
%first order
\emph{differential} operator $\hat{\Lc}^{\xi}_n(t,T)$ such that
%\begin{align}
%\Lc_n(t,T)e^{i \<\xi,x\>}
    %&= \hat{\Lc}_n(t,T) e^{i \<\xi,x\>}
%\end{align}
%where $\hat{\Lc}_n(t,T)$ is a \emph{differential operator}, which acts on $\xi$ rather than $x$.
%For computational purposes, it is sometimes convenient to express prices as an inverse Fourier
%transform. The aim of this section is to derive new pricing formulas whose implementation is more
%direct than those given in Proposition \ref{thm:dyson}. This will be done by switching from the
%state variables $x$ to the variables $\xi$ in the Fourier space. Since the operator $\Lc_{n}(t,T)$
%in \eqref{eq:Ln} acts in the variables $x$, for clarity, we shall denote it also by
%$\Lc_{n}^{x}(t,T)$ below.
%
 %%The following theorem provides and expression for $\uh_n$, the Fourier transform of $u_n$ in the price expansion \eqref{eq:ubar.N}.
%Next we define $\hat{\Lc}^{\xi}_{n}(t,T)$ as the operator acting in the variable $\xi$, such that
\begin{align}\label{e33}
  \Lc^{x}_{n}(t,T)e^{i x \xi}=\hat{\Lc}^{\xi}_{n}(t,T)e^{i x \xi} ,
\end{align}
where, for clarity, we have explicitly indicated using the superscript $\xi$ that %$\Lc^{x}_{n}(t,T)$ acts on $x$ and
$\hat{\Lc}^{\xi}_{n}(t,T)$ acts on $\xi$.  With a slight abuse of terminology, we call
$\hat{\Lc}^{\xi}_{n}$ the \emph{symbol}
\footnote{
\label{footnote}
The operator $\hat{\Lc}^{\xi}_{n}$ is not a
function as in the classical theory of pseudo-differential calculus.  However $e^{-i
\<\xi,x\>} \hat{\Lc}^{\xi}_{n} e^{i x \xi}$ is the symbol of $\Lc_n^x(t,T)$.
For the interested reader, any book on pseudo-differential operators is an appropriate resource to learn about symbols.  See, for example \cite{jacob2001pseudo} or \cite{hoh1998pseudo}.
}
of the
operator $\Lc_{n}^x(t,T)$ in \eqref{eq:def_Ln}.
%As we shall prove below, $\hat{\Lc}^{\xi}_{n}$ is a \emph{differential operator} and not a function as in the classical theory of pseudo-differential calculus: this is due to the fact that $\Lc^{x}_{n}(t,T)$ is an operator with variable (polynomial) coefficients.

Let us consider the operator $\Mc^{x}(t,s)$ in \eqref{eq:M}; its symbol $\widehat{\Mc}^{\xi}(t,s)$ is defined analogously to \eqref{e33}, i.e.
\begin{align}\label{e35}
  \Mc^{x}(t,s)e^{i x \xi}=\widehat{\Mc}^{\xi}(t,s)e^{i x \xi}.
\end{align}
Explicitly, we have
\begin{align}
 \widehat{\Mc}^{\xi}(t,s)= F(\xi,t,s)-i \partial_{\xi_{i}},%\qquad i=1,\dots,d,
\end{align}
where the function $F$ is defined as%\ora{What is $\widehat{\Mc}_{i}^{\xi}(t,t_k)$?}
\begin{align}
F(\xi,t,s)&=-i\xi\Psiv(t,s,\xi)
            +   \mv(t,s)d s + i \xi \Cv(t,s)\\
        &=   \int_{t}^{s}  \int_{\Rb} z \( e^{i z \xi} - 1 \)
            \nu_0(\tau,d z)d \tau
            +   \mv(t,s)d s + i \xi \Cv(t,s).
\end{align}
We note that, while $\Mc^{x}$ is a first order \emph{integro-differential} operator, its symbol
$\widehat{\Mc}^{\xi}$ is a first order \emph{differential} operator. For this reason, it is more
convenient to use the symbol $\widehat{\Mc}^{\xi}$ instead of the operator $\Mc^{x}$. %Note also that
%\begin{align}
%\Mc^{x}_{i}(t,t_k)\Mc^{x}_{j}(t,t_k)e^{i x \xi}=\Mc^{x}_{i}(t,t_k)\widehat{\Mc}_{j}^{\xi}(t,t_k)e^{i x \xi}
%  =\widehat{\Mc}_{j}^{\xi}(t,t_k)\Mc^{x}_{i}(t,t_k)e^{i x \xi}=\widehat{\Mc}_{j}^{\xi}(t,t_k)\widehat{\Mc}_{i}^{\xi}(t,t_k)e^{i x \xi}.
%\end{align}
%Since $\Mc^{x}_{i}$ and $\Mc^{x}_{j}$ commute when applied to a function that admits a Fourier representation, then $\widehat{\Mc}_{j}^{\xi}$ and $\widehat{\Mc}_{i}^{\xi}$ also commute when applied to such functions. In particular, the operator $\left(\widehat{\Mc}^{\xi}(t,t_{k})\right)^{\beta}$, for $\beta\in\mathbb{N}^{d}_{0}$, is well defined and we have
%\begin{align}\label{e40}
%\left(\widehat{\Mc}^{\xi}(t,t_{k})\right)^{\beta}e^{i x \xi}=
%\left(\Mc(t,t_k)\right)^{\beta}e^{i x \xi}.
%\end{align}
From identity \eqref{e35} we obtain directly the expression of the symbol of $\Gc_{j}$ in \eqref{def_Gn}.  Indeed, recalling the expression \eqref{e34} of $\nu_{j}$ we have
\begin{align}
\hat{\Gc}^{\xi}_j(t,s)&=
    %&= \Ac_j\left(t_k,\widehat{\Mc}^{\xi}(t,t_k)\right) \red{\quad \textrm{this inequality is not correct, is it? The terms do not commute..}}\\
        - (\xi^2 + i\xi)\,a_n\big(s,\widehat{\Mc}^{\xi}(t,s)\big)+\left(i\xi-1\right)\,\gam_n\big(s,\widehat{\Mc}^{\xi}(t,s)\big)\\
           &\quad +\sum_{m=1}^{M_{n}}\left( -i \xi \int_\Rb (e^z-1-z) \, \nu_{n,m}(s,dz) + \int_{\mathbb{R}}  \left( e^{i z\xi} - 1 - i z\xi \right)\nu_{n,m}(s,d z)\right)\left(\widehat{\Mc}^{\xi}(t,s)\right)^{m}%\\
%        &= \sum_{|\beta|\le M_{j}} \int_{\Rb}\left(e^{ i\<z, \xi \>} - 1 - i\<z, \xi \>\right)\nu_{j,\beta}\left(s,d z%\right)\, \left(\widehat{\Mc}^{\xi}(t,s)\right)^{\beta}
%            + \sum_{|\alpha |\leq 2} \left(i\xi\right)^{\alpha}
%            a_{\alpha,j}\left(s,\widehat{\Mc}^{\xi}(t,s)\right)
.
    \label{eq:Gh.def}
\end{align}
Thus we have proved the following lemma
 %that gives the explicit expression of
%$\hat{\Lc}_{n}(t,T)$.
\begin{lemma}
\label{lemmand}
We have
\begin{align}
 \hat{\Lc}^{\xi}_{n}(t,T)
    &=  \sum_{k=1}^n
            \int_{t}^T d t_1 \int_{t_1}^T d t_2 \cdots \int_{t_{k-1}}^T d t_k
            \sum_{i \in I_{n,k}}
            \hat{\Gc}^{\xi}_{i_1}(t,t_1)
            \hat{\Gc}^{\xi}_{i_{2}}(t,t_{2})
            \cdots \hat{\Gc}^{\xi}_{i_k}(t,t_k)  , \label{eq:Lnh}
\end{align}
with $I_{n,k}$ as defined in \eqref{eq:def_Ln_bis}.
\end{lemma}
%%{\bf\blue [Andrea: make a Lemma here for formula \eqref{e33}-\eqref{eq:Lnh}?]}
%\ora{I think we need to be slightly more explicit here.  Please see the step-by-step computations
%I provided in the proof of Theorem \ref{thm:dyson}.  I think this is the level of detail we need
%to provide.}
\noindent The following theorem extends the Fourier pricing formula \eqref{eq:u0.fourier} to higher order
approximations.
%  {\blue\bf [Andrea: I would state the following theorem in terms of characteristic functions and NOT of solutions. It would be more clear and easy to apply to compute option prices via Fourier.]}
\begin{theorem}
\label{thm:fourier}
Assume that $h,\hat{h} \in L^1(\Rb,dy)$. Then, for any $n\geq 1$ we have %$\uh_n(t,\xi)$, the Fourier transform of $u_n(t,x)$, is given by
\begin{align}
  u_n(t,x)&=\frac{1}{2 \pi }  \int_{\Rb} \hat{p}_n(t,x,T,\xi)\hat{h}(-\xi)\,d \xi, \label{eq:un_fourier}
   %&=\frac{1}{(2 \pi )^d}  \int_{\Rb^d}
            %e^{\Phi_0(t,T,\xi)}\Lc_n(t,T)e^{i \< \xi, x \>} \phih(-\xi)\,d \xi \\
\end{align}
where $\hat{p}_n(t,x,T,\xi)$ is the $n$th order term of the approximation of the characteristic
function of $X$. Explicitly, we have
\begin{align}\label{e22cc}
 \hat{p}_n(t,x,T,\xi):=\hat{p}_0(t,x,T,\xi)
    \(e^{-i x \xi} \hat{\Lc}^{\xi}_{n}(t,T)e^{i x \xi} \)
\end{align}
where $\hat{p}_0(t,x,T,\xi)$ is the $0$th order approximation in \eqref{eq:charact_0} and
$\hat{\Lc}^{\xi}_{n}(t,T)$ is the differential operator defined in \eqref{eq:Lnh}.
%
%$\chi_n(t,T,x,\xi):=\frac{}{e^{i \< \xi, x \>}}$ being the symbol of the operator $\Lc_n(t,T)$.
\end{theorem}

\begin{proof}
We first note that, since the approximating operator $\Lc^{x}_n$ acts in the $x$ variables, then
it commutes\footnote{This was one of the main points of the {\it adjoint expansion method}
proposed by \cite{RigaPagliaraniPascucci}.} with the Fourier pricing operator \eqref{eq:u0.fourier}. Thus, by
\eqref{eq:un_tilde} combined with \eqref{eq:u0.fourier}, we get
\begin{align}
u_n(t,x)=\Lc^{x}_n(t,T)u_0(t,x)&=\frac{1}{2 \pi }  \int_{\Rb}\Lc^{x}_n(t,T)
            e^{i x \xi  + \Phiv_0(t,T,\xi)} \hat{h}(-\xi)\, d \xi\\
            &=\frac{1}{2 \pi }  \int_{\Rb}
            \hat{p}_0(t,x,T,\xi)
                        \(e^{-i x \xi} \Lc^{x}_{n}(t,T)e^{i x \xi} \)
                        \hat{h}(-\xi)\, d
            \xi,
%            \intertext{($\Lc_n(t,T)$ only acts on the $x$ variable)}
%         & =\frac{1}{(2 \pi )^d}  \int_{\Rb^d}
%            e^{ \Phi_0(t,T,\xi)} \phih(-\xi)\Lc_n(t,T)e^{i \< \xi, x \> }\, d \xi  .
\end{align}
and the thesis follows from \eqref{e33}. %since $\Lc_n(t,T)$ only acts on the $x$-variables, thus commuting with the
%Fourier transform operator.
\end{proof}
\begin{remark}
%Note that the symbol $\chi_n$ in the Fourier representation \eqref{eq:un_fourier} for the $u_n$ is
%actually fully explicit.
Computing the term in parenthesis above $\(e^{-i x \xi} \hat{\Lc}^{\xi}_{n}(t,T)e^{i
x \xi} \)$ is a straightforward exercise since the symbol $\hat{\Lc}^{\xi}_{n}(t,T)$, given
in \eqref{eq:Lnh}, is a differential operator.
% As a matter of example, we report here its
%expression for $n=1$ under the Taylor expansion for the coefficients proposed in
%Example \ref{example:Taylor}: % (i.e. $B_1(x)=x-\bar{x}$):
%\begin{align}
%  \frac{\Lc_1(t,T)e^{i \<\xi, x \>}}{e^{i \<\xi,x \>}}  =\sum_{j=1}^d  \int_t^T \bar{\Ac}^j_1(s,\xi)\left(  x_j -\bar{x}_j +
%\int_{t}^{s}   m_{j}(r)d r +i\sum_{k=1}^d \xi_k  \int_{t}^{s} C_{j,k}(r) d r - i
%\partial_{\xi_j}\Psi_{0}(t,s,\xi)
%   \right) ds  ,
%   \end{align}
%with $m,C$ and $\Psi_{0}$ as in \eqref{eq:m}-\eqref{eq:C}-\eqref{eq:Psi0},
%%  $$\Psi_{0}(t,T,\xi)=\int_{t}^{T} \int_{\Rb_0^d}\nu_0(s,d z)d s\, \( e^{i \< \xi, z \>} - 1 - i \< \xi, z \> \),$$
%and where
%\begin{align}
%%\Ac_n(t,x)&=
%\bar{\Ac}^j_1(s,\xi)= \sum_{|\alpha |\leq 2}  \partial_{x_j} a_{\alpha}(s,\bar{x}) (i\xi)^{\alpha} +\int_{\Rb_0^d} \left( e^{i \<\xi,z \>} - 1 - i\<\xi,z \>\right) \partial_{x_j} h_{\alpha}(s,\bar{x}) \bar{\nu}(s,d z).
%\end{align}
\end{remark}
\begin{example}
Let $(\Ac_0,\Ac_1)$ the $1$-st order Taylor expansion of $\Ac$ proposed in
Example \ref{example:Taylor}. Then we have % (i.e. $B_1(x)=x-\bar{x}$):
\begin{align}
 \hat{p}_1(t,x;T,\xi)  =\hat{p}_0(t,x;T,\xi)  \int_t^T \bar{\Ac}_1(s,\xi)\big(  x-\bar{x}+\mv(t,s)+i\xi \Cv(t,s)-i \partial_{\xi}\Psiv(t,s,\xi)%+\int_t^s \int_{\mathbb{R}} \left(e^{iz\xi}-1\right)z\, \nu_0(r,dz)dr
   \big) ds  ,%\quad t<T,\ x,\xi\in\mathbb{R},
\end{align}
with
\begin{align}
\bar{\Ac}_1(s,\xi) &=  \gam_1(s) ( i\xi - 1 ) + a_1(s) ( -\xi^2 - i\xi ) - i\xi \int_\Rb (e^z-1-z)
\, \nu_1(s,dz) + \int_\Rb \left(  e^{iz\xi} -1 - iz\xi  \right)\,h_1(s,\bar{x},z)\bar{\nu}(s,dz) \label{eq:A_1_bar},
\end{align}
and
\begin{equation}
\gamma_1(s)=\partial_x \gamma(s,\bar{x}),\qquad a(s)=\partial_x a(s,\bar{x}),\qquad h_1(s,\bar{x},z)=\partial_x h(s,\bar{x},z).
\end{equation}
\end{example}
\begin{remark}
If $h(y) \notin L^1(\Rb,dy)$ but $h(y)e^{cy}  \in L^1(\Rb,dy)$ for some $c \in \Rb$ (which is the case for Call and Put payoffs), one can still use expressions \eqref{eq:u0.fourier} and \eqref{eq:un_fourier} by fixing an imaginary component of $\xi$.  This technique, known as a \emph{generalized Fourier transform}, is described in detail in \cite{lewis2000} and \cite{lipton2002}.
%Thus, the price of a European call option can be computed using standard Fourier methods
%\begin{align}
%u(t,x,y)
%    &=  \frac{1}{2\pi} \int_\Rb d \xi_r \, \phih(-\xi) \Eb_{x,y} e^{ i \xi X_{T-t}} , &
%\phih(\xi)
%    &=  \frac{-e^{k-i k \xi}}{ i \xi + \xi^2 } , &
%\xi
%    &=  \xi_r + i \xi_i , &
%\xi_i
%    &<  -1 . \label{eq:u.Heston}
%\end{align}
%Note that,
\end{remark}

%%%%%%%%%%%%%%%%%%%%%%%%%%%%%%%%%%%%%%%%%%%%%%%%%%%%%%%%%
%
%                               Error bounds
%
%%%%%%%%%%%%%%%%%%%%%%%%%%%%%%%%%%%%%%%%%%%%%%%%%%%%%%%%%

\section{Gaussian jumps: explicit densities and pointwise error bounds}\label{sec:errors}
\label{errbou} We examine here the particular case when the L\'evy measure $\nu$ coincides with a
normal distribution with state dependent parameters. Specifically, throughout this section we
will assume
\begin{align}\label{eqbound3}
\nu(t,x,dz) &=\lam(t,x)\, \mathscr{N}_{m(x),\delta^2(x)}(dz)
 :=\frac{\lam(t,x)}{\sqrt{2\pi}\delta(x)}e^{-\frac{(z-m(x))^2}{2\delta^2(x)}}dz.
\end{align}
We will show that, under such a choice, the representation formula given in Theorem
\ref{th:un_general_repres} leads to closed form (fully explicit) approximations for densities,
prices and Greeks. Furthermore we will prove some sharp pointwise error bounds for such
approximations at a given order $N\in \mathbb{N}_0$.

For sake of simplicity, we will work specifically with the Taylor series expansion of Example
\ref{example:Taylor}.
%\begin{align}
%\d_x^n \nu(t,\xb,dz)=g_n(t)\nu(t,\xb,dz),\quad n\geq 0,
%\end{align}
%where
%\begin{align}
%g_n(t)=\partial^n_x\left( \frac{\lam(t,\bar{x})}{\sqrt{2\pi}\delta(\bar{x})}e^{-\frac{(z-m(\bar{x}))^2}{2\delta^2(\bar{x})}}\right)/\left(   \frac{\lam(t,\bar{x})}{\sqrt{2\pi}\delta(\bar{x})}e^{-\frac{(z-m(\bar{x}))^2}{2\delta^2(\bar{x})}}   \right).
%\end{align}
Throughout this section we will often make use of the convolution operator
\begin{align}\label{eq:op_conv_norm}
 \mathscr{C}_{\rho,\theta}f(x):=\mathscr{C}^x_{\rho,\theta}f(x)=\int_{\mathbb{R}}f(x+z)\frac{1}{\sqrt{2\pi \theta}}e^{-\frac{(z-\rho)^2}{2\theta}}dz,\qquad \rho\in\mathbb{R},\quad \theta>0.
\end{align}
Let us first observe that %, in light of Remark \ref{rem:A_0},
the leading term $p_0(t,x;T,y)$ in the expansion of the fundamental solution $p(t,x;T,y)$ is the
transition density of a time-dependent compound Poisson process with L\'evy measure
\begin{align}
\nu_0(t,dz) &=\lam_0(t)\, \mathscr{N}_{m_0,\delta_0^2}(dz)
 :=\frac{\lam_0(t)}{\sqrt{2\pi}\delta_0}e^{-\frac{(z-m_0)^2}{2\delta_0^2}}dz,
\end{align}
and thus it can be written as %{\bf\blue There is $a$ appearing in the following formula: is it correct? See also Remark \ref{remand2}.}
\begin{align}\label{eq:p_0_gaus}
p_0(t,x;T,y)&=e^{-\int_t^T (\lam_0(s)+\g_0(s)) d s}\sum_{n=0}^{\infty}\frac{\left(\int_t^T \lam_0(s) ds \right)^n}{n!} p_{0,n}(t,x;T,y) \\ \label{eq:p_0_gaus_1}
p_{0,n}(t,x;T,y)&=
\frac{1}{\sqrt{2\pi} \left(\int_t^T a_0(s) ds+n\, \delta_0^2 \right)^{\frac{1}{2}}}\exp\left(-\frac{\left(x-y+n\, m_0-
  \int_t^T \left( \frac{a_0(s)}{2}+\lam_0(s)  e^{\frac{\delta_0^2}{2}}-\lam_0(s) \right) ds \right)^2}
  {2\left(\int_t^T a_0(s) ds+n\, \delta_0^2 \right)}\right).
\end{align}
This also implies that the leading term $u_0(t,x)$ in the price expansion is explicit, as long as
the integrals of the payoff function $h$ against the Gaussian densities $p_{0,n}(t,x;T,\cdot)$ are
computable in closed form.

Moreover we have the following representation for the operators $(\Gc_n^x)_{n\geq 1}$ appearing in
Theorem \ref{th:un_general_repres}.
\begin{proposition}\label{prop:G_n_gaus}
For any $n\geq 1$, the operator $\Gc^x_{n}$ in \eqref{def_Gn} is given by
\begin{align}
\Gc^x_{n}(t,s) =  \left(\Mc^x(t,s)-\bar{x} \right)^n \Ac_n(s),
\end{align}
where
\begin{align} \label{eq:M_gaus}
\Mc^x(t,s)f(x)&= x+\int_t^s \left( \gam_0(r) - a_0(r)
       - \lam_0(r)\left(e^{\frac{\delta_0^2 }{2}+m_0 }-1\right)\right)dr+2\int_t^T  a_0(r)dr\, \partial_x \\
      &\quad +\int_t^s \lambda_0(r)dr\, \left(m_0- \delta^2_0\partial_x\right)  \mathscr{C}^x_{m_0,\delta_0^2},
\end{align}
and
\begin{align}
 &\Ac_n(s) =  a_n(s) ( \d_{xx}  - \d_{x} )+\gam_n(s) ( \d_{x} - 1 ) -
 g_n(s,\partial_x)\left(e^{\frac{\delta_0^2 }{2}+m_0 }-1\right)\,  \d_{x} +
 g_n(s,\partial_x)(\mathscr{C}^x_{m_0,\delta_0^2} -1), \label{eq:A_n_gaus}\\
 \label{eq:taylor_coef_gaus}
 &a_n(s)=  \frac{1}{n!}\d_x^n a(s,\xb),\quad
 \gam_n(s)=  \frac{1}{n!}\d_x^n \gam(s,\xb), %\quad \nu_n(s,dz)=  \frac{1}{n!}\d_x^n \nu(s,\xb,dz),
\end{align}
with %$\mathscr{C}^x$ is defined as in \eqref{eq:op_conv_norm}, and
$(g_n(s,\cdot))_{n\geq 0}$ being polynomials whose coefficients only depend on
%$(\lambda_k(s),m_k,\delta_k)_{0\leq k \leq n}$.
\begin{align}\label{eq:taylor_coef_lev}
 \lam_i(t):= \frac{1}{i!}\d_x^i \lam(t,\xb),\quad m_i:= \frac{1}{i!}\d_x^i m(\xb),\quad
 \delta_i:= \frac{1}{i!}\d_x^i \delta(\xb),\qquad 0\leq i\leq n.
\end{align}
\end{proposition}
\begin{remark}\label{remand2}
%The L\'evy measure %$\d_x^n \nu(t,\xb,dz)$
%$\nu_{n}$ depends in \eqref{} on
Note that the action of the operators $\Gc^x_{n}$ on the L\'evy
type density $p_0(t,x;T,y)$, as well as on $u(t,x)$, can be explicitly characterized. Indeed, a
direct computation shows that, for any $k\geq 0$,
\begin{align}
\partial_x p_{0,k}(t,x;T,y)&=-\frac{x-y+n\, m_0-
  \int_t^T \left( \frac{a_0(s)}{2}+\lam_0(s)  e^{\frac{\delta_0^2}{2}}-\lam_0(s) \right) ds }
  {2\left(\int_t^T a_0(s) ds+n\, \delta_0^2 \right)} p_{0,k}(t,x;T,y),\\
  \mathscr{C}^x_{m_0,\delta_0^2}\, p_{0,k}(t,x,;T,y)&= p_{0,k+1}(t,x;T,y),
\end{align}
and
\begin{align}
   \mathscr{C}^x_{m_0,\delta_0^2}\, \left(x\, p_{0,k}(t,x,;T,y)\right)&=(x+m_0-\delta_0^2\partial_x) \mathscr{C}^x_{m_0,\delta_0^2}\, p_{0,k}(t,x,;T,y),\\
    \mathscr{C}^x_{m_0,\delta_0^2}\, \left(\partial_x p_{0,k}(t,x,;T,y)\right)&=\partial_x\mathscr{C}^x_{m_0,\delta_0^2}\, p_{0,k}(t,x,;T,y).
\end{align}
\end{remark}
We now fix $N\geq 0$ and prove some pointwise error estimates for the $N$-th order
approximation of the fundamental solution of $p(t,x;T,y)$, defined as
\begin{align}
p^{(N)}(t,x;T,y)=\sum_{n=0}^N p_n (t,x;T,y),
\end{align}
where the functions $p_n(\cdot,\cdot\,;T,y)$ solve \eqref{eq:v.0.pide}-\eqref{eq:v.n.pide} with $h=\delta_y$.
Hereafter, we will assume the coefficients of the operator $\Ac$ in \eqref{eq:A}, with $\nu$ as in \eqref{eqbound3}, to satisfy the following assumption.
\begin{assumption}\label{assumption:parab}
There exists a constant $M>0$ such that
\begin{enumerate}
\item[i)]{\it(parabolicity)} for any $t\in[0,T]$ and $x\in\mathbb{R}$,
\begin{align}
M^{-1}\leq a(t,x)\leq M;
\end{align}
\item[ii)]{\it(non degeneracy of the L\'evy measure)} {the L\'evy measure $\nu$ is as in \eqref{eqbound3} and,} for any $t\in[0,T]$ and $x\in\mathbb{R}$,
\begin{align}
 M^{-1}\leq  \delta^2(x)\leq M,\qquad 0\leq \lambda(t,x)\leq M, \qquad t\in[0,T],\ x\in\R;
\end{align}
\item[iii)]{\it(regularity and boundedness)} for any $t\in[0,T]$, the functions $a(t,\cdot),\g(t,\cdot)$, $\lambda(t,\cdot)$, $\delta(\cdot)$,
$m(\cdot)\in C^{N+1}(\mathbb{R})$, and all of their $x$-derivatives up to order $N+1$ are bounded by $M$, uniformly with respect to $t\in[0,T]$.
\end{enumerate}
\end{assumption}
\begin{theorem}\label{t1}
%Assume that
%\begin{align}
%m\leq a(t,x)\leq M,\qquad 0\leq \gamma(t,x)\leq M, \qquad t\in[0,T],\ x\in\R,
%\end{align}
%for some positive constants $m$ and $M$, and that
%\begin{align}\label{eqbound3}
%\nu(t,x,dz) &=\lam(t,x)\, \mathscr{N}_{\m(t,x),\delta^2(t,x)}(dz)
% :=\frac{\lam(t,x)}{\sqrt{2\pi}\delta(t,x)}e^{-\frac{(z-\m(t,x))^2}{2\delta^2(t,x)}}dz,
%\end{align}
%with
%\begin{align}
% m\leq  \delta^2(t,x)\leq M,\qquad 0\leq \lambda(t,x),|\mu(t,x)|\leq M, \qquad t\in[0,T],\ x\in\R.
%\end{align}
%Moreover assume that $a,\g,\lambda,\delta,\mu$ and their $x$-derivatives are bounded and Lipschitz
%continuous in $x$, and uniformly bounded with respect to $t\in[0,T]$.
Let {$N\in \mathbb{N}_0$}, and $\xb=y$ or $\xb=x$ in \eqref{eq:taylor_coef_gaus}-\eqref{eq:taylor_coef_lev}.
Then, {under Assumption \ref{assumption:parab}}, for any $x,y \in \Rb$ and $t< T$ we have\footnote{Here $\left\|  \partial_x \nu  \right\|_{\infty}:=\max\{
\left\|
\partial_x\lambda  \right\|_{\infty},\left\|  \partial_x\delta  \right\|_{\infty},\left\|
\partial_x \mu  \right\|_{\infty} \}$, where $\|  \cdot  \|_{\infty}$ denotes the sup-norm on
$(0,T)\times\R$. Note that $\left\|\partial_{x}\nu\right\|_{\infty}=0$ if $\lambda,\delta,\mu$ are
constants.}
\begin{align}\label{eqbound10}
 \left|p(t,x;T,y) - p^{(N)}(t,x;T,y)\right|\leq   g_{N}(T-t) \left(\bar{\G}(t,x;T,y)+\left\|  \partial_x \nu  \right\|_{\infty} \widetilde{\G}(t,x;T,y) \right),
 %\qquad\ x,y \in \Rb,\ 0<t\leq T,
\end{align}
where
\begin{align}
{ g_{N}(s)=\Oc \left(s^{\frac{1+\min(1,N)}{2}}\right)},\quad \textrm{as } s \to 0^+ .
\end{align} Here, the function $\bar{\G}$ is the fundamental solution of the
constant coefficients jump-diffusion operator
\begin{align}
  &\p_t u(t,x)+\frac{\bar{M}}{2}\p_{xx} +\bar{M}\int_{\R} \left( u(t,x+z)-u(t,x)\right) \mathscr{N}_{\bar{M},\bar{M}}(dz),
\end{align}
where $\bar{M}$ is a suitably large constant, and $\widetilde{\G}$ is defined as
\begin{align}
\widetilde{\G}(t,x;T,y)=\sum_{k=0}^{\infty}%\left(
\frac{\bar{M}^{k/2}(T-t)^{k/2}}{\sqrt{k!}}%\right)^{\frac{1}{2}}
\mathscr{C}^{k+1}\bar{\G}(t,x;T,y),
\end{align}
with $\mathscr{C}_{\bar{M}}=\mathscr{C}^x_{0,\bar{M}}$ being the convolution operator defined in
\eqref{eq:op_conv_norm}.
%\begin{align}
%\mathscr{C}f(x)=\int_{\R}f(x+z)\mathscr{N}_{\bar{M},\bar{M}}(dz).
%\end{align}
\end{theorem}
\noindent
The proof the Theorem \ref{t1} is postponed to Section \ref{sec:proof}.
\begin{remark}\label{remt1}
As we shall see in the proof of Theorem \ref{t1}, the functions $\mathscr{C}^k \bar{\G}$ take the
following form
\begin{align}\label{andconv}
\mathscr{C}^{k}\bar{\G}(t,x;T,y)&=e^{-\bar{M}(T-t)}\sum_{n=0}^{\infty}\frac{\left(\bar{M}
(T-t)\right)^n}{n!\sqrt{2\pi \bar{M} (T-t+n+k)}}\, \exp\left(-\frac{\left(x-y
+\bar{M}(n+k)\right)^2}  {2\bar{M}(T-t+n+k)}\right),\qquad  k\ge 0,
\end{align}
and therefore $\widetilde{\G}$ can be explicitly written as
\begin{align}
\widetilde{\G}(t,x;T,y)=e^{-\bar{M}(T-t)}\sum_{n,k=0}^{\infty} \frac{\left(\bar{M}
(T-t)\right)^{n+\frac{k}{2}} }{n!\sqrt{k!}\sqrt{2\pi \bar{M} (T-t+n+k+1)}}  \,
\exp\left(-\frac{\left(x-y +\bar{M}(n+k+1)\right)^2}
  {2\bar{M}(T-t+n+k+1)}\right).
\end{align}
\end{remark}
By Remark \ref{remt1}, it follows that, when $k=0$ and $x\neq y$, the asymptotic behaviour as
$t\to T$ of the sum in \eqref{andconv} depends only on the $n=1$ term.  Consequently, we have
$\bar{\G}(t,x;T,y)=\Oc (T-t)$ as $(T-t)$ tends to $0$. On the other hand, for $k\ge 1$,
$\mathscr{C}^{k}\bar{\G}(t,x;T,y)$, and thus also $\widetilde{\G}(t,x;T,y)$, tends to a positive
constant as $(T-t)$ goes to $0$. It is then clear by \eqref{eqbound10} that, with $x\neq y$ fixed,
the asymptotic behavior of the error, when $t$ tends to $T$, changes from
{$(T-t)^{\frac{1+\min(1,N)}{2}}$ to $(T-t)^{\frac{1+\min(1,N)}{2}+1}$} depending on whether
the L\'evy measure is locally-dependent or not.
%\blu{(I do not
%understand the following sentence.  We should discuss this via skype at some
%point.)}\red{(Basically, fixing $x\neq y$, when $t$ is close to $T$ the error goes as $(T-t)$ in
%the general case, but we get more accuracy, i.e. the error goes as $(T-t)^2$, as soon as we
%consider a non-local L\'evy measure. If you still have any doubt we can discuss it via Skype.)}
\begin{remark}
The proof of Theorem \ref{t1} is also interesting for theoretical purposes. Indeed, it actually
represents a procedure to construct $p(t,x;T,y)$. Note that with $p^{(N)}(t,x;T,y)$ being known
explicitly, equation \eqref{eqbound10} provides pointwise {upper bounds} for the fundamental solution of the integro-differential operator with variable coefficients $(\partial_t +\Ac)$.
\end{remark}

Theorem \ref{t1} extends the previous results in \cite{RigaPagliaraniPascucci} where only the
purely diffusive case (i.e $\lambda \equiv 0$) is considered. In that case an estimate analogous to
\eqref{eqbound10} holds with
 $$g_{N}(s)=\Oc \left(s^{\frac{N+1}{2}}\right),\quad \textrm{as } s \to 0^+.$$
Theorem \ref{t1} shows that for jump processes, {one obtains an improvement on the asymptotic
convergence from $(T-t)^{\frac{1}{2}}$ to $(T-t)$} when passing from $N=0$ to $N=1$. On the other
hand, increasing the order of the expansion for $N$ greater than one, theoretically does not give
any gain in the rate of convergence of the {approximation} expansion as $t \to T^-$; this is
due to the fact that the expansion is based on a local (Taylor) approximation while the PIDE
contains a non-local part. {We refer to Section \ref{sec:difference_diffusion} for further
details about this aspect. As for the estimate \eqref{eqbound10}, this is} in accord with the
results in \cite{gobet-smart} where only the case of constant L\'evy measure is considered. Thus
Theorem \ref{t1} extends the latter results to state dependent Gaussian jumps using a completely
different technique. Extensive numerical tests showed that the first order approximation gives
very accurate results and the precision appears to be further improved by considering higher order
approximations.

A straightforward corollary of Theorem \ref{t1} is the following estimate of the error for the $N$-th order approximation of the price, defined as
\begin{align}
u^{(N)}(t,x)=\sum_{n=0}^N u_n (t,x), \label{eq:uN}
\end{align}
where the functions $u_n(\cdot,\cdot\,;T,y)$ solve \eqref{eq:v.0.pide}-\eqref{eq:v.n.pide}.
\begin{corollary}\label{cor2}
Let $\xb=y$ or $\xb=x$ in \eqref{eq:taylor_coef_gaus}-\eqref{eq:taylor_coef_lev}.
Then, for any $x,y \in \Rb$ and $t< T$ we have
\begin{align}
 \left |u(t,x) - u^{(N)}(t,x)\right|\leq   g_{N}(T-t)  \int_{\mathbb{R}} |h(y)| \left(\bar{\G}(t,x;T,y)+\left\|  \partial_x \nu  \right\|_{\infty} \widetilde{\G}(t,x;T,y) \right) dy.
\end{align}
\end{corollary}
Some possible extensions of these asymptotic error bounds to general L\'evy measures are possible,
though they are certainly not straightforward. Indeed, the proof of Theorem \ref{t1} is based on
some pointwise uniform estimates for the fundamental solution of the constant coefficient
operator, i.e., the transition density of a compound Poisson process with Gaussian jumps. When
considering other L\'evy measures these estimates would be difficult to carry out, especially in
the case of jumps with infinite activity, but they might be obtained in some suitable normed
functional space. This might lead to error bounds for short maturities, which are expressed in
terms of a suitable norm, as opposed to uniform pointwise bounds. {We aim to elaborate more on
this direction in our future research.}
%\begin{remark}
%\label{rmk:practice}
%Since, in general, it is hard to derive the truncation error bound, the reader may wonder how to determine the number of terms to  include in the asymptotic expansion.  Though we provide a general expression for the $n$-th term, realistically, only the fourth order term can be computed.  That said, as we shall see from the examples in Section \ref{sec:examples}, in practice, three terms provide an approximation which is accurate enough for most applications (i.e., the resulting approximation error is smaller than the bid-ask spread typically quoted on the market).  Since, $v^{(n)}$ only requires only a single Fourier integration, there is no numerical advantage for choosing smaller $n$.  As such, for financial applications we suggest using $n=3$ or $n=4$.
%\end{remark}

%%%%%%%%%%%%%%%%%%%%%%%%%%%%%%%%%%%%%%%%%%%%%%%%%%%%%%%%%
%
%                    Examples
%
%%%%%%%%%%%%%%%%%%%%%%%%%%%%%%%%%%%%%%%%%%%%%%%%%%%%%%%%%

\section{Examples}
\label{sec:examples}
In this section, in order to illustrate the versatility of our asymptotic
expansion, we apply our approximation technique to a variety of different L\'evy-type models. We
study not only option prices and transition densities, but also implied volatilities and credit
spreads.  In each setting, if the exact or approximate density/option price/credit spread has been
computed by a method other than our own, we compare this to the density/option price/credit spread
obtained by our approximation.  For cases where the exact or approximate density/option
price/credit spread is not analytically available, we use Monte Carlo methods to verify the
accuracy of our method.
\par
Note that, some of the examples considered below do not satisfy the conditions listed in Section
\ref{sec:model}.   In particular, we will consider coefficients $(a,\gam,\nu)$ that are not
bounded.  Nevertheless, the formal results of Section \ref{sec:formal} work well in the examples
considered.

%%%%%%%%%%%%%%%%%%%%%%%%%%%%%%%%%%%%%%%%%%%%%%%%%%%%%%%%%
%           Compairson CEV-like L\'evy-type processes
%%%%%%%%%%%%%%%%%%%%%%%%%%%%%%%%%%%%%%%%%%%%%%%%%%%%%%%%%

\subsection{CEV-like L\'evy-type processes}
\label{sec:CEV}
We consider a L\'evy-type process of the form \eqref{eq:dX} with CEV-like
volatility and jump-intensity.  Specifically, the $\log$-price dynamics are given by
\begin{align}
%\sig(t,x)
    %&= \sig_0 e^{(\beta-1)x}, &
a(x)
    &=  \frac{1}{2} \del^2 e^{2(\beta-1)x}, &
\nu(x,dz)
    &=  e^{2(\beta-1)x} \Nc(dz), &
\gam(x)
    &=  0, &
\del
    &\geq   0, &
\beta
    &\in [0,1], \label{eq:CEV.like}
\end{align}
where $\Nc(dx)$ is a L\'evy measure.  When $\Nc \equiv 0$, this model reduces to the CEV model of
\cite{CoxCEV}.  Note that, with $\beta \in [0,1)$, the volatility and jump-intensity increase as
$x \to -\infty$, which is consistent with the leverage effect (i.e., a decrease in the
value of the underlying is often accompanied by an increase in volatility/jump intensity).  This
characterization will yield a negative skew in the induced implied volatility surface.
%As the
%class of models described by \eqref{eq:CEV.like} is of the form \eqref{eq:proportional0} with
%$f(x)=e^{2(\beta-1)x}$, this class naturally lends itself to the two-point Taylor series
%approximation of Example \ref{ex:two-pt}.  Thus, for certain numerical examples in this Section,
%we use basis functions $B_n$ given by \eqref{eq:B.two-pt}.  In this case we choose expansion
%points $\xb_1$ and $\xb_2$ in a symmetric interval around $X_0$ and in \eqref{eq:M} we choose $M =
%f(X_0) =e^{2(\beta-1)X_0}$.
For the numerical examples for this model, we use the one-point Taylor
series expansion of $\Ac$ as in Example \ref{example:Taylor} with $\xb = X_t$.
\par
We will consider the case where the L\'evy measure $\Nc(dz)$ is Gaussian:
\begin{align}
%\text{Gaussian:}&&
\Nc(dz)
    &=  \lam \frac{1}{\sqrt{2\pi \eta^2}}\exp\( \frac{-(z-m)^2}{2 \eta^2} \) dz .
            \label{eq:eta.gaussian}
%\text{Variance-Gamma:}&&
%\Nc(dz)
    %&=  \( \frac{e^{-\lam_{-} |z|}}{\kappa|z|} \Ib_{\{z<0\}} + \frac{e^{-\lam_{+} z}}{\kappa z} \Ib_{\{z>0\}} \) dz,
            %\label{eq:eta.VG} \\
%&&
%\lam_{\pm}
    %&=  \( \sqrt{\frac{\theta^2 \kappa^2}{4} + \frac{\rho^2 \kappa}{2}} \pm \frac{\theta \kappa}{2}\)^{-1}
\end{align}
%\label{sec:Merton.CEV}
In our first numerical experiment, we consider the case of Gaussian jumps.
That is, $\Nc(dz)$ is given by \eqref{eq:eta.gaussian}.  We fix the following parameters
\begin{align}
\del
    &= 0.20, &
\beta
    &= 0.5, &
\lam
    &= 0.3, &
m
    &= -0.1, &
\eta
    &=  0.4, &
S_0 = e^x
    &=  1. \label{eq:parameters}
\end{align}
In order to examine the convergence of our density approximation, in Figure \ref{fig:density} we plot the
approximate transition density $p^{(n)}(t,x;T,y)$ for different values of $n$.  We note that, for $T-t\leq 5$, the transition
densities $p^{(4)}(t,x;T,y)$ and $p^{(3)}(t,x;T,y)$ are nearly identical. This is typical in our
numerical experiments.  Numerical results associated with Figure \ref{fig:density} are given in Table \ref{tab:density}.
\par
Computation times are also an important consideration.  From Theorem \ref{thm:fourier} and \eqref{eq:uN}, we observe that
\begin{align}
u^{(n)}(t,x)
    &=  \frac{1}{2 \pi }  \int_{\Rb}  \hat{p}_0(t,x,T,\xi)  \hat{h}(-\xi)
            \( 1 + \sum_{m=1}^n e^{-i x \xi} \hat{\Lc}^{\xi}_{m}(t,T)e^{i x \xi} \) \,d \xi , \label{eq:fourier.integral}
\end{align}
where, to obtain $p^{(n)}(t,x;T,y)$ from $u^{(n)}(t,x)$, one simply sets $h = \del_y$.  Thus, the $n$-th order approximation (either for an option price $u^{(n)}$ or the transition density $p^{(n)}$) as a single Fourier integral, which must be computed numerically.  The difference in computation times for a given order of approximation will depend \emph{only} on the factor in parenthesis, which is simply a polynomial in $\xi$ and can always be computed explicitly.  To gauge the numerical cost of computing the $n$th order approximation of the transition density, we measure the average time needed to compute $p^{(n)}(t,x;T,y)$ over a range of $y$-values.  We call the average time it takes to compute $p^{(n)}$ divided by the  average time it takes to compute $p^{(0)}$ the \emph{computation time of $p^{(n)}$ relative to $p^{(0)}$}.  Computation times relative to $p^{(0)}$ are given in Table \ref{tab:density}.

\subsection{Comparison with \cite{lorig-jacquier-1}}
\label{sec:lorig} In \citet{lorig-jacquier-1}, the author considers a class of time-homogeneous
L\'evy-type processes of the form:
\begin{align}
\left. \begin{aligned}
a(x)
    &=  \frac{1}{2} \left( b_0^2 + \eps b_1^2 \eta(x) \right), \\
\gam(x)
    &=  c_0 + \eps c_1 \eta(x), \\
\nu(x,dz)
    &=      \nu_0(dz) + \eps \eta(x) \nu_1(dz).
\end{aligned} \right\} \label{eq:lorig}
\end{align}
Here, $(b_0,b_1,c_0,c_1,\eps)$ are non-negative constants, the function $\eta \geq 0$ is smooth
and $\nu_0$ and $\nu_1$ are L\'evy measures.  When $\eta(x) = e_\beta(x) := e^{\beta x}$, the
authors obtain the following expression for European-style options written on $X$
\begin{align}
u(t,x)
    &=      \sum_{n=0}^\infty \eps^n w_n(T-t,x), \label{eq:w} \\
w_n(t,x)
    &=      e_{n\beta}(x) \int_\Rb d\xi \( \sum_{k=0}^n \frac{e^{t \pi_{\xi-ik\beta}}}
                {\prod_{j\neq k}^n (\pi_{\xi-ik\beta}-\pi_{\xi-ij\beta})}\)
                \( \prod_{k=0}^{n-1} \chi_{\xi-ik\beta}\) \hat{h}(\x) e^{i xi x}.
\end{align}
where $x=X_t$
%is the time to maturity, $x$ is the present value of $X$
and
\begin{align}
\pi_\xi
    &=      \frac{1}{2} b_0^2 \( -\xi^2 - i \xi \)
                + c_0 (i\xi - 1)
                - \int_\Rb \nu_0(dz) \left( e^{z} - 1 - z \right) i \xi
                + \int_\Rb \nu_0(dz) \left( e^{ i \xi z} - 1 - i \xi z \right), \\
\chi_\xi
    &=      \frac{1}{2} b_1^2 \( -\xi^2 - i \xi \)
                + c_1 (i\xi - 1)
                - \int_\Rb \nu_1(dz) \left( e^{z} - 1 - z \right) i \xi
                + \int_\Rb \nu_1(dz) \left( e^{ i \xi z} - 1 - i \xi z \right).
\end{align}
As in \eqref{eq:u0.fourier}, $\hat{h}(\x)$ is the (possibly generalized) inverse Fourier transform of the
option payoff $h(x)$.
%\par
%Now consider the following model:
%\begin{align}
%a(x)
    %&=  A f(x), &
%\gam(x)
    %&=  \Gamma f(x) , &
%\nu(x,dz)
    %&=  f(x) \Nc(dz) &
%f(x)
        %&=  a_0 + \eps a_1 \eta(x) , \label{eq:proportional}
%\end{align}
%The models described by \eqref{eq:lorig} and \eqref{eq:proportional} coincide if we choose
%\begin{align}
%\frac{1}{2}b_i^2
    %&=  a_i A , &
%c_i
    %&=  a_i  \Gamma , &
%\nu_i(dz)
    %&=  a_i \Nc(dz) , &
%i
    %&=  \{0,1\} .
%\end{align}
%Furthermore, comparing equations \eqref{eq:proportional0} with \eqref{eq:proportional}, we see that \eqref{eq:proportional} is precisely the form considered in Example \ref{ex:two-pt}.  Thus, in this Section we use the two-point Taylor series approximation of Example \ref{ex:two-pt} with basis functions $B_n$ given by \eqref{eq:B.two-pt}.  We choose expansion points $\xb_1$ and $\xb_2$ in a symmetric interval around $X_0$ and in \eqref{eq:M} we choose $M = f(X_0) =e^{\beta X_0}$.
\par
In our numerical experiment, we use the  Taylor
series expansion of $\Ac$ as in Example \ref{example:Taylor} with $\xb = X_t$.
We consider Gaussian jumps (i.e., $\Nc$ given by \eqref{eq:eta.gaussian}) and we fix the following parameters:
\begin{align}
\left.
\begin{aligned}
\beta
    &=  -2.0 , &
b_i
    &=  0.15 , &
c_i
    &=  0.0 , &
\nu_i
        &=  \Nc , &
i
        &=  \{0,1\} , \\
\eps
    &=  1.0 , &
\lam
    &=  s = 0.2 , &
m
    &=  -0.2 , &
T-t
    &=  0.5 , &
X_t
    &=  0.0 ,
\end{aligned} \right\} \label{eq:params}
\end{align}
where the L\'evy measure $\Nc$ is given by \eqref{eq:eta.gaussian}.
Using Theorem \ref{th:un_general_repres}, we compute the approximate prices $u^{(0)}(t,x;K)$ and
$u^{(2)}(t,x;K)$ of a series of European puts with strike prices $K \in [0.5, 1.5]$ (we add the
parameter $K$ to the arguments of $u^{(n)}$ to emphasize the dependence of $u^{(n)}$ on the strike
price $K$).  We also compute the price $u(t,x;K)$ using \eqref{eq:w}.  In \eqref{eq:w}, we
truncate the infinite sum at $n=8$.
\par
As prices are often quoted in implied volatilities, we convert prices to implied volatilities by
inverting the Black-Scholes formula numerically.  That is, for a given put price $u(t,x;K)$, we
find $\sig(t,K)$ such that
\begin{align}
u(t,x;K)
    &=  u^{\text{\rm BS}}(t,x;K,\sig(t,K)), \label{eq:impvol.solve}
\end{align}
where $u^{\text{\rm BS}}(t,x;K,\sig)$ is the Black-Scholes price of the put as computed assuming a
Black-Scholes volatility of $\sig$.  For convenience, we introduce the notation
\begin{align}
\text{IV}[u(t,x;K)]:= \sig(t,K)
\end{align}
to indicate the implied volatility induced by option price $u(t,x;K)$.
\par
The results of our numerical experiment are plotted in Figure \ref{fig:compare.lorig}.  We observe
a nearly exact match between the induced implied volatilities $\text{IV}[u^{(2)}(t,x;K)]$ and
$\text{IV}[u(t,x;K)]$, where $u(t,x;K)$ (with no superscript) is computed by truncating
\eqref{eq:w} at $n=8$.

%%%%%%%%%%%%%%%%%%%%%%%%%%%%%%%%%%%%%%%%%%%%%%%%%%%%%%%%%
%           Compairson with Normal Inverse Gaussian
%%%%%%%%%%%%%%%%%%%%%%%%%%%%%%%%%%%%%%%%%%%%%%%%%%%%%%%%%

\subsection{Comparison to NIG-type processes}
\label{sec:nig}
There is a one-to-one correspondence between the generator $\Ac$ of a L\'evy-type process and its \emph{symbol} $\phi$, the correspondence being given by
\begin{align}
\Ac(t,x) e^{i \xi x}
    &=  \phi(t,x,\xi) e^{i \xi x} .
\end{align}
Thus, L\'evy-type processes can be uniquely characterized either through their generator $\Ac$ or their symbol $\phi$.
If $X^0$ is an additive or L\'evy process with symbol $\phi$, we have the following expression for $\hat{p}_0(t,x;T,\xi)$
\begin{align}
 \hat{p}_0(t,x;T,\xi)
    &:= \Eb[ e^{i \xi X_T^0} | X_t^0 = x ]
   =\exp\(i \xi x + \int_t^T \phi(s,x,\xi) ds\) .
\end{align}
\par
A \emph{Normal Inverse Gaussian} (NIG) (see \citet*{barndorff}) is a L\'evy
process $X^0$ with symbol
\begin{align}
\phi(\x)
    &=  i \mu \xi - \del \[ \sqrt{ \alpha^2 - (\beta + i\xi)^2 } - \sqrt{\alpha^2-\beta^2} \].
\end{align}
In Chapter 14, equation (14.1) of \citet{levendorskii}, that authors consider NIG-like Feller processes with symbol
\begin{align}
\phi(x,\x)
    &=  i \mu(x) \xi - \del(x) \[ \sqrt{ \alpha^2(x) - (\beta(x) + i\xi)^2 } - \sqrt{\alpha^2(x)-\beta^2(x)} \], \label{eq:phi.NIG}
\end{align}
where $\mu,\del,\alpha,\beta \in C_b^\infty(\Rb)$, $\del,\alpha>0$, $\mu, \beta \in \Rb$, and
where there exist constants $c$ and $C$ such that $\del(x)>c$, $\alpha(x)-|\beta(x)|>c$ and
$|\mu(x)| \leq C$.  Note that if $X$ is a NIG-type process with symbol $\phi(x,\x)$, then $S=e^X$
is a martingale if and only if $\phi(x,-i)=0$.  Thus, the triple $(\alpha,\beta,\del)$ fixes
$\mu$.
\par
\cite{levendorskii} deduce the following asymptotic expansion for $u(t,x)$ (see the equations following (14.27) and equation (16.40)).
\begin{align}
u(t,x)
    &:= \Eb\[ h(X_T) | X_t = x \] \\
    &=   \int_\Rb d\xi \frac{1}{\sqrt{2\pi}} e^{i \x x} e^{(T-t )\phi(x,\x)}
            \( 1 + \frac{1}{2} (T-t)^2 [ i \d_x \phi(x,\x) ] [ \d_\xi \phi(x,\x) ] + \cdots \)\hh(\x), \label{eq:NIG.expand}
\end{align}
We note that, if one uses the Taylor series expansion of $\Ac$ as in Example \ref{example:Taylor} with $\xb
= x$, then
%$\phi(x,\x)=\phi_0(\x)$ and $\d_x \phi(x,\x) \d_\xi \phi(x,\x) =
%\phi_1(\x)\phi_0'(\x)$.  Thus, from Corollary \ref{thm:u} and equations \eqref{eq:u0.hat} and
%\eqref{eq:u1.hat}, it is easy to see that
expansion \eqref{eq:NIG.expand} is contained within $u_0+u_1$, the first order price approximation obtained in Theorem \ref{th:un_general_repres}.
\par
In our numerical experiment, we use the Taylor
series expansion from Example \ref{example:Taylor} with $\xb = X_t$.
We fix the following parameters
\begin{align}
\del(x)
    &=  \del_0 e^{2(\gam-1)x} , &
\gam
    &=  0.5 , &
\del_0
    &=  2.0 , &
\alpha
    &=  40 , &
\beta
    &=  -10 , &
X_t
    &=  0.0 , &
T-t
    &=  0.25 , \label{eq:nig.params}
\end{align}
and, using Theorem \ref{th:un_general_repres}, we compute the approximate prices $u^{(0)}(t,x;k)$ and
$u^{(3)}(t,x;k)$ of a series of European puts with strike prices $k = \log K \in [-0.3, 0.3]$ (we once again add the
parameter $k$ to the arguments of $u^{(n)}$ to emphasize the dependence of $u^{(n)}$ on the $\log$ strike
price $k$).  We also compute the exact price $u$ using Monte Carlo simulation.  After converting prices to implied volatilities we plot the results in Figure \ref{fig:iv-nig}.  We observe
a nearly exact match between the induced implied volatilities $\text{IV}[u^{(3)}(t,x;k)]$ and
$\text{IV}[u(t,x;k)]$.

%%%%%%%%%%%%%%%%%%%%%%%%%%%%%%%%%%%%%%%%%%%%%%%%%%%%%%%%%
%           Bonds Credit Spreads and Yields
%%%%%%%%%%%%%%%%%%%%%%%%%%%%%%%%%%%%%%%%%%%%%%%%%%%%%%%%%

\subsection{Yields and credit spreads in the JDCEV setting}
\label{sec:JDCEV} Consider a defaultable bond, written on $S$, that pays one dollar at time $T>t$
if no default occurs prior to maturity (i.e., $S_T>0$, $\zeta>T$) and pays zero dollars otherwise.
Then the time $t$ value of the bond is given by
\begin{align}
V_t
    &=  \Eb [ \Ib_{\{\zeta>T\}} | X_t ]
    =       \Ib_{\{\zeta>t\}} u(t,X_t;T), &
u(t,X_t;T)
    &=  \Eb [ e^{-\int_t^T \gam(s,X_s) ds}| X_t ].
\end{align}
We add the parameter $T$ to the arguments of $u$ to indicate dependence of $u$ on the maturity
date $T$.  Note that $u(t,x;T)$ is both the price of a bond and the \emph{conditional survival
probability}: $\mathbb{Q}( \zeta > T | X_t = x, \zeta > t )$.  The \emph{yield} $Y(t,x;T)$ of such
a bond, on the set $\{\zeta>t\}$, is defined as
\begin{align}
Y(t,x;T)
    &:= \frac{- \log u(t,x;T)}{T-t}. \label{eq:Y}
\end{align}
The \emph{credit spread} is defined as the yield minus the risk-free rate of interest. Obviously,
in the case of zero interest rates, we have: yield $=$ credit spread.
\par
In \citet*{JDCEV}, the authors introduce a class of unified credit-equity models known as
\emph{Jump to Default Constant Elasticity of Variance} or JDCEV.  Specifically, in the
time-homogeneous case, the underlying $S$ is described by \eqref{eq:dX} with
\begin{align}
a(x)
    &=  \frac{1}{2} \del^2 e^{2 \beta x}, &
\gam(x)
    &=  b + c \, \del^2 e^{2 \beta x}, &
\nu(x,dz)
    &=  0,
%\sig(x)
    %&= a e^{\beta x}, &
%\gam(x)
    %&= b + c \,a^2 e^{2 \beta x}, &
%\nu(x,dz)
    %&= 0,
\end{align}
where $\del>0$, $b \geq 0$, $c \geq 0 $.  We will restrict our attention to cases in which $\beta
< 0$.  From a financial perspective, this restriction makes sense, as it results in volatility and
default intensity \emph{increasing} as $S \to 0^+$, which is consistent with the leverage effect.
Note that when $c>0$, the asset $S$ may only go to zero via a jump from a strictly positive value.
That is, according to the Feller boundary classification for one-dimensional diffusions (see
\citet*{borodin}, p.14), the endpoint $-\infty$ is a \emph{natural boundary} for the killed
diffusion $X$ (i.e., the probability that $X$ reaches $-\infty$ in finite time is zero).  The
survival probability $u(t,x;T)$ in this setting is computed in \citet{carr}, equation (8.13).  We
have
\begin{align}
u(t,x;T)
%&= u(T-t,x) \\
    &=  \sum_{n=0}^\infty \bigg( e^{-(b+\om n)(T-t)}\frac{\Gam(1+c/|\beta|)\Gam(n+1/(2|\beta|))}{\Gam(\nu+1)\Gam(1/(2|\beta|))n!} \\ & \qquad
            \times A^{1/(2|\beta|)}e^x \exp \(- A e^{-2 \beta x} \) {}_1 F_1(1-n+c/|\beta|;\nu+1;Ae^{-2\beta x}) \bigg) \label{eq:u.exact}
\end{align}
where ${}_1 F_1$ is the Kummer confluent hypergeometric function, $\Gam(x)$ is a Gamma function and
\begin{align}
\nu
    &=  \frac{1 + 2 c}{2|\beta|}, &
A
    &=  \frac{b}{\del^2|\beta|}, &
\om
    &=  2 |\beta| b.
\end{align}
We compute $u(t,x;T)$ using both equation \eqref{eq:u.exact} (truncating the infinite series at
$n=70$) as well as using Theorem \ref{th:un_general_repres}.  We use the Taylor series expansion of $\Ac$
expansion of Example \ref{example:Taylor} with $\xb = X_t$.  After computing bond prices, we then
calculate the corresponding credit spreads using \eqref{eq:Y}.  Approximate spreads are denoted
\begin{align}
Y^{(n)}(t,x;T)
    &:= \frac{- \log u^{(n)}(t,x;T)}{T-t}.
\end{align}
The survival probabilities are and the corresponding yields are plotted in Figure
\ref{fig:survival}.  Values for the yields from Figure \ref{fig:survival} can also be found in
Table \ref{tab:creditspread}.
\begin{remark}
\label{rmk:jdcev}
To compute survival probabilities $u(t,x;T)$, one assumes a payoff function $h(x) = 1$ and obtains  %Note that
%the Fourier transform of a constant is simply a Dirac delta function:
%$\hh(\x)=\delta(\x)$.
$$
u(t,x;T)=\int_{\mathbb{R}}p(t,x;T,y) dy=\hat{p}(t,x;T,0).
$$
Thus, when computing survival probabilities and/or credit
spreads, no numerical integration is required.  Rather, one uses \eqref{eq:un_tilde}
%the following identity
%\begin{align}
%\int_\Rb \uh(\x) \d_\xi^n \del(\x) d\xi
%    &=  (-1)^n \d_\xi^n \uh(\xi)|_{\xi=0}.
%\end{align}
%to compute inverse Fourier transforms.  From the above identity and equations \eqref{eq:u1.hat} - \eqref{eq:u2.hat} one
and easily obtains
\begin{align}
u_0(t,x;T)
    &=  e^{-\(b+\del^2 c e^{2 x \beta }\) \tau}, \\
u_1(t,x;T)
    &=  e^{-\(b+\del^2 c e^{2 x \beta }\) \tau}
            \left( -\del^2 b c e^{2 x \beta } \tau^2 \beta
            +\frac{1}{2} \del^4 c e^{4 x \beta } \tau^2 \beta -\del^4 c^2 e^{4 x \beta } \tau^2 \beta \right), \\
u_2(t,x;T)
    &=  e^{-\(b+\del^2 c e^{2 x \beta }\) \tau}
            \Big(
            -\del^4 c e^{4 x \beta } \tau^2 \beta^2
            -\frac{2}{3} \del^2 b^2 c e^{2 x \beta } \tau^3 \beta^2
            +\del^4 b c e^{4 x \beta } \tau^3 \beta ^2
                        \\ &\qquad
            -2 \del^4 b c^2 e^{4 x \beta } \tau^3 \beta ^2
            -\frac{1}{3} \del^6 c e^{6 x \beta } \tau^3 \beta ^2
            +2 \del^6 c^2 e^{6 x \beta } \tau^3 \beta ^2
                        \\ &\qquad
            -\frac{4}{3} \del^6 c^3 e^{6 x \beta } \tau^3 \beta ^2
            +\frac{1}{2} \del^4 b^2 c^2 e^{4 x \beta } \tau^4 \beta^2
            -\frac{1}{2} \del^6 b c^2 e^{6 x \beta } \tau^4 \beta^2
            +\del^6 b c^3 e^{6 x \beta } t^4 \beta^2
                        \\ &\qquad
            +\frac{1}{8} \del^8 c^2 e^{8 x \beta } \tau^4 \beta^2
            -\frac{1}{2} \del^8 c^3 e^{8 x \beta } \tau^4 \beta^2
            +\frac{1}{2} \del^8 c^4 e^{8 x \beta } \tau^4 \beta^2
            \Big). \label{eq:survival}
\end{align}
where $\tau:=T-t$.  It is interesting to note that
\begin{align}
u^{(n)}(t,x;T)
    &=  \sum_{k=0}^n u_k(t,x;T)
    =       e^{-\(b+\del^2 c e^{2 x \beta }\) \tau} \( 1 + \Oc(\tau^2) \) ,
\end{align}
which guarantees that the $\d_\tau u^{(n)} |_{\tau=0} < 0$ (i.e., as $\tau$ increases from zero, the approximate survival probability decreases, as expected).
\end{remark}

%%%%%%%%%%%%%%%%%%%%%%%%%%%%%%%%%%%%%%%%%%%%%%%%%%%%%%%%%%%%%%%%%
%
%               Hertmite vs Taylor
%
%%%%%%%%%%%%%%%%%%%%%%%%%%%%%%%%%%%%%%%%%%%%%%%%%%%%%%%%%%%%%%%%%

\subsection{Hermite vs Taylor approximations}
\label{sec:hermite}
We are interested in comparing the relative accuracy of the Taylor series and Hermite polynomial approximations (examples \ref{example:Taylor} and \ref{example:Hilbert}).  To this end, we consider the Constant Elasticity of Variance (CEV) model of \cite{CoxCEV}.  The $\log$ dynamics are given by
\begin{align}
a(x)
    &=  \frac{1}{2} \del^2 e^{2(\beta-1)x}, &
\nu(x,dz)
    &=  0 , &
\gam(x)
    &=  0 , &
\beta
    &\in [0,1], \label{eq:CEV}
\end{align}
We consider two approximations for the variance function $a$ -- Taylor and Hermite.  We have
\begin{align}
\text{Taylor}:&&
a_\text{T}^{(n)}(x)
    &:=     \sum_{k=0}^n \frac{\d_x^k a(\xb)}{k!} (x - \xb)^k , \\
\text{Hermite}:&&
a_\text{H}^{(n)}(x)
    &:=     \sum_{k=0}^n \< a ,\Hv_k( \cdot - \xb ) \>_\Gamma \Hv_k(x - \xb) ,
\end{align}
Fix a maturity date $T$ and let $t<T$.  Denote by $u(t,x;K)$ the price at time $t<T$ of a call option with strike price $K$.  The exact call option price is given in \cite{CoxCEV}.  Denote by $u_\text{T}^{(n)}(t,x;K)$ the $n$th order approximation of a call price, as obtained using the Taylor series approximation of $a$.  Likewise, denote by $u_\text{H}^{(n)}(t,x;K)$ the $n$th order approximation of a call price, as obtained using the Hermite polynomial approximation of $a$.  In figure \ref{fig:hermite} we plot as a function of $\log$ moneyness $k := (\log K - x)$ the exact implied volatility
$\text{IV}[u(t,x;K)]$
as well as the Taylor and Hermite approximations of implied volatility
$\text{IV}[u_\text{T}^{(n)}(t,x;K)]$ and $\text{IV}[u_\text{H}^{(n)}(t,x;K)]$
for $n=\{0,1,2,3,4\}$.  We also plot, as a function of $x$ the exact diffusion coefficient
$a(x)$
as well as the Taylor and Hermite approximations of the diffusion coefficient
$a_\text{T}^{(n)}(x)$ and $a_\text{H}^{(n)}(x)$
for $n=\{0,1,2,3,4\}$.
It is clear from Figure \ref{fig:hermite} that the Taylor expansion $a_\text{T}^{(n)}(x)$ provides a more accurate approximation of $a(x)$ than the Hermite expansion $a_\text{H}^{(n)}(x)$ for every $n \leq 4$.  Not surprisingly, Figure \ref{fig:hermite} also shows that implied volatility induced by the Taylor expansion $\text{IV}[u_\text{T}^{(n)}(t,x;K)]$ provides a more accurate approximation of the exact implied volatility $\text{IV}[u(t,x;K)]$ than does  the Hermite approximation $\text{IV}[u_\text{H}^{(n)}(t,x;K)]$.  Though, for $n=4$, both approximations are remarkably accurate for $\log$ moneyness $k \in (-0.4,0.4)$.

%%%%%%%%%%%%%%%%%%%%%%%%%%%%%%%%%%%%%%%%%%%%%%%%%%%%%%%%%%%%%%%%%
%
%               Hertmite vs Taylor
%
%%%%%%%%%%%%%%%%%%%%%%%%%%%%%%%%%%%%%%%%%%%%%%%%%%%%%%%%%%%%%%%%%

\subsection{Accuracy: jumps vs no jumps}
\label{sec:jump-vs-nojump}
In this example, we examine (numerically) whether or not the addition of jumps affects the accuracy of our asymptotic approximation for Call prices.  To this end, we consider the CEV-like L\'evy-type process with Gaussian jumps, introduced in Section \ref{sec:CEV}.  We fix the following parameters:
\begin{align}
\del
    &= 0.20, &
\beta
    &= 0.5, &
m
    &= -0.1, &
\eta
    &=  0.2, &
S_0 = e^x
    &=  1 , &
T-t
        &=  0.5 .
\end{align}
We consider two scenarios: $\lam = 0$ (no jumps) and $\lam = 0.2$ (with jumps).  In each scenario we compute our third order approximation for Call prices $u^{(3)}(t,x;K)$ using the Taylor series approximation (Example \ref{example:Taylor}).  We also compute, in the case of no jumps, the exact call price using the formulas given in \cite{CoxCEV}.  In the case where the jump intensity $\lam$ is non-zero, we compute a 95\% confidence interval for call prices via Monte Carlo simulation.  Finally, call prices are converted to implied volatilities: $\text{IV}[u(t,x;K)]$.  The results are plotted in Figure \ref{fig:jump-vs-nojump}.

%%%%%%%%%%%%%%%%%%%%%%%%%%%%%%%%%%%%%%%%%%%%%%%%%%%
%
%                   Begin Proof
%
%%%%%%%%%%%%%%%%%%%%%%%%%%%%%%%%%%%%%%%%%%%%%%%%%%%

\section{Proof of Theorem \ref{t1}}
\label{sec:proof} %\proof
For sake of simplicity we only prove the assertion when {the default
intensity and mean jump size are zero $\g=m=0$, when the jump intensity and diffusion component
are time-independent $a(t,x)\equiv a(x)$, $\lam(t,x)\equiv \lam(x)$ and when the standard
deviation of the jumps is constant $\delta(x)\equiv\delta$.} Thus we consider the
integro-differential operator
\begin{align*}
L u(t,x)
    &=\p_t u(t,x)+\frac{a(x)}{2}(\p_{xx}-\p_{x})u(t,x) -\lam(x) \left(  e^{\frac{\delta^2}{2}}-1  \right)
  \p_x u(t,x) \\ &\qquad +\lambda(x)\int_{\R}\left( u(t,x+z)-u(t,x) \right)\nu_{\delta^2}(dz) ,
\end{align*}
with
\begin{align}
\nu_{\delta^{2}}(dz)
    &=\frac{1}{\sqrt{2\pi }\delta }e^{-\frac{z^2} {2 \delta^{2}}}dz.
\end{align}
We will give some details on how to extend the proof to the general case at the end of the section.
% i.e. default-free case and jumps with null expectation. For the general case the proof is totally analogous.
Our idea is to use our expansion as a \emph{parametrix}.  That is, our expansion will serve as the
starting point of the classical iterative method
introduced by \cite{Levi} to construct the fundamental solution $p(t,x;T,y)$ of $L$.
% modify and
%adapt the standard characterization of the fundamental solution given by the parametrix method
%originally . In particular,  for
%
%will be constructed starting from a {\it parametrix} function by means of an iterative argument
%and by suitably controlling the errors of the approximation.
Specifically, as in \cite{RigaPagliaraniPascucci}, we take as a parametrix our $N$-th order
approximation $p^{(N)}(t,x;T,y)$ with $\bar{x}=y$ in
\eqref{eq:taylor_coef_gaus}-\eqref{eq:taylor_coef_lev}. The case $\xb=x$ can be analogously proved
by using the backward parametrix approach (see \cite{pascucci-parametrix}). For sake of brevity we
skip the details for the latter case.

%$$\overline{x}=y+\left(  \frac{a(y)}{2} -\l(y)  e^{\frac{\delta^2}{2}}+\l(y)  \right)  (T-t).$$
%\begin{remark}\label{remerr1}
%Note that, at $0$-order we have $$ p_{y}^0(t,x;T,y)=\G^{a\left(y\right),\delta^2,\l}(t,x;T,y) $$
%%  $$\G^{a\left(y+\left(  \frac{a(y)}{2} -\l(y)  e^{\frac{\delta^2}{2}}+\l(y)  \right) (T-t)\right),\d^2,\l\left(y+\left(  \frac{a(y)}{2} -\l(y)  e^{\frac{\delta^2}{2}}+\l(y)  \right) (T-t)\right)}(t,x;T,y)$$
%according to the notation introduced in Section \ref{sec:estimates}. Thus, Lemmas \ref{lemestim1},
%\ref{lemestim2}, \ref{lemestim3}, \ref{lemestim5} and \ref{lemestim6} can be applied to
%$p_{y}^0(t,x;T,y)$.
%\end{remark}
%We first prove the case $N=1$.
By analogy with the classical approach (see, for instance,
\cite{Friedman} and \cite{DiFrancescoPascucci2}, \cite{Pascucci2011} for the pure diffusive case,
or \cite{GarroniMenaldi} for the
integro-differential case), we have %that $p$ takes the form
\begin{align}\label{eqbound2}
 p(t,x;T,y)=p^{(N)}(t,x;T,y)+\int_{t}^{T}\int_{\R}p^{(0)}(t,x;s,\x)\Phi(s,\x;T,y)d\x ds ,
\end{align}
where $\Phi$ is %the function in \eqref{eqbound1} below, which is
determined by imposing the condition %$L p\left(,\cdot,\cdot;T,y\right)=0$:
  $$0=L p(t,x;T,y)=L p^{(N)}(t,x;T,y)+\int_{t}^{T}\int_{\R}L p^{(0)}(t,x;s,\x)\Phi(s,\x;T,y)d\x ds
  -\Phi(t,x;T,y).$$
Equivalently, we have
  $$\Phi(t,x;T,y)=L p^{(N)}(t,x;T,y)+\int_{t}^{T}\int_{\R}L p^{(0)}(t,x;s,\x)\Phi(s,\x;T,y)d\x ds , $$
and therefore by iteration
  \begin{align}\label{eqbound1}
  \Phi(t,x;T,y)=\sum_{n=0}^{\infty}Z^{(N)}_{n}(t,x;T,y),
  \end{align}
where
  \begin{align}\label{eqbound6}
  Z^{(N)}_{0}(t,x;T,y) & :=L p^{(N)}(t,x;T,y), \\ \label{eqbound9}
  Z^{(N)}_{n+1}(t,x;T,y)& :=\int_{t}^{T}\int_{\R}L p^{(0)}(t,x;s,\x)Z^{(N)}_{n}(s,\x;T,y)d\x ds.
  \end{align}
The proof of Theorem \ref{t1} is based on several technical lemmas which we relegate to
Section \ref{sec:pointwise_estimates}. In particular, we will use such preliminary estimates to
provide pointwise bounds
for each of the terms $Z^{(N)}_{n}$ in \eqref{eqbound1}. %\blu{(I changed the previous equation reference.  Please double check that I have changed it correctly)}.
Finally, these bounds combined with formula \eqref{eqbound2} give the  estimate of $\left|p(t,x;T,y)-
p^{(N)}(t,x;T,y)\right|$.
\par
For any $\a,\theta>0$ and $\l\geq 0$, consider the integro-differential operators
  \begin{align}
  L^{\a,\theta,\l}u(t,x)&=\p_t u(t,x)+\frac{\a}{2}(\p_{xx}-\p_{x})u(t,x)-\l \left(  e^{\frac{\theta}{2}}-
  1  \right)\p_x u(t,x) +\l\int_{\R}\left( u(t,x+z)-u(t,x) \right) \nu_{\theta}(dz),\\
  \bar{L}^{\a,\theta,\l}u(t,x)&=\p_t u(t,x)+\frac{\a}{2}\p_{xx}u(t,x) +\l\int_{\R}\left( u(t,x+z)-u(t,x) \right) \nu_{\theta}(dz).
  \end{align}
%Here, $\nu_{\delta^{2}}$ is the Lévy measure of a compound Poisson process
%  $$
%  X_t=\sum_{n=0}^{N}Z_n,
%  $$
%where $N\sim \mathrm{Po}(\l t)$ and $Z_n\sim N(0,\delta)$, $n=1,2,\cdots$.\\
The function $\G^{\a,\theta,\l}(t,x;T,y):=\G^{\a,\theta,\l}(T-t,x-y)$ where
  \begin{align}
  \G^{\a,\theta,\l}(t,x)&:=e^{-\l t}\sum_{n=0}^{\infty}\frac{(\l t)^n}{n!} \G_n^{\a,\theta,\l}(t,x),\label{eq:proof_Gamma}\\
  \G_n^{\a,\theta,\l}(t,x)&:=\frac{1}{\sqrt{2\pi (\a t+n\theta)}}\exp\left(-\frac{\left(x-
  \left(\frac{\a}{2}+\l  e^{\frac{\theta}{2}}-\l  \right)t\right)^2}
  {2(\a t+n\theta)}\right) \label{eq:proof_Gamma_n},
  \end{align}
is the fundamental solution of $L^{\a,\theta,\l}$. Analogously, the function
$\bar{\G}^{\a,\theta,\l}(t,x;T,y):=\bar{\G}^{\a,\theta,\l}(T-t,x-y)$ where
  \begin{align}
  \bar{\G}^{\a,\theta,\l}(t,x)&:=e^{-\l t}\sum_{n=0}^{\infty}\frac{(\l t)^n}{n!} \bar{\G}_n^{\a,\theta}(t,x),\\
  \bar{\G}_n^{\a,\theta}(t,x)&:=\frac{1}{\sqrt{2\pi (\a t+n\theta)}}\exp\left(-\frac{x^2}
  {2(\a t+n\theta)}\right) \label{eq:proof_Gamma_n_bar},
  \end{align}
is the fundamental solution of $\bar{L}^{\a,\theta,\l}$. Note that under our assumptions, at order
zero, by \eqref{eq:p_0_gaus}-\eqref{eq:p_0_gaus_1} we have
\begin{align}\label{e31}
 p^{(0)}(t,x;T,y)=\G^{a\left(y\right),\delta^2,\lam(y)}(t,x;T,y).
\end{align}
We also recall the definition of convolution operator $\mathscr{C}_{\theta}$:
\begin{align}\label{eqestim1}
\mathscr{C}_{\theta}f(x)=\mathscr{C}^x_{0,\theta}f(x):=\int_{\R}f(x+z)\frac{1}{\sqrt{2\pi\theta}}e^{-\frac{z^2}{2\theta}}
dz.
\end{align}
Note that, for any $\theta>0$, we have
\begin{align}
 \mathscr{C}_{\theta} \G^{\a,\theta,\l}(t,\cdot)(x)&=e^{-\l t}\sum_{n=0}^{\infty}\frac{(\l
 t)^n}{n!}\G_{n+1}^{\a,\theta,\l}(t,x),\\ %\quad\forall\a,\t>0,\ \l\geq 0,\ z\in\R.
\mathscr{C}_{\theta} \bar{\G}^{\a,\theta,\l}(t,\cdot)(x)&=e^{-\l t}\sum_{n=0}^{\infty}\frac{(\l
 t)^n}{n!}\bar{\G}_{n+1}^{\a,\theta}(t,x),
\end{align}
with $\bar{\G}_{n}^{\a,\theta}$ and $\G_{n}^{\a,\theta,\l}$ as in \eqref{eq:proof_Gamma_n} and \eqref{eq:proof_Gamma_n_bar} respectively.
%  $$\G^{a\left(y+\left(  \frac{a(y)}{2} -\l(y)  e^{\frac{\delta^2}{2}}+\l(y)  \right) (T-t)\right),\d^2,\l\left(y+\left(  \frac{a(y)}{2} -\l(y)  e^{\frac{\delta^2}{2}}+\l(y)  \right) (T-t)\right)}(t,x;T,y)$$
%The rest of this section is devoted to prove some pointwise estimates of $\G^{\a,\delta,\l}$,
%$\mathscr{C}_{\delta} \G^{\a,\delta,\l}$ and their spatial derivatives, uniformly with respect
%to $(\a,\delta,\l)\in [m_{\a},M_{\a}]\times [m_{\delta},M_{\delta}]\times [0,M_{\l}]$, for
%some constants $M_{\a}\geq m_{\a}>0$, $M_{\delta}\geq m_{\delta}>0$ and $M_{\l}>0$. In order to
%shorten notation, we will consider in what follows $m_{\a}=m_{\delta}$ and
%$M_{\a}=M_{\delta}=M_{\l}$. Nevertheless, any following statement is still true in the general
%case.

%for some positive constants $m, M$.

\begin{proposition}\label{properr3}
For any $c>1$ and $\t >0$, there exists a positive constant $C$, only dependent on
$c,\t,M,N$, and $(\|a_{i}\|_{\infty},\|\lambda_{i}\|_{\infty} )_{i=1,\cdots,N+1}$, such that
\begin{align}\label{eqbound7}
\big|Z^{(N)}_n(t,x;T,y) \big| \leq \frac{C^{n+1}
{(T-t)^{\frac{\min{(1,N)}+n-1}{2}}}}{\sqrt{n!}}\left(1+\|\lambda_{1}\|_{\infty}\mathscr{C}_{cM}^{n+1}\right)\,
\bar{\G}^{cM,cM,cM}(t,x;T,y),
\end{align}
for any $n\in\mathbb{N}_0$, $x,y\in\R$ and $t,T\in\R$ with $0<T-t\le \t$.
%, and where
%  $$\kappa_{n}=\frac{C^n}{\G_{E}\left(\frac{n+2}{2}\right)}$$
%and
%$\G_{E}$ denotes the Euler Gamma function.
\end{proposition}
The proof of Proposition \ref{properr3} is postponed to Section \ref{sec:proof_estimates_Z_N}.
We are now in position to prove Theorem \ref{t1}. %for $N=1$.
Indeed, by equations \eqref{eqbound2},
\eqref{eqbound1} and Proposition \ref{properr3} we have
\begin{align*}
 &\big|p(t,x;T,y)-p^{(N)}(t,x;T,y)\big|&\\
 &\leq \sum_{n=0}^{\infty}
 \frac{C^{n+1}}{\sqrt{n!}} \int_t^T {(T-s)^{\frac{\min{(1,N)}+n-1}{2}}} \int_{\R}
 p^{(0)}(t,x;s,\x)\  \left(1+\|\lam_{1}\|_{\infty} \mathscr{C}_{cM}^{n+1}\right)  \bar{\G}^{cM,cM,cM}(s,\x;T,y)  d\x ds&\\
%
%&+  \|\lam_{1}\|_{\infty}  \sum_{n=0}^{\infty}
% \frac{C^{n+1}}{\sqrt{n!}} \int_t^T (T-s)^{\frac{n}{2}} \int_{\R}
% p^{(0)}(t,x;s,\x)\ \mathscr{C}_{cM}^{n+2} \bar{\G}^{cM,cM,cM}(s,\x;T,y)  d\x ds
\intertext{(and by Lemma \ref{lemestim1} with $\eta=0$% and $N=1$ respectively
)} &\leq \sum_{n=0}^{\infty}
 \frac{C^{n+1}}{\sqrt{n!}} \int_t^T {(T-s)^{\frac{\min{(1,N)}+n-1}{2}}} \int_{\R}
 \bar{\G}^{cM,cM,cM}(t,x;s,\x)\  \left(1+\|\lam_{1}\|_{\infty} \mathscr{C}_{cM}^{n+1}\right)  \bar{\G}^{cM,cM,cM}(s,\x;T,y)  d\x ds&.%\\
%&\hspace{-24pt}+  \|\lam_{1}\|_{\infty}  \sum_{n=0}^{\infty}
% \frac{C^{n+1}}{\sqrt{n!}} \int_t^T (T-s)^{\frac{n}{2}} \int_{\R}
% \bar{\G}^{cM,cM,cM}(t,x;s,\x)\ \mathscr{C}_{cM}^{n+2} \bar{\G}^{cM,cM,cM}(s,\x;T,y)  d\x ds
\end{align*}
Now, by the semigroup property
\begin{align}\label{eqbound20}
 \int_{\R} \mathscr{C}_{\theta}^{k}\bar{\G}^{\alpha,\theta,\l}(t,x;s,\x)
 \mathscr{C}_{\theta}^{N}\bar{\G}^{\alpha,\theta,\l}(s,\x;T,y)\, d\x
 =\mathscr{C}_{\theta}^{k+N}\bar{\G}^{\alpha,\theta,\l}(t,x;T,y),\qquad  k,N\in \mathbb{N}_0,
\end{align}
we get
\begin{align*}
 \big|p(t,x;T,y)-p^{(N)}(t,x;T,y)\big| & \leq 2\,{(T-s)^{\frac{\min{(1,N)}+1}{2}}}
 \left( \sum_{n=0}^{\infty} \frac{C^{n+1}
 (T-t)^{\frac{n}{2}}}{\sqrt{n!}} \left(1+\|\lam_{1}\|_{\infty} \mathscr{C}_{cM}^{n+1}\right)\bar{\G}^{cM,cM,cM}(t,x;T,y) \right),
\end{align*}
for any $x,y\in\R$ and $t,T\in\R$ with $0<T-t\le \t$, and since $$ \sum_{n=0}^{\infty} \frac{C^{n+1}
 (T-t)^{\frac{n}{2}}}{\sqrt{n!}}\,  \mathscr{C}_{cM}^{n+1}\bar{\G}^{cM,cM,cM}(t,x;T,y)
$$ can be easily checked to be convergent, this concludes the proof of Theorem \ref{t1}. %for $N=1$.

We conclude this section with a brief discussion on how to drop the additional hypothesis on the coefficients introduced at the beginning of the section. In order to include state-dependency in the standard deviation of the jumps, i.e. $\delta=\delta(x)$, no modification is required in the first part of the proof since all the preliminary lemmas in Section \ref{sec:pointwise_estimates} naturally apply to the general case. On the other hand, the proof of Proposition \ref{properr3} requires some simple modifications to account for the additional terms in the expansion introduced by the state dependency of the convolution operator (see Proposition \ref{prop:G_n_gaus}).
To extend the proof to non-null mean of the jumps, i.e. $m=m(x)\neq 0$, it is enough to extend Lemmas \ref{lemestim1}-\ref{lemestim9} to the more general functions such as
  \begin{align}
  \G^{\a,m,\theta,\l}(t,x)&:=e^{-\l t}\sum_{n=0}^{\infty}\frac{(\l t)^n}{n!} \G_n^{\a,m,\theta,\l}(t,x),\label{eq:proof_Gamma_bis}\\
  \G_n^{\a,m,\theta,\l}(t,x)&:=\frac{1}{\sqrt{2\pi (\a t+n\theta)}}\exp\left(-\frac{\left(x+nm-
  \left(\frac{\a}{2}+\l  e^{\frac{\theta}{2}}-\l  \right)t\right)^2}
  {2(\a t+n\theta)}\right) \label{eq:proof_Gamma_n_bis}.
  \end{align}
As for the time-dependency of the coefficients $a(t,x)$ and $\gamma(t,x)$, the proof easily follows by the regularity hypothesis iii) in Assumption \ref{assumption:parab}.

\subsection{Proof of Proposition \ref{properr3}}\label{sec:proof_estimates_Z_N}

The proof of Proposition \ref{properr3} is based on the two following propositions.
\begin{proposition}\label{properr2}
For any $c>1$ and $\t >0$, there exists a positive constant $C$, only dependent on
$c,\t,M,\|\lambda_{1}\|_{\infty}$ and $\|a_{1}\|_{\infty}$, such that
\begin{align}\label{eqbound5}
\left|(x-y)^{2-n}(\partial_{xx}-\partial_{x})p_{n}(t,x;T,y)\right| \leq
C(1+\|\lambda_{1}\|_{\infty}\mathscr{C}_{cM})\ \bar{\G}^{cM,cM,cM}(t,x;T,y),
\end{align}
for any {$n \in \{0,1\}$}, $x,y\in\R$ and $t,T\in\R$ with $0<T-t\le \t$.
\end{proposition}
\begin{proof}
Recalling the expression of $p_{0}(t,x;T,y)\equiv p^{(0)}(t,x;T,y)$ in \eqref{e31}, the case
$n=0$ directly follows from Lemmas \ref{lemestim3}, \ref{lemestim2} %, \ref{lemestim6}
and \ref{lemestim1} with $\eta=0$.
%by the property
%\begin{align}\label{eqbound4}
%\frac{1}{(\a(T-t)+n\delta^2)^{k/2}}\leq\frac{1}{(m(T-t))^{k/2}}\quad \forall k\geq 1,
%\end{align}
%and eventually by Lemmas \ref{lemestim5} and \ref{lemestim1}.

For the case $n=1$ we first observe that, by Theorem \ref{th:un_general_repres} along with
Proposition \ref{prop:G_n_gaus}, the function
$p_1(t,x;T,y)$ %in \eqref{eq:sol_fond.expand}
takes the form
%solution of \eqref{eq:v.n.pide} with $n=1$, $\Ac_{1}$ as in \eqref{e41}-\eqref{e8} and $h=\delta_{y}$,
\begin{align}
 p_{1}(t,x;T,y)=\left( (T-t)(x-y)+\frac{(T-t)^2}{2}J \right)\Ac_1 p^{(0)}(t,x;T,y),
\end{align}
where $J$ is the operator
\begin{align}\label{eqbound14}
J=a_0 (2\partial_x -1)  - \lambda_0 \left( e^{\frac{\delta^2}{2}}-1+\delta^2\partial_x
\mathscr{C}_{\delta^2} \right) ,
\end{align}
whereas $\Ac_1$ acts as
\begin{align}
 \Ac_1 %&=a_1(\partial_{xx}-\partial_{x})u(x)-\lambda_1\left(\partial_x u(x)
% \int_{\R}(e^z-1)\nu_{\delta^2}(dz)- \int_{\R}\left(u(x+z)-u(x)\right)\nu_{\delta^2}(dz)  \right)\\
 &=
 a_1(\partial_{xx}-\partial_{x})-\lambda_1\left(\left( e^{\frac{\delta^2}{2}}-1
 \right)\partial_x - \mathscr{C}_{\delta^2}+1  \right),
\end{align}
and $\mathscr{C}_{\delta^2}$ is the convolution operator defined in \eqref{eqestim1}.
Therefore, we have %\blu{((I think it is confusing to have
%$v^{(i)}$ and $p^{(i)}$ in the same equation.  I would recommend using only $p^{(1)}$.  If you
%refer to equations in the main text with $v^{(1)}$, you can just add ``with $v^{(1)}\to
%p^{(1)}$''.))} \red{No, $v^{(n)}$ was wrong. What acually mean $v_{n}$ that is different from
%$v^{(n)}=p^{(n)}$. Now it should be ok..}
\begin{align*}
(x-y)(\partial_{xx}-\partial_{x})v_{1}(t,x;T,y)&=(T-t)(x-y)\left((x-y)(\partial_{xx}-\partial_{x})+2\partial_x
-1   \right)\Ac_1 p^{(0)}(t,x;T,y)\\ &\quad + \frac{(T-t)^2}{2}(x-y)J
(\partial_{xx}-\partial_{x})\Ac_1 p^{(0)}(t,x;T,y),
\end{align*}
%with $J$ as in \eqref{eqbound14} and
In the computations that follow below, in order to shorten notation, we omit the dependence of
$t,x,T,y$ in any function. By the commutative property of the operators $\partial_x$ and
$\mathscr{C}$, and by applying Lemmas \ref{lemestim3}, \ref{lemestim5} and \ref{lemestim2} with
$\eta=1$, there exists a positive constant $C_1$ only dependent on $c,\t,M,\|\lambda_{1}\|_{\infty}$
and $\|a_{1}\|_{\infty}$ such that
\begin{align}
|(T-t)(x-y)\left((x-y)(\partial_{xx}-\partial_{x})+2\partial_x -1
\right)a_1(\partial_{xx}-\partial_{x}) p^{(0)}|&\leq C_1 \G^{c a(y), c\delta^2,\lambda(y)},\\
\left|(T-t)(x-y)\left((x-y)(\partial_{xx}-\partial_{x})+2\partial_x -1   \right)
\lambda_1\left(\left( e^{\frac{\delta^2}{2}}-1 \right)\partial_x +1  \right) p^{(0)}\right|&\leq
C_1 \left(\G^{c a(y), c\delta^2,\lambda(y)}+\G^{c a(y),4 c\delta^2,\lambda(y)}\right),\\
\frac{(T-t)^2}{2}|(x-y)J (\partial_{xx}-\partial_{x})a_1(\partial_{xx}-\partial_{x}) p^{(0)}|&\leq
C_1 \left(\G^{c a(y), c\delta^2,\lambda(y)}+ \G^{c a(y),4 c\delta^2,\lambda(y)}\right),\\
\label{eqbound16} \frac{(T-t)^2}{2}\left|(x-y)J (\partial_{xx}-\partial_{x})\lambda_1\left(\left(
e^{\frac{\delta^2}{2}}-1 \right)\partial_x +1  \right) p^{(0)}\right|&\leq C_1 \left(\G^{c a(y),
c\delta^2,\lambda(y)}+\G^{c a(y),4 c\delta^2,\lambda(y)}\right),
\end{align}
for any $x,y\in\R$ and $t,T\in\R$ with $0<T-t\le \t$. Analogously, by the commutative property of
$\partial_x$ and $\mathscr{C}$, and by applying Lemmas \ref{lemestim2}, \ref{lemestim3},
\ref{lemestim8} and \ref{lemestim7} with $\eta=2$, there exists a positive constant $C_2$ only
dependent on $c,\t,M,\|\lambda_{1}\|_{\infty}$ and $\|a_{1}\|_{\infty}$ such that
\begin{align}
|(T-t)(x-y)\left((x-y)(\partial_{xx}-\partial_{x})+2\partial_x -1   \right)\lambda_1
\mathscr{C}_{\delta^2} p^{(0)}|&\|\lambda_{1}\|_{\infty}\leq C_2 \,
\left(\mathscr{C}_{c\delta^2}\G^{c a(y), c\delta^2,\lambda(y)}+\mathscr{C}_{4c\delta^2}\G^{c a(y),
4c\delta^2,\lambda(y)}\right),\\ \label{eqbound17} \frac{(T-t)^2}{2}|(x-y)J
(\partial_{xx}-\partial_{x})\lambda_1 \mathscr{C}_{\delta^2} p^{(0)}|&\leq
\|\lambda_{1}\|_{\infty}\, C_2 \left(\mathscr{C}_{c\delta^2}\G^{c a(y),
c\delta^2,\lambda(y)}+\mathscr{C}_{4c\delta^2}\G^{c a(y), 4c\delta^2,\lambda(y)}\right),
\end{align}
for any $x,y\in\R$ and $t,T\in\R$ with $0<T-t\le \t$. Then, \eqref{eqbound5} follows from
\eqref{eqbound16} and \eqref{eqbound17} by applying Lemma \ref{lemestim1} with $\eta=0$ and $\eta=1$
respectively. %, as $m\leq a(y),16 \delta^2\leq M$ and $0\leq \lambda(y)\leq M$ for any $y\in\R$.
\end{proof}

\begin{proposition}\label{properr4}
For any $c>1$ and $\t >0$, there exists a positive constant $C$, only dependent on
$c,\t,M,\|\lambda_{1}\|_{\infty}$ and $\|a_{1}\|_{\infty}$, such that
\begin{align}\label{eqbound18}
\left|(x-y)^{2-n}\left(\left( e^{\frac{\delta^2}{2}}-1
\right)\partial_x+\mathscr{C}_{\delta^2}-1\right)p_{n}(t,x;T,y)\right| \leq
C(1+\mathscr{C}_{cM})\, \bar{\G}^{cM,cM,cM}(t,x;T,y),
\end{align}
for any $n \in \{0,1\}$, $x,y\in\R$ and $t,T\in\R$ with $0<T-t\le \t$.
\end{proposition}
\begin{proof}
For simplicity we only prove the thesis for $n=0$. The proof for $n=1$ is entirely analogous to
that of Proposition \ref{properr2}. Once again, hereafter we omit the dependence of $t,x,T,y$ in
any function we consider. Recalling the expression of $p_{0}(t,x;T,y)\equiv p^{(0)}(t,x;T,y)$ in
\eqref{e31}, by Lemmas \ref{lemestim3}, \ref{lemestim2} and \ref{lemestim8}, there exists a
positive constant $C_1$ only dependent on $c,\t,M$ such that $$ \left|(x-y)^{2}\left(\left(
e^{\frac{\delta^2}{2}}-1 \right)\partial_x+\mathscr{C}_{\delta^2}-1\right)v_{0}\right|\leq
C_1\left(  \G^{c a (y),4c \delta^2, \lambda(y)}+(1+\mathscr{C}_{16 c \delta^2}) \G^{c a (y),16c
\delta^2, \lambda(y)} \right), $$ for any $x,y\in\R$ and $t,T\in\R$ with $0<T-t\le \t$. Then,
\eqref{eqbound18} follows from Lemma \ref{lemestim1} with $\eta=0$ and with $\eta=1$. %, as $m\leq a(y),16
%\delta^2\leq M$ and $0\leq \lambda(y)\leq M$ for any $y\in\R$.
\end{proof}

We are now in position to prove Proposition \ref{properr3}.
\begin{proof}[Proof of Proposition \ref{properr3}]
We first prove the case $N=1$. Let us define the operators
\begin{align*}
 L_0=\partial_t + \Ac_{0},\qquad L_1=\partial_t + \Ac_{0} + (x-y)\Ac_{1}.
\end{align*}
%with $\Ac_{0},\Ac_{1}$ as in \eqref{no number yet}.
Let us recall that, by \eqref{eq:v.0.pide} and \eqref{eq:v.n.pide} with $n=1$, we have
  $$L_{0}p_{0}=0,\qquad L_{0}p_{1}=-(L_{1}-L_{0})p_{0}.$$
Thus, by \eqref{eqbound6} we have
%\begin{align*}
%Z^{(1)}_{0}(t,x;T,y) &= (a(y)-a_0(y))(\partial_{xx}-\partial_{x}) v^1_y(t,x;t,x)\\
%               &\quad + \left( a(y)-a_0(y)-a_1(y)(x-y)\right) v^0_y(t,x;t,x).
%\end{align*}
%In fact, by \eqref{eqbound6} we have
\begin{align*}
 Z^{(1)}_{0}(t,x;T,y) =L p^{(1)}(t,x;T,y)&=Lp_{0}(t,x;T,y)+Lp_{1}(t,x;T,y)\\
%  &= (L-L_1)v^{(0)}(t,x;T,y)+L_0 v^{(0)}(t,x;T,y)\\ &\quad +(L-L_0)v^{(1)}(t,x;T,y)+L_0
% v^{(1)}(t,x;T,y)+(L_1-L_0)v^{(0)}(t,x;T,y)
%\intertext{()}
 &= (L-L_1)p_{0}(t,x;T,y)+(L-L_0)p_{1}(t,x;T,y),
\end{align*}
where $(L-L_0)$ and $(L-L_1)$ are explicitly given by
\begin{align}\label{eqbound21}
(L-L_0)&=(a(x)-a(y))(\partial_{xx}-\partial_{x})+ (\lam(x)-\lam(y)) \left(\left(
e^{\frac{\delta^2}{2}}-1 \right)\partial_x+\mathscr{C}_{\delta^2}-1\right),\\ (L-L_1)&=(a(x)-a(y)-
a'(y)(x-y))(\partial_{xx}-\partial_{x})\\ &\quad + (\lam(x)-\lam(y)-\lam'(y)(x-y)) \left(\left(
e^{\frac{\delta^2}{2}}-1 \right)\partial_x+\mathscr{C}_{\delta^2}-1\right).
\end{align}
Thus, by the Lipschitz assumptions on $a$, $\lam$ and their first order derivatives, we obtain
\begin{align}
|Z^{(1)}_{0}(t,x;T,y)|&\leq \|a_{2}\|_{\infty}|x-y|^{2}|(\partial_{xx}-\partial_{x})
p_{0}(t,x;T,y)|+\|a_{1}\|_{\infty}|x-y||(\partial_{xx}-\partial_{x}) p_{1}(t,x;T,y)|\\ &\qquad\ +
\|\lam_{2}\|_{\infty}|x-y|^{2}  \left| \left(\left( e^{\frac{\delta^2}{2}}-1
\right)\partial_x+\mathscr{C}_{\delta^2}-1\right) p_{0}(t,x;T,y)\right|\\ &\qquad\ +
\|\lam_{1}\|_{\infty}|x-y|  \left| \left(\left( e^{\frac{\delta^2}{2}}-1
\right)\partial_x+\mathscr{C}_{\delta^2}-1\right) p_{1}(t,x;T,y)\right|.
\end{align}
and, as $\|\lam_{1}\|_{\infty}=0$ implies $\|\lam_{2}\|_{\infty}=0$, by Propositions
\ref{properr2} and \ref{properr4} there exists a positive constant $C$, only dependent on
$c,\t,M$,$\|\lambda_{1}\|_{\infty}$,$\|\lambda_{2}\|_{\infty}$, $\|a_{1}\|_{\infty}$ and
$\|a_{2}\|_{\infty}$, such that \eqref{eqbound7} holds for $N=1$ and $n=0$.
%that with Proposition \ref{properr2} gives \eqref{eqbound7} for $n=0$.
To prove the general case, we proceed by induction on $n$. First note that, by \eqref{eq:v.0.pide}
we have
\begin{align}
 |L p^{(0)}(t,x;T,y)|& %\leq |(L-L_0)p^{(0)}(t,x;T,y)|   + |L_0 p^{(0)}(t,x;T,y)|
 =|(L-L_0)p^{(0)}(t,x;T,y)|\\
\intertext{(and by \eqref{eqbound21} and the Lipschitz property of $\a,\lam$)}
 &\leq
 \|a_{1}\|_{\infty}|x-y||(\partial_{xx}-\partial_{x}) p^{(0)}(t,x;T,y)|\\
 &\quad\ + \|\lam_{1}\|_{\infty}|x-y|  \left| \left(\left( e^{\frac{\delta^2}{2}}-1 \right)\partial_x+\mathscr{C}_{\delta^2}-1\right)
p^{(0)}(t,x;T,y)\right| \intertext{(and by applying Lemmas \ref{lemestim1}, \ref{lemestim3},
\ref{lemestim2} and \ref{lemestim8} with $\eta=0,1$)}\label{eqbound8}
 &\leq C_{0}\left(\frac{1}{\sqrt{T-t}}+ \|\lam_{1}\|_{\infty}\mathscr{C}_{cM} \right) \bar{\G}^{cM,cM,cM}(t,x;T,y),
\end{align}
for any $x,y\in\R$ and $t,T\in\R$ with $0<T-t\le \t$, and where $C_{0}$ is a positive constant
only dependent on $c,\t,M,\|\lam_{1}\|_{\infty}$ and $\|a_{1}\|_{\infty}$. Assume now
\eqref{eqbound7} holds for $n\geq 0$.  Then by \eqref{eqbound9} we obtain
\begin{align}
 |Z^{(1)}_{n+1}(t,x;T,y)|&\leq \int_{t}^T \int_{\R} |Lp^{(0)}(t,x;s,\x)||Z^{(1)}_{n}(s,\x;T,y)|d\x ds
\intertext{(and by inductive hypothesis and by \eqref{eqbound8})} &\leq
 \frac{C^{n+1}C_{0}}{\sqrt{n!}}\int_{t}^T (T-s)^{\frac{n}{2}}\int_{\R} \left(\frac{1}{\sqrt{s-t}}+
 \|\lam_{1}\|_{\infty}\mathscr{C}_{cM} \right)\bar{\G}^{cM,cM,cM}(t,x;s,\x)\\
 &\quad\cdot\left(1+\|\lambda_{1}\|_{\infty}\mathscr{C}_{cM}^{n+1}\right)
 \bar{\G}^{cM,cM,cM}(s,\x;T,y)d\x ds.
\end{align}
Now, by the semigroup property \eqref{eqbound20}, and by the fact that\footnote{Here $\G_E$ represents the Euler Gamma function.}
\begin{align}
\int_{t}^T \frac{(T-s)^{\frac{n}{2}}}{\sqrt{s-t}}ds=\frac{\sqrt{\pi } (T-t)^{\frac{n+1}{2}}
\G_E\left(\frac{2+n}{2}\right)}{\G_E\left(\frac{3+n}{2}\right)}\leq \frac{\k
(T-t)^{\frac{n+1}{2}}}{\sqrt{n+1}} ,
\end{align}
with $\k=\sqrt{2}\pi$, we obtain
\begin{align}
 |Z^{(1)}_{n+1}(t,x;T,y)|&\leq \frac{C^{n+1}C_{0}}{\sqrt{n!}} \left(  \frac{\k
 (T-t)^{\frac{n+1}{2}}}{\sqrt{n+1}}\left(1+ \|\lam_{1}\|_{\infty}\mathscr{C}_{cM}^{n+1}
 \right)\right) \bar{\G}^{cM,cM,cM}(t,x;T,y)\\ \label{eqbound19} &\quad +\frac{C^{n+1}C_{0}}{\sqrt{n!}}
 \left(\frac{2 (T-t)^{\frac{n+2}{2}}}{n+2}   \|\lam_{1}\|_{\infty}\left(  \mathscr{C}_{cM}+
 \|\lam_{1}\|_{\infty}\mathscr{C}_{cM}^{n+2} \right)\right) \bar{\G}^{cM,cM,cM}(t,x;T,y).
\end{align}
Now, by Lemma \ref{lemestim9} we have
\begin{align}
\mathscr{C}_{cM}^{n+1}\bar{\G}^{cM,cM,cM}(t,x;T,y)&\leq 2\,
\mathscr{C}_{cM}^{n+2}\bar{\G}^{cM,cM,cM}(t,x;T,y),\\
\mathscr{C}_{cM}\bar{\G}^{cM,cM,cM}(t,x;T,y)&\leq \sqrt{n+2}\,
\mathscr{C}_{cM}^{n+2}\bar{\G}^{cM,cM,cM}(t,x;T,y) .
\end{align}
{Inserting the above results} into \eqref{eqbound19} we obtain
\begin{align}
 |Z^{(1)}_{n+1}(t,x;T,y)|&\leq \frac{C^{n+1}C_{0}}{\sqrt{n!}} \frac{(T-t)^{\frac{n+1}{2}}}{\sqrt{n+1}}\left(
 \k +2\|\lam_{1}\|_{\infty}\left(\k+\sqrt{\t}(1+\|\lam_{1}\|_{\infty}) \right)
 \mathscr{C}_{cM}^{n+2}  \right) \bar{\G}^{cM,cM,cM}(t,x;T,y)\\ &\leq
 \frac{C^{n+1}C_{1}(T-t)^{\frac{n+1}{2}}}{\sqrt{(n+1)!}}\left(  1+\|\lam_{1}\|_{\infty}
 \mathscr{C}_{cM}^{n+2} \right)\bar{\G}^{cM,cM,cM}(t,x;T,y),
\end{align}
where
\begin{align}
 C_{1}=2C_{0}\left(  \k+\sqrt{\t}(1+\|\lam_{1}\|_{\infty})\right).
\end{align}
Now, without loss of generality  we can assume $C_{1}\leq C$, and thus we obtain \eqref{eqbound7} for %$N=1$ and
$n+1$.
%\blu{I do not understand this:}\red{unless to increase $C$} \blu{Does this mean: ``we can find $C$ large enough such that $C_1\leq C$} \red{Yes, that is basically the meaning, but $C$ is already fixed so we can't ``find $C$ large enough such that $C_1\leq C$''. The correct way would be repeating the induction procedure replacing $C$ with $\max(C,C_1)$. For brevity I'd just say:}

The proof for $N>1$ is based on the same arguments.  However, in the general case the
technical details become significantly more complicated. In practice, proceeding by induction, one can extend Propositions \ref{properr2} and Proposition
\ref{properr4} to a general $n\in\N$. Eventually, after proving the identity
\begin{align}\label{eq:identity_Lp}
L p^{(N)}(t,x;T,y)=\sum_{n=0}^N (L-L_n)p_{N-n}(t,x;T,y)%+(L-L_0)v^{(1)}(t,x;T,y)
,
\end{align}
one will be able to prove the estimate \eqref{eqbound7} on $\big|Z^{(N)}_n(t,x;T,y) \big|$ for a generic $N\geq 1$. Finally, the case $N=0$ is simpler because the identity \eqref{eq:identity_Lp} simply reduces to
\begin{align}\label{eq:identity_Lp_bis}
L p^{(0)}(t,x;T,y)=(L-L_0)p_{0}(t,x;T,y)%+(L-L_0)v^{(1)}(t,x;T,y)
,
\end{align}
and the proof becomes a simple application of Lemmas \ref{lemestim1}-\ref{lemestim9}.
\end{proof}

\subsection{Discussion on the difference with respect to the diffusion case}\label{sec:difference_diffusion}
It has been proved by \cite{RigaPagliaraniPascucci} that, in the purely diffusive case (i.e $\lambda \equiv 0$), error bounds analogous to  \eqref{eqbound10} hold with
 $$g_{N}(s)=\Oc \left(s^{\frac{N+1}{2}}\right),\quad \textrm{as } s \to 0^+.$$
In other words, the rate of convergence of the $N$-th order approximation as $t \to T^-$ is proportional to $(T-t)^{\frac{N+1}{2}}$. The expansion $\sum_n p_n(t,x;T,y)$ is thus \emph{asymptotically convergent in $T-t$}. On the other hand, Theorem \ref{t1} shows that, when considering non-null L\'evy measures, the rate of convergence do not improve for $N$ greater than $1$.

The reasons for this discrepancy can be found in the different asymptotic behaviors of the leading term $p_0(t,x;T,y)=\G^{a\left(y\right),\delta^2,\lam(y)}(t,x;T,y)
$ in the fundamental solution expansion of $L$, with and without jumps.
%Indeed, when  the operator is strictly differential, i.e. $\lambda\equiv 0$, the leading term reduces to
%$$
%p_0(t,x;T,y)=\G^{a\left(y\right),0,0}(t,x;T,y)=\G_0^{a\left(y\right),0,0}(t,x;T,y),
%$$
%which is the Gaussian fundamental solution of a \emph{heat-type} operator.
Indeed, while the short-time asymptotic behavior at the pole $x=y$ does not change whether $\lambda\equiv 0$ or not, namely $p_0(t,x;T,x)\sim\frac{1}{\sqrt{T-t}}$ as $T-t$ goes to $0$, the asymptotic behavior away from the pole radically changes when passing from the purely diffusion to the jump-diffusion case. In particular, by \eqref{eq:proof_Gamma}-\eqref{eq:proof_Gamma_n}-\eqref{e31} it is clear that in general $p_0(t,x;T,y)\sim T-t$ as $T-t$ goes to $0$. On the other hand, in the particular case
 of $L$ being strictly differential, i.e. $\lambda\equiv 0$, the leading term reduces to
$$
p_0(t,x;T,y)=\G^{a\left(y\right),0,0}(t,x;T,y)=\G_0^{a\left(y\right),0,0}(t,x;T,y),
$$
which is the Gaussian fundamental solution of a \emph{heat-type} operator, and thus tends to $0$ exponentially as $T-t$ goes to $0$. For this reason, the differential version of Lemma \ref{lemestim8} becomes
\begin{equation}
|x-y| \G^{a\left(y\right),0,0}(t,x;T,y) \leq C \sqrt{T-t}\,\G^{c a\left(y\right),0,0}(t,x;T,y),
\end{equation}
as it is also clear by Lemma \ref{lemestim2} with $n=0$. Due to this fact, in the purely diffusive case, higher order polynomials of the kind $(x-y)^{N+1}$ arising from the $N$-th order Taylor expansion of the operator $L$, allow to gain an accuracy factor equal to $(T-t)^{\frac{N+1}{2}}$. On the contrary, in the jump-diffusion case, such polynomials can be only used to cancel out the negative powers of the time introduced by the space derivatives, by combining Lemma \ref{lemestim3} and Lemma \ref{lemestim2}.

\subsection{Pointwise estimates}\label{sec:estimates}\label{sec:pointwise_estimates}

In the rest of the section, we will always assume that
\begin{align}\label{e30}
  M^{-1}\leq \a,\theta \leq M,\qquad 0\le\l\leq M.
\end{align}
Even if not explicitly stated, all the constants appearing in the estimates \eqref{e20},
\eqref{e23}, \eqref{e24}, \eqref{e40}, \eqref{e22a} and \eqref{eqbound15} of the following lemmas
will depend also on $M$.
\begin{lemma}\label{lemestim1}
For any $T>0$ and $c>1$ there exists a positive constant $C$ such that\footnote{Here
$\mathscr{C}_{\theta}^0$ denotes the identity operator.}
  \begin{align}\label{e20}
  \mathscr{C}_{\theta}^\eta\G^{\a,\theta,\l}(t,x)\leq C\, \mathscr{C}_{c M}^\eta \bar{\G}^{cM,cM,cM}(t,x),
  \end{align}
for any $t\in (0,T]$, $x\in\R$ and $\eta\in\mathbb{N}_0$.
\end{lemma}

\begin{proof}
For any $n\geq 0$ we have
  $$
  \G_n^{\a,\theta,\l}(t,x)
  %\leq \frac{1}{\sqrt{2\pi m(t+n)}}
  %e^{-\frac{\left(x-M\left(\frac{1}{2}- e^{\frac{cM}{2}}+1\right)t\right)^2} {2cM(t+n)}} %\max_{s\leq \tau,\ z\in\R}
  \leq\sqrt{c} M q_n(t,x)\bar{\G}_n^{c M,c M}(t,x),
  $$
where
  $$
  q_n(t,x)=\exp\left( -\frac{\left(x-\left(\frac{\a}{2}+\l  e^{\frac{\theta}{2}}-\l  \right)t\right)^2}
  {2(\a t+n\theta)}  +\frac{x^2} {2cM(t+n )}\right).
  $$
A direct computation shows that
\begin{align*}%\label{}
  \max_{x\in \R} q_n(t,x)&= \exp\left(\frac{s^2 \left(\alpha +2 \left(e^{\frac{\theta}{2}}-1\right) \l \right)^2}{8 \left(c M (n+s)-s \alpha -n \delta ^2\right)}\right) \leq \exp\left(\frac{T \left(\alpha +2 \left(e^{\frac{\th}{2}}-1\right) \l \right)^2}{8 (cM-\a)}\right) ,
\end{align*}
for any $t\in(0,T]$, $n\ge 0$ and $\a,\theta,\l$ in \eqref{e30}. Then the thesis is a
straightforward consequence of the fact that $q_n(t,x)$ is bounded on $(0,T]\times \R$, uniformly
with respect to $n\ge 0$ and $\a,\theta,\l$ in \eqref{e30}. %Eventually it is enough to observe that
%  $$
%  \frac{\sqrt{ M t+n M}}{\sqrt{m t+n m}}\leq \left(\frac{M}{m}\right)^{\frac{1}{2}}\qquad \forall n\ge 0,\ t>0.
%  $$
\end{proof}
\begin{lemma}\label{lemestim3}
For any $T>0$, $k\in\N$ and $c>1$, there exists a positive constant $C$ such that
\begin{align}\label{e23}
  \left|\partial_{x}^k\G_n^{\a,\theta,\l}(t,x)\right|\leq \frac{C}{(\a t+n\theta)^{k/2}}
  \G_n^{c \a,c \theta,\l}(t,x) ,
\end{align}
for any $x\in\R$, $t\in]0,T]$ and $n\in\N_0$.
\end{lemma}
\begin{proof}
For any $k\ge 1$ we have
  $$
  \partial_{x}^k\G_n^{\a,\theta,\l}(t,x)=
  \frac{1}{(\a t+n\theta)^{k/2}}\G_n^{\a,\theta,\l}(t,x)p_k\left( \frac{x-\left(  \frac{\a}{2} +\l  e^{\frac{\theta}{2}}-\l  \right) t}{\sqrt{\a t+n\theta}}
  \right),
  $$
where $p_k$ is a polynomial of degree $k$. To prove the Lemma we will show that there exists a
positive constant $C$, which depends only on $m,M,T,c$ and $k$, such that
%\begin{enumerate}
%\item[(a)]
%  \begin{align}
\begin{align}\label{e21bb}
 \left(\frac{\left|x-\left(  \frac{\a}{2} +\l  e^{\frac{\theta}{2}}-\l  \right)
 t\right|}{\sqrt{\a t+n\theta}}\right)^{j}\G_n^{\a,\theta,\l}(t,x)&\leq C\, \G_n^{c \a,c
 \theta,\l}(t,x),\qquad j\le k.
\end{align}
Proceeding as above, we set
  $$
  \left(\frac{\left|x-\left(  \frac{\a}{2} +\l  e^{\frac{\theta}{2}}-\l  \right)
 t\right|}{\sqrt{\a t+n\theta}}\right)^j
 \G_n^{\a,\theta,\l}(t,x)=\G_n^{c\a,c\theta,\l}(t,x)q_{n,j}(t,x),
  $$
where
  $$
  q_{n,j}(t,x)=\left(\frac{\left|x-\left(  \frac{\a}{2} +\l  e^{\frac{\theta}{2}}-\l  \right)
 t\right|}{\sqrt{\a t+n\theta}}\right)^j
   \exp\(-\frac{\left(x-\left(\frac{\a}{2}  +\l  e^{\frac{\theta}{2}}-\l  \right)t\right)^2}
   {2(\a t+n \theta)}+\frac{\left(x-\left(\frac{c\a}{2} +\l  e^{\frac{c\theta}{2}}-\l \right)t\right)^2}
    {2(c\a t+n c\theta)}\).
  $$
Then the thesis follows from the boundedness of $q_{n,j}$ on $(0,T]\times\R$, uniformly with
respect to $n\ge 0$ and $\a,\theta,\l$ in \eqref{e30}. Indeed the maximum of $q_{n,j}$ can be
computed explicitly and we have
  $$\lim_{n\to\infty}\left(\max_{x\in\R, \ t\in]0,T]}q_{n,j}(t,x)\right)=\left(\frac{c j }{(c-1)e}
  \right)^{\frac{j}{2}}.$$
\end{proof}

\begin{lemma}\label{lemestim5}
For any $T>0$ and $\eta\in\N$, there exists a positive constant $C$ such that
  \begin{align}\label{e24}
 % a)\hspace{20pt}
 \l t\, \mathscr{C}_{\theta}^\eta \G^{\a,\theta,\l}(t,x)\leq C\, \G^{\a,2(\eta+1) \theta,\l}(t,x)\\
 %\label{e25}
%  b)\hspace{31pt}\mathscr{C}_{\theta}^\eta \G^{\a,\theta,\l}(t,x)&\leq C\, \mathscr{C}_{\theta}\G^{\a,2(\eta+1) \theta,\l}(t,x)
  \end{align}
for any $t\in(0,T]$ and $x\in\R$.
%for any $m\leq \a,\theta \leq M$, $0\leq\l\leq M$, $z\in\R$ and $0<s\leq\t$.
\end{lemma}
\begin{proof}
We first prove there exists a constant $C_{0}$, which depends only on $m,M,T$ and $\eta$, such that
\begin{align}\label{eqestim2}
\G_{n+\eta}^{\a,\theta,\l}(t,x)& \leq C_{0}\, \G_{n}^{\a,2(\eta+1)\theta,\l}(t,x),\\ \label{eqestim3}
\G_{\eta}^{\a,\theta,\l}(t,x)& \leq C_{0}\, \G_{1}^{\a,2(\eta+1)\theta,\l}(t,x),
\end{align}
for any $t\in]0,T]$, $x\in\R$, $n\ge1$ and $\a,\theta,\l$ in \eqref{e30}. To prove \eqref{eqestim2} we observe that %, for any $n\geq 1$,we have
 $$
  \G_{n+\eta}^{\a,\theta,\l}(t,x)\leq \frac{1}{\sqrt{2\pi (\a t+(n+\eta)\theta)}} \exp\left(-\frac{\left(x-\left(\frac{\a}{2}
  +\l  e^{(\eta+1)\theta}-\l \right)t\right)^2} {2(\a t+2 n(\eta+1)\theta)}\right) %\max_{s\leq \tau,\ z\in\R}
  q_n(t,x),
 $$
where
 $$q_n(t,x)=\exp\left(-\frac{\left(x-\left(\frac{\a}{2}  +\l  e^{\frac{\theta}{2}}-\l
 \right)t\right)^2} {2(\a t+(n+\eta)\theta)}+\frac{\left(x-\left(\frac{\a}{2} +\l
 e^{(\eta+1)\theta}-\l \right)t\right)^2} {2(\a t+2 n(\eta+1)\theta)}\right). $$
Now it is easy to check that
 $$\max_{x\in\R}q_n(t,x)=\exp\left(\frac{\left(e^{(1+\eta) \theta}-e^{\frac{\theta}{2}}\right)^2 t^2
 \l ^2}{2 (n-\eta+2 n \eta) \theta} \right)\leq \exp\left(\frac{\left( e^{(1+\eta) \theta}-e^{\frac{\theta}{2}}\right)^2 t^2 \l ^2}{2 \eta \theta} \right). $$
for any $t\geq0$. Thus $q_n$ is bounded on $(0,T]\times \R$, uniformly with respect to $n\in\N$
and $\a,\theta,\l$ in \eqref{e30}. {To see the above bound, simply} observe that
  $$
  \frac{\sqrt{\alpha  t +2n(\eta+1)\theta}}{\sqrt{\alpha  t +(\eta+n)\theta}}\leq \sqrt{2(\eta+1)}.
  $$
The proof of \eqref{eqestim3} is completely analogous. Finally, by
\eqref{eqestim2}-\eqref{eqestim3} we have
\begin{align*}
 \l t\,
 \mathscr{C}_{\theta}^\eta \G^{\a,\theta,\l}(t,x)& = e^{-\l t}\l t\,
 \G_\eta^{\a,\theta,\l}(t,x)+\l t\,
 e^{-\l t} \sum_{n=1}^{\infty}\frac{(\l t)^n}{n!}
 \G_{n+\eta}^{\a,\theta,\l}(t,x)\\
 & \leq C_{0}\left( e^{-\l t} \l t\,
 \G_1^{\a,2(\eta+1)\theta,\l}(t,x)
 +\l t\,
 e^{-\l t} \sum_{n=1}^{\infty}\frac{(\l t)^n}{n!}
 \G_{n}^{\a,2(\eta+1)\theta,\l}(t,x) \right)
 \\ &
 \leq C_{0} (1+M T)\
 \G^{\a,2(\eta+1)\theta,\l}(t,x).
\end{align*}
\end{proof}

\begin{lemma}\label{lemestim7}
For any $T>0$ and $\eta\in\N$ with $\eta\geq 2$, there exists a positive constant $C$ such that
  \begin{align}\label{e40}
 \mathscr{C}_{\theta}^\eta \G^{\a,\theta,\l}(t,x)\leq C\, \mathscr{C}_{2\eta\theta}\G^{\a,2\eta \theta,\l}(t,x)\\
  \end{align}
for any $t\in(0,T]$ and $x\in\R$.
%for any $m\leq \a,\theta \leq M$, $0\leq\l\leq M$, $z\in\R$ and $0<s\leq\t$.
\end{lemma}
\begin{proof}
By \eqref{eqestim2}
\begin{align}
 \mathscr{C}_{\theta}^\eta \G^{\a,\theta,\l}(t,x)=e^{-\l t}\sum_{n=0}^{\infty}\frac{(\l t)^n}{n!}\G^{\a,\theta,\l}_{n+1+(\eta-1)}(t,x)\leq C\, e^{-\l t}\sum_{n=0}^{\infty}\frac{(\l t)^n}{n!}\G^{\a,2\eta\theta,\l}_{n+1}(t,x)= C\, \mathscr{C}_{2\eta\theta}\G^{\a,2\eta \theta,\l}(t,x).
\end{align}
\end{proof}

\begin{lemma}\label{lemestim2}
For any $T>0$, $\eta\in \N$ and $c>1$, there exists a positive constant $C$ such that
%\begin{enumerate}
%\item[(a)]
%  \begin{align}
\begin{align} \label{e22a}
 \left(\frac{\left|x\right|}{\sqrt{\a
 t+n\theta}}\right)^\eta\G_n^{\a,\theta,\l}(t,x)&\leq C \G_n^{c \a,c \theta,\l}(t,x),
\end{align}
%  \end{align}
%\end{enumerate}
for any $x\in\R$, $t\in(0,T]$ and $n\in\N_0$.
\end{lemma}

%We first treat the case $n\geq 1$. A direct computation shows that
%  $$\lim_{n\to\infty}\left(\max_{z\in\R}q_n(t,x)\right)=\left(\frac{c N }{(c-1)e}\right)^{N/2},$$
%and thus $q_n(t,x)$ is bounded as a function of $(n,z)$ for any $(s,\a,\theta,\l)\in
%(0,\t]\times [m,M]^2\times [0,M]$. Therefore, uniform continuity of $q_n$ with respect to
%$(n,s,z,\a,\theta,\l)$ implies boundedness on $\N^+\times[0,\t]\times\R\times [m,M]^2\times
%[0,M]$.\\ In the case $n=0$ we can extend $g_0(\cdot,z)$ with continuity on $[0,\t]$ setting
%$g_0(0,z)\equiv 0$. Note that $g_0(s,\cdot)$ has a maximum for any $(s,\a,\theta,\l)\in
%[0,\t]\times [m,M]^2\times [0,M]$, and thus, as $g_0$ is uniformly continuous with respect to
%$(s,z,\a,\theta,\l)$, it is also bounded on $[0,\t]\times\R\times [m,M]^2\times [0,M]$.
%\\
%\end{proof}
%
%\begin{lemma}\label{lemestim2bis}
%For any $0<m\leq M$, $\t>0$, $N\geq 1$ and $c>1$, there exists a positive constant $C_{m,M,\t,N,c}$ such that
%  \begin{align}
%  \left(\frac{\left|z\right|}{\sqrt{\a s+n\theta}}\right)^N\G_n^{\a,\theta,\l}(t,x)\leq C_{m,M,\t,N,c} \G_n^{c \a,c \theta,\l}(t,x)
%  \end{align}
%for any $m\leq \a,\theta \leq M$, $0\leq\l\leq M$, $z\in\R$, $0<s\leq\t$ and $n\in\N$.
%\end{lemma}
%\begin{proof}
\begin{proof} We first show that there exist three constants $C_1=C_{1}(M,T,\eta,c)$, $C_{2}=C_{2}(\eta,c)$ and
$C_{3}=C_{3}(M,T,\eta,c)$ such that
\begin{align}\label{eqestim4}
  e^{-\frac{\left(x-\left(\frac{\a}{2}+\l  e^{\frac{\theta}{2}}-\l  \right)t\right)^2} {2(\a
  t+n\theta)}}&\leq C_{1} e^{-\frac{x^2} {2c^{1/3}(\a T+n\theta)}},\\
 \label{eqestim5}
  \left(\frac{\left|x\right|}{\sqrt{\a t+n\theta}}\right)^\eta  e^{-\frac{x^2}
  {2c^{1/3}(\a T+n\theta)}}&\leq  C_{2} e^{-\frac{x^2} {2c^{2/3}(\a T+n\theta)}},\\
 \label{eqestim6}
  e^{-\frac{x^2} {2c^{2/3}(\a T+n\theta)}}&\leq C_{3}
  e^{-\frac{\left(x-\left(\frac{c\a}{2}+\l  e^{\frac{c\theta}{2}}-\l  \right)t\right)^2} {2c(\a
  t+n\theta)}} ,
\end{align}
for any $x\in\R$, $t\in(0,T]$ and $n\ge 0$. In order to prove \eqref{eqestim4} we consider
%$$
%\G_n^{\a,\theta,\l}(t,x)=\frac{1}{\sqrt{2\pi (\a s+n\theta)}}e^{-\frac{z^2} {2c^{1/3}(\a s+n\theta)}} q_n(t,x),
%$$
%where
 $$ q_n(t,x)=\exp\left(-\frac{\left(x-\left(\frac{\a}{2}+\l  e^{\frac{\theta}{2}}-\l
 \right)t\right)^2} {2(\a t+n\theta)}+\frac{x^2} {2c^{1/3}(\a t+n\theta)}\right) , $$
and show that
 $$ \max_{x\in\R}q_n(t,x)=\exp\left(\frac{\left(  \frac{\a}{2}+\l
 e^{\frac{\theta}{2}}-\l \right)^2 t^2}{2 (c^{1/3}-1) \left(t \alpha +n \theta\right)}\right)\leq
 \exp\left(\frac{\left( \frac{\a}{2}+\l e^{\frac{\theta}{2}}-\l \right)^2 T}{2 (c^{1/3}-1)}\right) , $$
for any $t\in(0,T]$. Thus $q_n$ is bounded on $(0,T]\times \R$, uniformly in $n\ge0$ and
$\a,\theta,\l$ in \eqref{e30}. The proof of \eqref{eqestim6} is completely analogous.  Equation
\eqref{eqestim5} comes directly by setting
 $$C_{2}=\max_{a\in\R^+} \left(a^\eta e^{-\frac{a^2} {2c^{1/3}}+\frac{a^2} {2c^{2/3}}}\right)=e^{-\frac{\eta}{2}}
\left(\frac{c^{1/3} \sqrt{\eta}}{\sqrt{c^{1/3}-1}}\right)^{\eta}. $$ Now, by \eqref{eqestim4} we have
\begin{align*}
 \left(\frac{\left|x\right|}{\sqrt{\a t+n\theta}}\right)^\eta\G_n^{\a,\theta,\l}(t,x)&\leq C_{1}
 \left(\frac{\left|x\right|}{\sqrt{\a t+n\theta}}\right)^\eta \frac{e^{-\frac{x^2}
 {2c^{1/3}(\a T+n\theta)}}}{\sqrt{2\pi (\a T+n\theta)}}
\intertext{(by \eqref{eqestim5})}
 &\leq C_{1}C_{2}\frac{e^{-\frac{x^2} {2c^{2/3}(\a T+n\theta)}}}{\sqrt{2\pi (\a T+n\theta)}}
\intertext{(by \eqref{eqestim6})}
 &\leq C_{1}C_{2}C_{3}  \sqrt{c}\ \G_n^{c \a,c \theta,\l}(t,x).
\end{align*}
\end{proof}

%
%\begin{lemma}\label{lemestim6}
%For any $T>0$ and $c>1$ there exists a positive constant $C$ such that
%\begin{align}\label{eqbound12}
% |x|\G^{\a,\theta,\l}(t,x)\leq C \left(\G^{c \a, c \theta,\l}(t,x)+\G^{c \a, 4c \theta,\l}(t,x)\right),
%\end{align}
%for any $t\in]0,T]$ and $x\in\R$.
%\end{lemma}
%\begin{proof}
%By Lemma \ref{lemestim2} there is a constant $C_{0}$, only dependent on $m,M,T$ and $c$, such that
%\begin{align*}
% |x|\G^{\a,\theta,\l}(t,x)&\leq C_{0} e^{-\l t}\sum_{n=0}^{\infty}\frac{(\l t)^n}{n!}\sqrt{\a t+n
% \theta}\ \G_n^{c\a,c\theta,\l}(t,x)\\ &\leq C_{0} \sqrt{M T}
% \G^{c\a,c\theta,\l}(t,x)+C_0\sqrt{M}e^{-\l t}\sum_{n=0}^{\infty}\frac{(\l t)^n}{n!}n
% \G_n^{c\a,c\theta,\l}(t,x)\\ &= C_{0} \sqrt{M T} \G^{c\a,c\theta,\l}(t,x)+C_0\sqrt{M}\l t\,
% \mathscr{C}_{c\theta}\G^{c\a,c\theta,\l}(t,x),
%\end{align*}
%for any $t\in]0,T]$ and $x\in\R$ and $\a,\theta,\l$ in \eqref{e30}. Therefore, by Lemma
%\ref{lemestim5} with $N=1$ we get \eqref{eqbound12}.
%\end{proof}

\begin{lemma}\label{lemestim8}
For any $T>0$, $c>1$ and $j\in\N\cup\{0\}$ there exists a positive constant $C$ such that
{\begin{align}\label{eqbound15}
 |x|\mathscr{C}^{j}_{\theta}\G^{\a,\theta,\l}(t,x)\leq C\, \left(\mathscr{C}^{j}_{c\theta}
 \G^{c \a,  c \theta,\l}(t,x)+ \mathscr{C}^{j}_{4c\theta} \G^{c \a, 4 c \theta,\l}(t,x)\right),
\end{align}}
for any $t\in(0,T]$ and $x\in\R$.
\end{lemma}
\begin{proof}
By Lemma \ref{lemestim2} there is a constant $C_{0}$, only dependent on $m,M,T$ and $c$, such that
\begin{align*}
 |x|\mathscr{C}^{j}_{\theta}\G^{\a,\theta,\l}(t,x)&\leq C_{0} e^{-\l t}\sum_{n=0}^{\infty}
 \frac{(\l t)^n}{n!}\sqrt{\a t+(n+j)\theta} \G_{n+j}^{c\a,c\theta,\l}(t,x)\\ &\leq C_{0} \sqrt{M}(\sqrt{T}+j)
 \mathscr{C}^{j}_{c\theta}\G^{c\a ,c\theta,\l}(t,x)+C_0\sqrt{M}e^{-\l t}\sum_{n=0}^{\infty}\frac{(\l t)^n}{n!}n
 \G_{n+j}^{c\a ,c\theta,\l}(t,x)\\
 &\leq C_{0} \sqrt{M}(\sqrt{T}+j)
 \mathscr{C}^{j}_{c\theta}\G^{c\a ,c\theta,\l}(t,x)+C_0 M^{\frac{3}{2}} t\, \mathscr{C}^{j+1}_{c\theta}\G^{c\a,c\theta,\l}(t,x),
\end{align*}
for any $t\in(0,T]$ and $x\in\R$ and $\a,\theta,\l$ in \eqref{e30}. Therefore, the thesis follows
from  Lemma \ref{lemestim5} for $j=0$ and from Lemma \ref{lemestim7} for $j\ge 1$.
\end{proof}

\begin{lemma}\label{lemestim9}
For any $T>0$ and $\eta,k\in\N$ we have
\begin{align}
\mathscr{C}_{\theta}^{\eta}\bar{\G}^{\a,\theta,\l}(t,x)\leq \sqrt{k+1}\,
\mathscr{C}_{\theta}^{\eta+k}\bar{\G}^{\a,\theta,\l}(t,x),\qquad t\in]0,T], \ x\in\R.
\end{align}
\end{lemma}
\begin{proof}
%For any $n\geq 0$ we have
%$$
%\G^{\a,\theta,\l}_{n+N}(t,x)=
%\frac{1}{\sqrt{\a t +(n+N)\theta}} e^{-\frac{\left(x-\left(\frac{\a}{2}-\l  e^{\frac{\theta}{2}}+\l  \right)t\right)^2}{2(\a t+(n+N+k)\theta)}}
%  \G^{\a,\theta,\l}_{n+N+k}(t,x)q_n(t,x),
%$$
%with
%$$
%q_n(t,x)=\frac{\sqrt{\a t +(n+N+k)\theta}}{\sqrt{\a t +(n+N)\theta}}\exp\left(  -\frac{\left(x-
%  \left(\frac{\a}{2}-\l  e^{\frac{\theta}{2}}+\l  \right)t\right)^2}
%  {2(\a t+(n+N)\theta)}+\frac{\left(x-
%  \left(\frac{\a}{2}-\l  e^{\frac{\theta}{2}}+\l  \right)t\right)^2}
%  {2(\a t+(n+N+k)\theta)}  \right).
%$$
A direct computation shows that
\begin{align}
  \max_{x\in \R} \frac{\bar{\G}^{\a,\theta}_{n+\eta}(t,x)}{\bar{\G}^{\a,\theta}_{n+\eta+k}(t,x)}=\frac{\sqrt{\a t +(n+\eta+k)\theta}}{\sqrt{\a t +(n+\eta)\theta}}\leq \sqrt{k+1} ,
\end{align}
for any $t\le T$, $n\ge 0$, $\eta\ge 1$ and $\a,\theta,\l$ in \eqref{e30}.  This concludes the proof.
\end{proof}

\bibliographystyle{chicago}
\bibliography{LPP-bib}

\clearpage
\begin{figure}
\centering
\begin{tabular}{ccc}
$n=1$ & $n=1$ & $n=1$ \\
\includegraphics[width=.31\textwidth,height=.182\textheight]{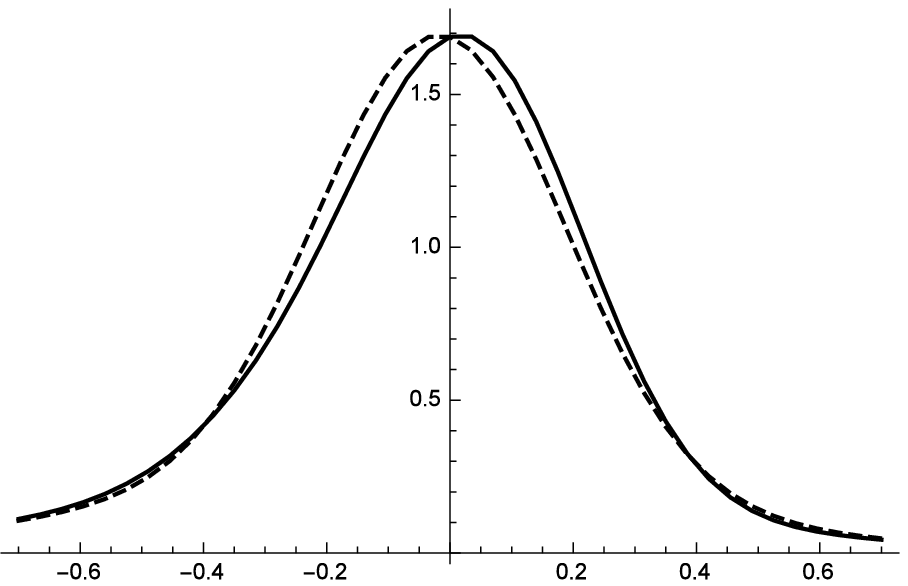} &
\includegraphics[width=.31\textwidth,height=.182\textheight]{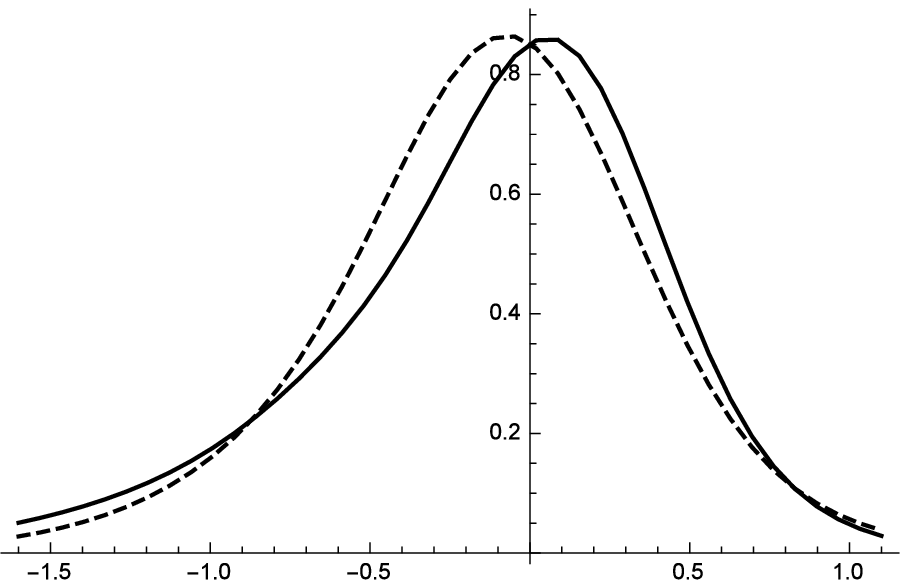} &
\includegraphics[width=.31\textwidth,height=.182\textheight]{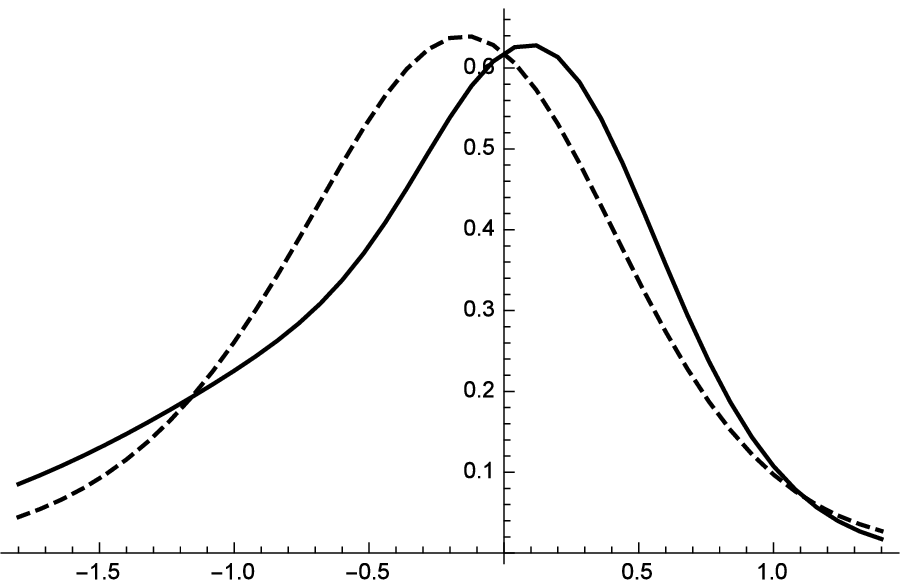} \\
$n=2$ & $n=2$ & $n=2$ \\
\includegraphics[width=.31\textwidth,height=.182\textheight]{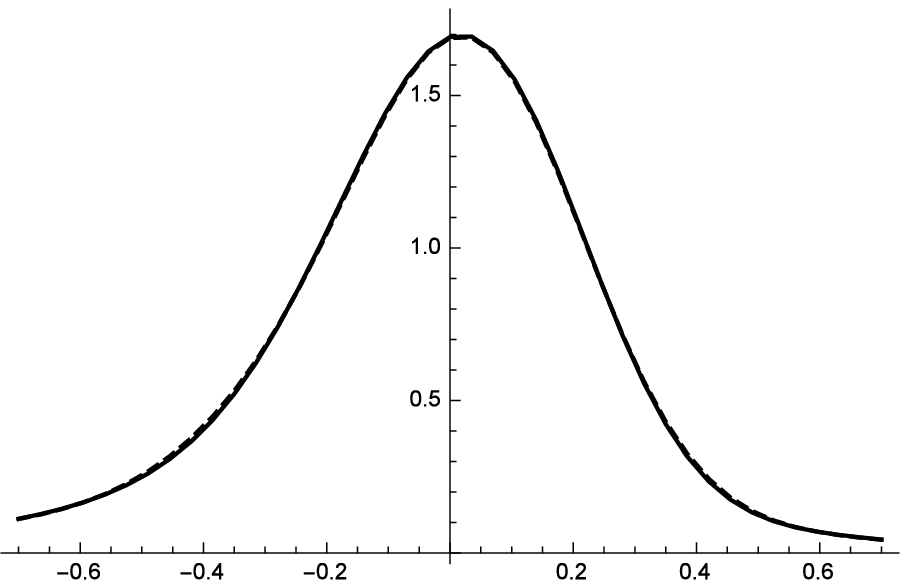} &
\includegraphics[width=.31\textwidth,height=.182\textheight]{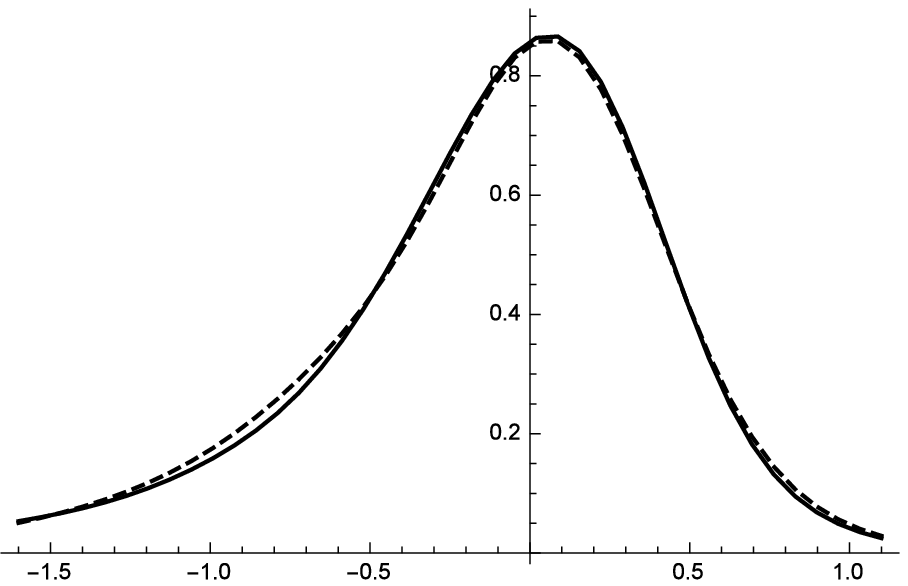} &
\includegraphics[width=.31\textwidth,height=.182\textheight]{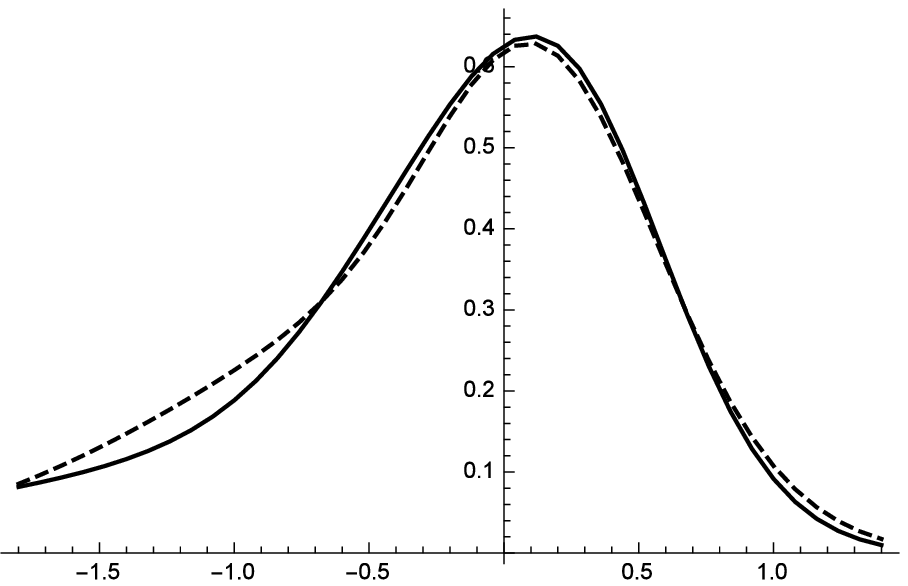} \\
$n=3$ & $n=3$ & $n=3$ \\
\includegraphics[width=.31\textwidth,height=.182\textheight]{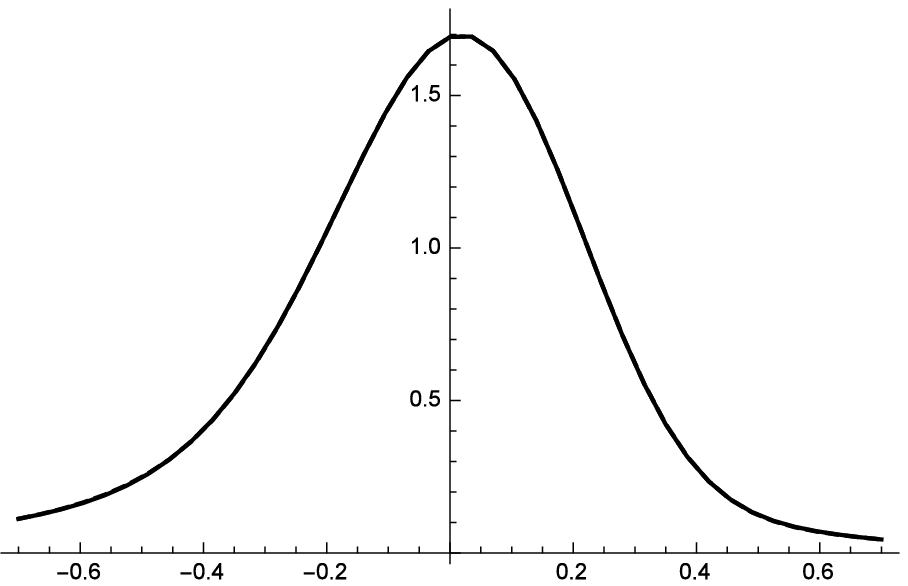} &
\includegraphics[width=.31\textwidth,height=.182\textheight]{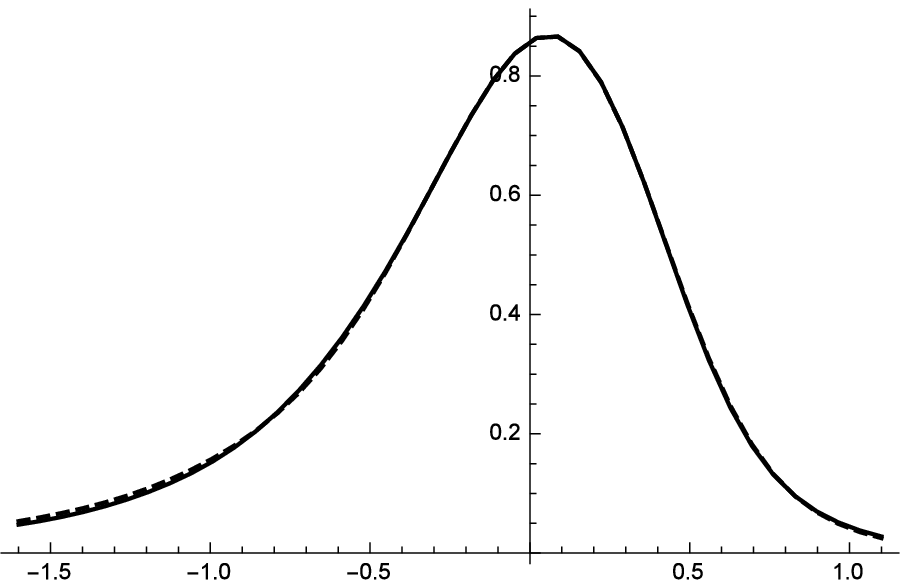} &
\includegraphics[width=.31\textwidth,height=.182\textheight]{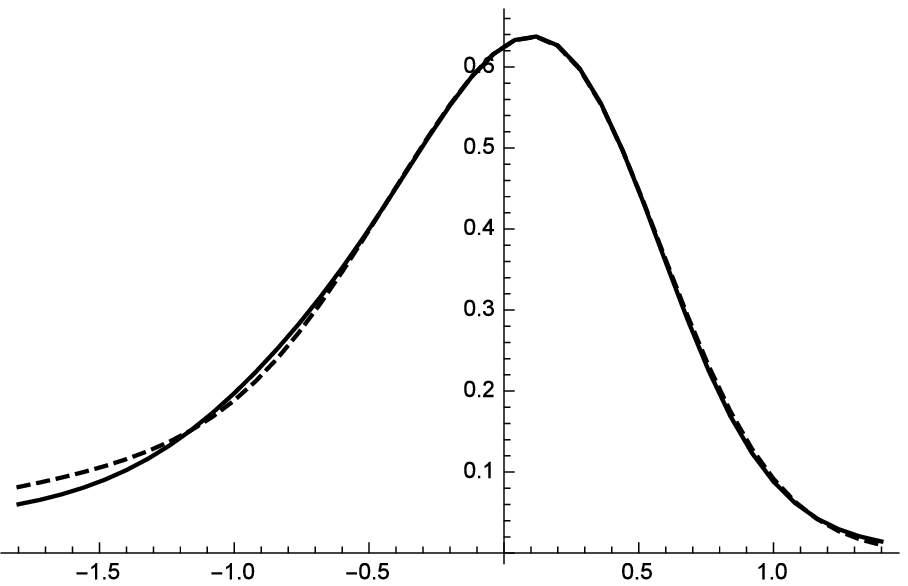} \\
$n=4$ & $n=4$ & $n=4$ \\
\includegraphics[width=.31\textwidth,height=.182\textheight]{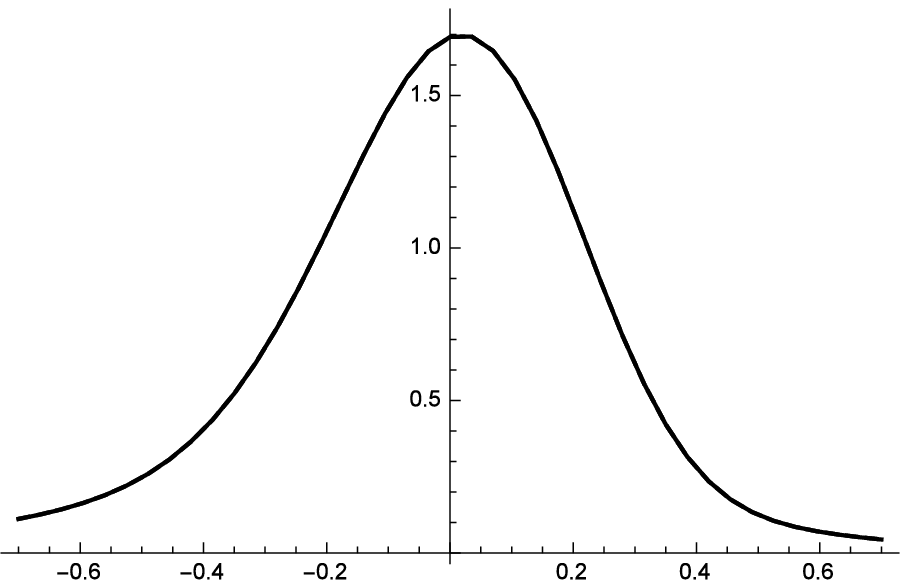} &
\includegraphics[width=.31\textwidth,height=.182\textheight]{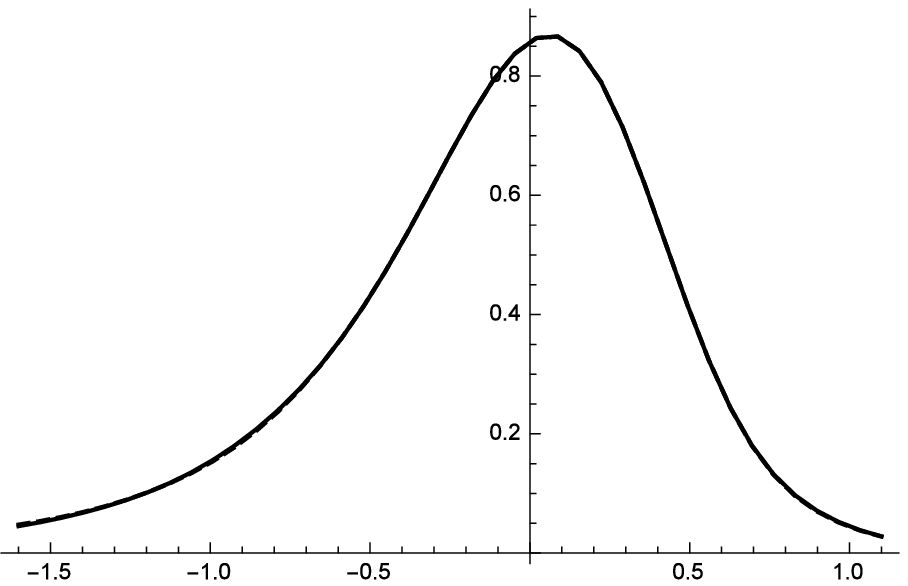} &
\includegraphics[width=.31\textwidth,height=.182\textheight]{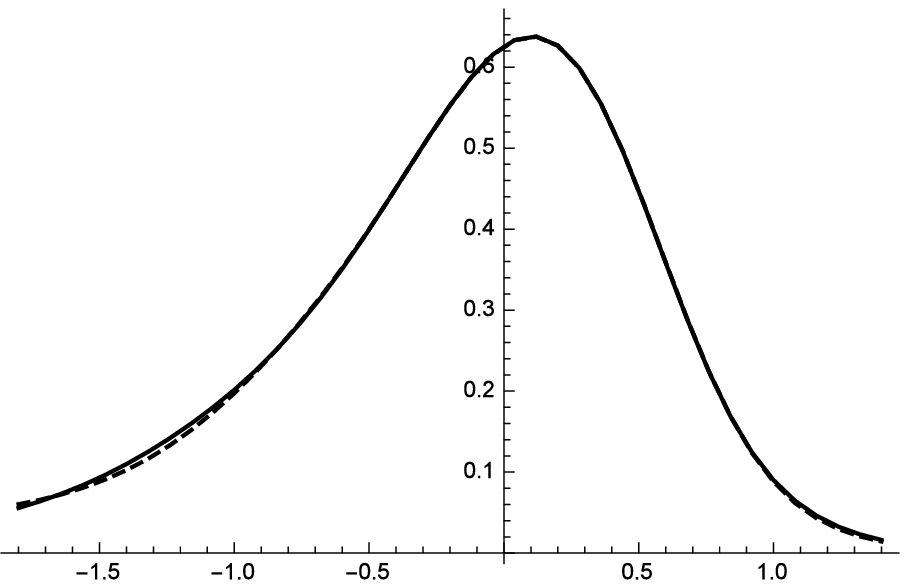} \\
$T-t=1.0$ & $T-t=3.0$ & $T-t=5.0$ \\
\end{tabular}
\caption{Using the model considered in Section \ref{sec:CEV} we plot $p^{(n)}(t,x;T,y)$
(solid black) and $p^{(n-1)}(t,x;T,y)$ (dashed black) as a function of $y$ for $n=\{1,2,3,4\}$ and
$t=\{1.0,3.0,5.0\}$ years.  For all plots we use the Taylor series expansion of Example \ref{example:Taylor}.
Note that as $n$ increases $p^{(n)}$ and $p^{(n-1)}$ become nearly indistinguishable.
Numerical values for $\sup_y|p^{(n)}(t,x;T,y) - p^{(n-1)}(t,x;T,y)|$ as well as computation times are given in Table \ref{tab:density}.
In all plots we use the parameter values are those listed in equation \eqref{eq:parameters}.}
\label{fig:density}
\end{figure}
% Mathematica2/Figure-1-1.00.nb

\begin{table}
\begin{minipage}{0.495\textwidth}
\centering
$\sup_y |p^{(n)}(t,x;T,y) - p^{(n-1)}(t,x;T,y)|$\\[0.5em]
\begin{tabular}{c|ccc}
\hline
$n$ & $T-t=1$ & $T-t=3$ & $T-t=5$ \\
\hline
1 & 0.1232 & 0.1138 & 0.1078 \\
2 & 0.0083 & 0.0160 & 0.0217 \\
3 & 0.0014 & 0.0056 & 0.0118 \\
4 & 0.0004 & 0.0028 & 0.0088 \\
\hline
\end{tabular}
\end{minipage}
\begin{minipage}{0.495\textwidth}
\centering
Computation time relative to $p^{(0)}$\\[0.5em]
\begin{tabular}{c|ccc}
\hline
$n$ & $T-t=1$ & $T-t=3$ & $T-t=5$ \\
\hline
1 & 1.14 & 1.07 & 1.04 \\
2 & 1.59 & 1.50 & 1.45 \\
3 & 2.32 & 2.28 & 2.21 \\
4 & 3.46 & 3.30 & 3.25 \\
\hline
\end{tabular}
\end{minipage}
\caption{Numerical results from Figure \ref{fig:density}. Left: We list as a function of $n$ and
$T-t$ the maximum difference between $p^{(n)}(t,x;T,y)$ and $p^{(n-1)}(t,x;T,y)$.
The supremum is taken over the range of values for $y$ shown in Figure \ref{fig:density}.
Right: We list as a function of $n$ and $T-t$
the average computation time of $p^{(n)}$ relative to $p^{(0)}$.
Relative computation times are described in the last paragraph of Section \ref{sec:CEV}.
In both tables we use the parameter values listed in equation \eqref{eq:parameters}.
}
\label{tab:density}
\end{table}
% Mathematica2/Figure-1-1.00.nb

%%%%%%%%%%%%%%%%  Comparison to Lorig %%%%%%%%%%%%%%%%%%%%%

\clearpage
\begin{figure}
\includegraphics[width=\textwidth,height=.5\textheight]{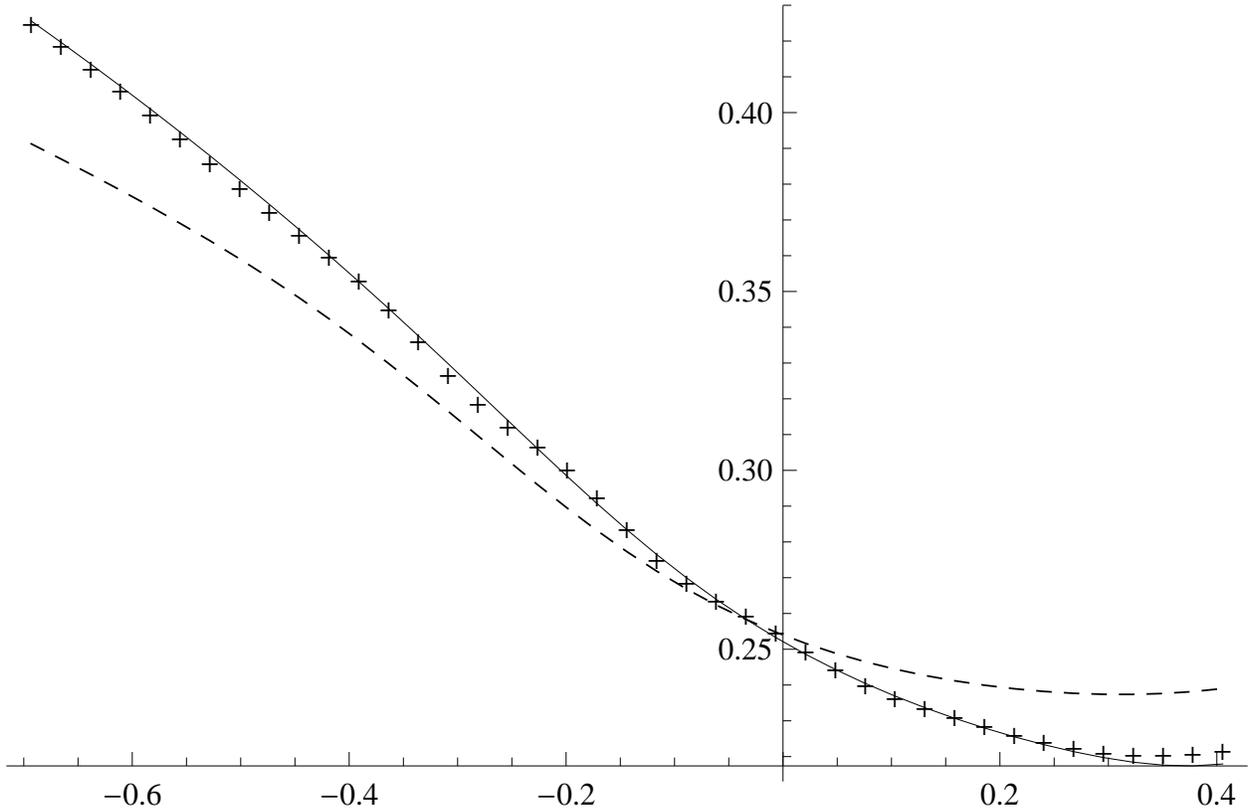}
\caption{Implied volatility (IV) is plotted as a function of $\log$-strike $k := \log K$ for the
model of Section \ref{sec:lorig}.  The dashed line corresponds to the IV induced by
$u^{(0)}(t,x)$.  The solid line corresponds to the IV induced by $u^{(2)}(t,x)$.  To compute
$u^{(i)}(t,x)$, $i\in\{0,2\}$, we use the Taylor series expansion of Example \ref{example:Taylor}  The
crosses correspond to the IV induced by the exact price, which is computed by truncating
(\ref{eq:w}) at $n=8$.
{Truncating (\ref{eq:w}) at $n=8$ ensures a high degree of accuracy since, according to \cite{lorig-jacquier-1}, the error in implied volatility encountered by truncating the series at any $n \geq 4$ is less than $10^{-4}$.}
Parameters for this plot are given in \eqref{eq:params}.} \label{fig:compare.lorig}
\end{figure}
%figure generated using Two-Point-ImpVol-Lorigb.nb

%%%%%%%%%%%%%%%%%%%%%%% Normal Inverse Gaussian   %%%%%%%%%%%%%%%%%%%%%%%%%%%

\clearpage
\begin{figure}
\centering
\includegraphics[width=\textwidth,height=0.55\textheight]{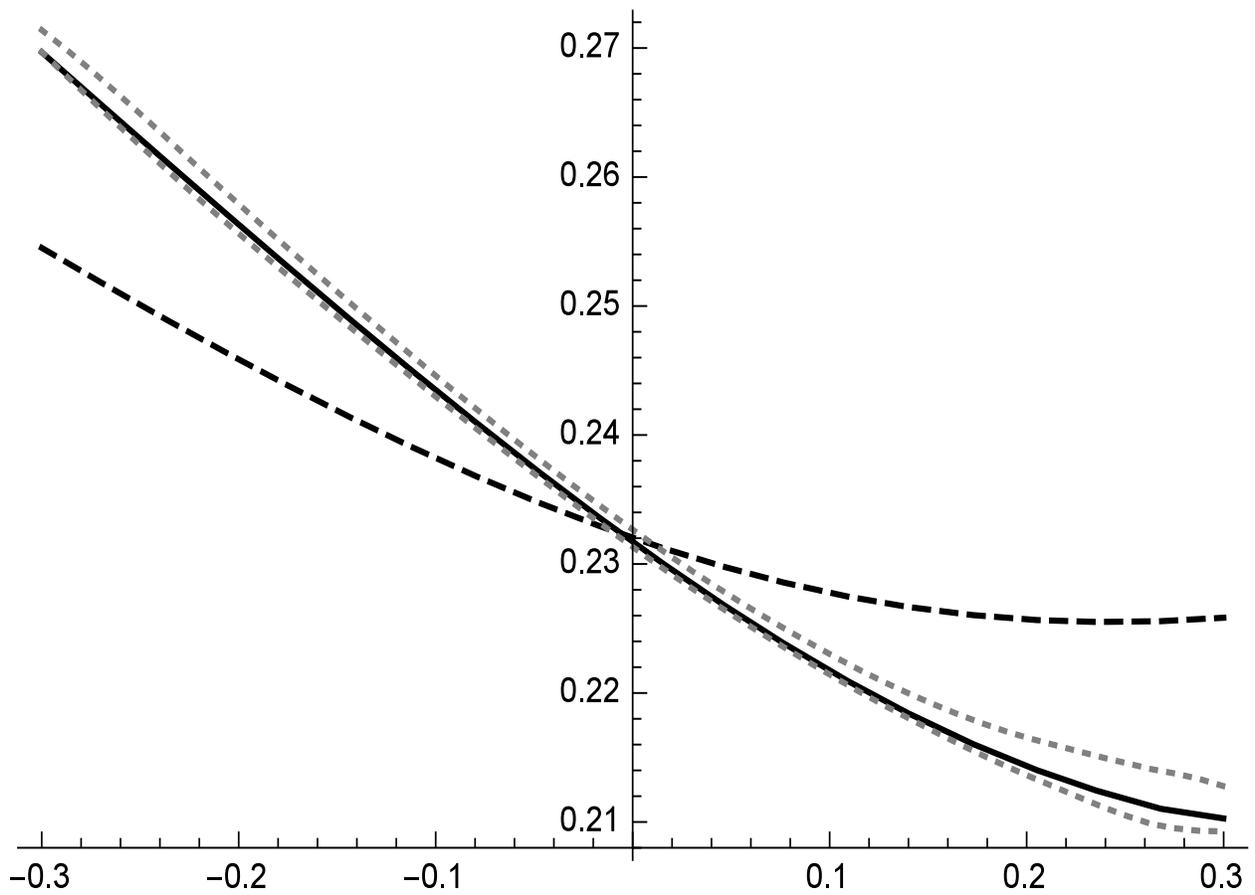}
\caption{Implied volatility (IV) is plotted as a function of $\log$-strike $k := \log K$ for the
model of Section \ref{sec:nig}.  The dashed line corresponds to the IV induced by
$u^{(0)}(t,x)$.  The solid line corresponds to the IV induced by $u^{(3)}(t,x)$.  The dotted lines  correspond to the $95\%$ confidence interval of IV resulting form a Monte Carlo simulation.  We use parameters given in equation \eqref{eq:nig.params}.}
\label{fig:iv-nig}
\end{figure}

%%%%%%%%%%%%%%%%%%%%%%% Bond Prices and Credit Spreads   %%%%%%%%%%%%%%%%%%%%%%%%%%%

\clearpage
\begin{figure}
\centering
\begin{tabular}{cc}
\includegraphics[width=.495\textwidth,height=.35\textheight]{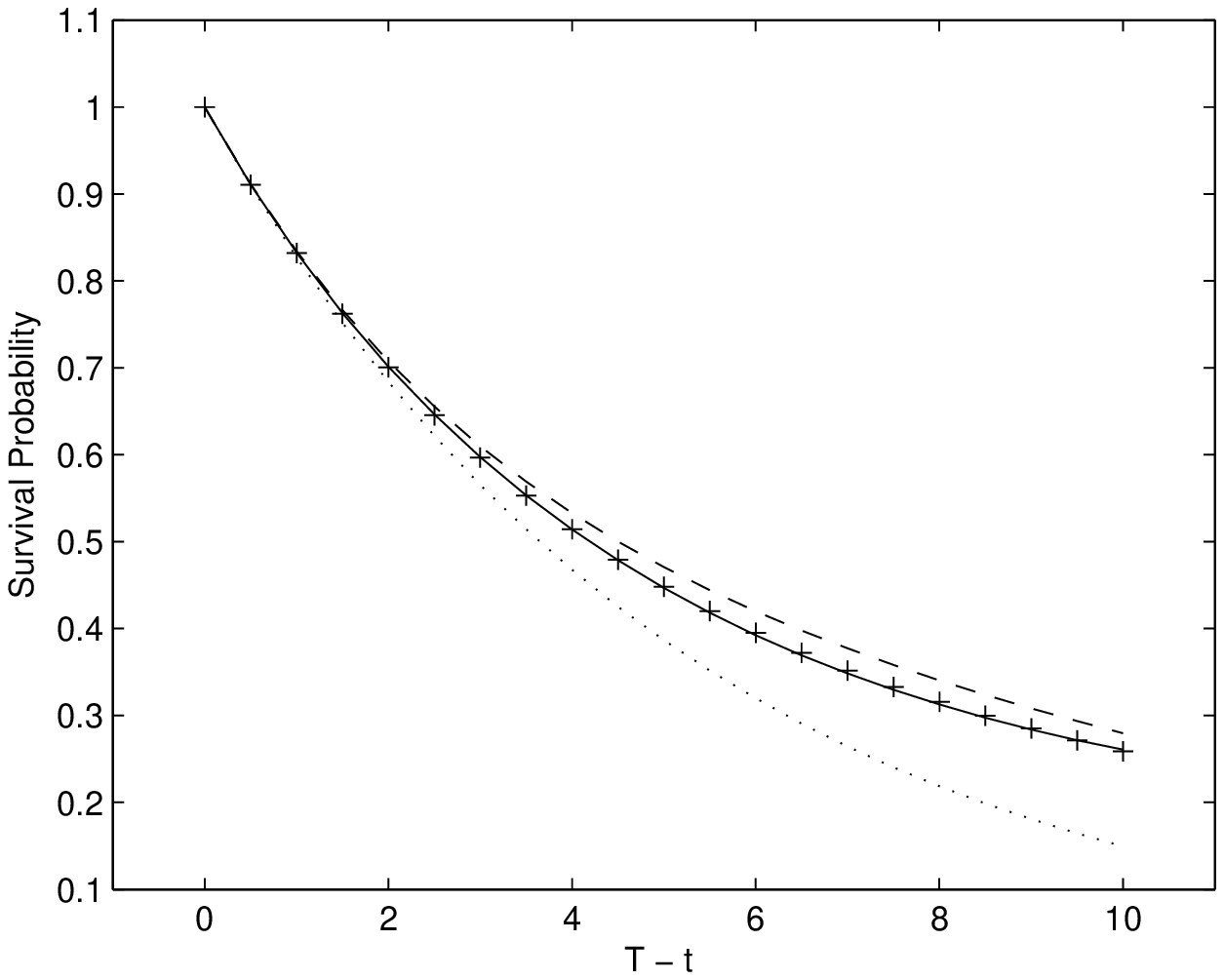} &
\includegraphics[width=.495\textwidth,height=.35\textheight]{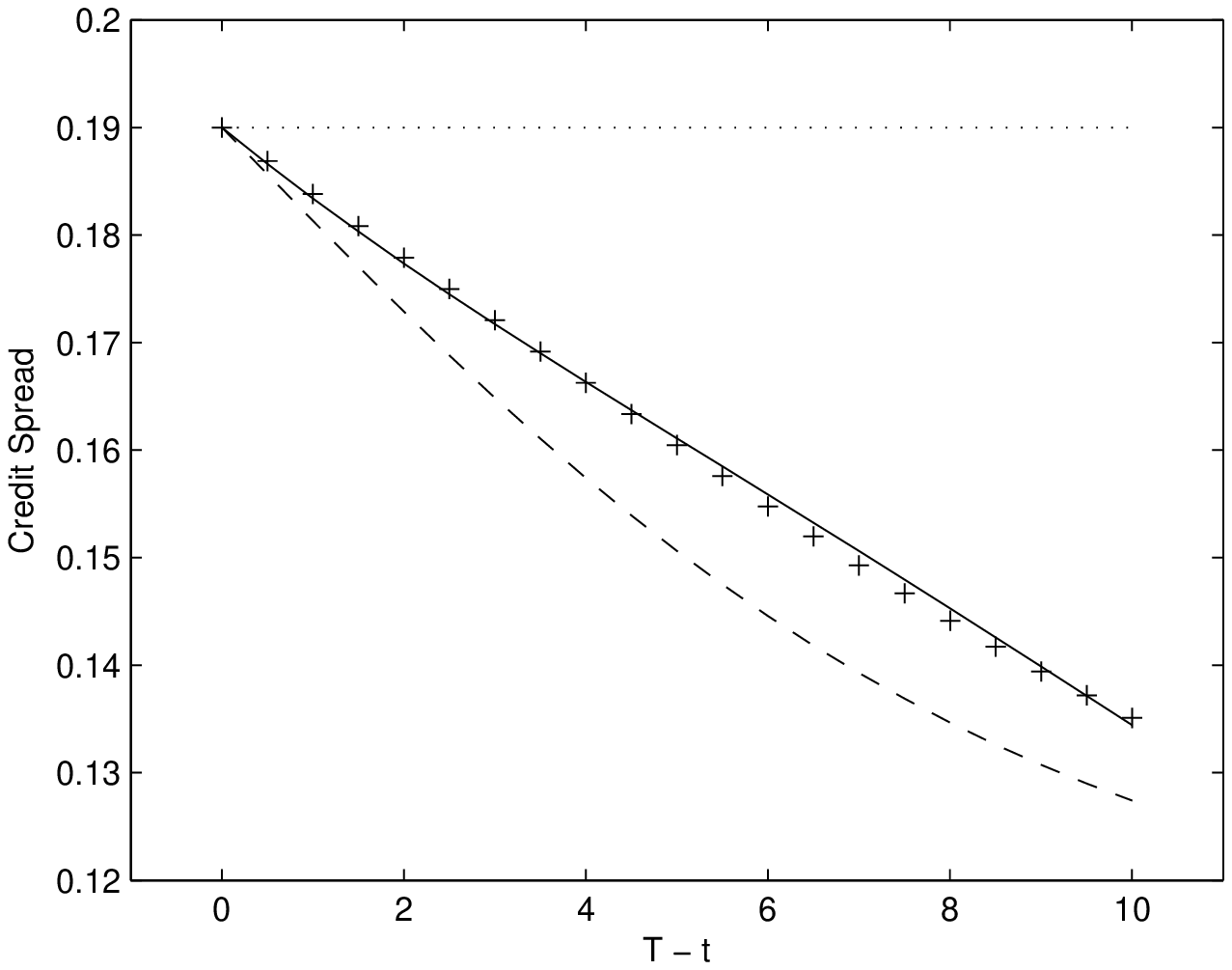}
\end{tabular}
\caption{
Left: survival probabilities $u(t,x;T) := \Pb_x[ \zeta > T | \zeta > t]$ for the JDCEV
model described in Section \ref{sec:JDCEV}. The dotted line, dashed line and solid line correspond
to the approximations $u^{(0)}(t,x;T)$, $u^{(1)}(t,x;T)$ and $u^{(2)}(t,x;T)$ respectively, all of
which are computed using Remark \ref{rmk:jdcev}.  The crosses indicate the exact survival
probability, computed by truncating equation (\ref{eq:u.exact}) at $n=70$.
Our numerical tests indicate that truncating (\ref{eq:u.exact}) at any $n \geq 40$ resulted in numerical values of $u$ that differ by less than $10^{-5}$.
Right:
%The yields $Y(t,x;T)$ on the defaultable bond described in Section \ref{sec:JDCEV}.  The $n$th order yield curve is computed using $Y^{(n)}(t,x;T):= - \log (u^{(n)}(T-t,x)) / (T-t)$.  The dotted line, dashed line and solid line correspond to the approximations $Y^{(0)}(T-t,x)$, $Y^{(1)}(T-t,x)$ and $Y^{(2)}(T-t,x)$ respectively.  The crosses indicate the exact yield, computed using (\ref{eq:u.exact}) and (\ref{eq:Y}).
the corresponding yields $Y^{(n)}(t,x;T):= - \log (u^{(n)}(t,x;T)) / (T-t)$ on a defaultable bond.
The parameters used in the plot are as follows: $x = \log(1)$, $\beta = -1/3$, $b=0.01$, $c=2$ and $a = 0.3$.}
\label{fig:survival}
\end{figure}
% Figure and plot generated using JDCEV2.nb

\begin{table*}
\centering
\begin{tabular}{c|cccc}
$T-t$ & $Y$ & $Y-Y^{(0)}$ & $Y-Y^{(1)}$ & $Y-Y^{(2)}$ \\
\hline
\hline
 1.0 & 0.1835 & -0.0065 & 0.0022 & 0.0001 \\
 2.0 & 0.1777 & -0.0123 & 0.0048 & 0.0003 \\
 3.0 & 0.1720 & -0.0180 & 0.0071 & 0.0003 \\
 4.0 & 0.1663 & -0.0237 & 0.0089 & -0.0001 \\
 5.0 & 0.1605 & -0.0295 & 0.0099 & -0.0006 \\
 6.0 & 0.1548 & -0.0352 & 0.0102 & -0.0011 \\
 7.0 & 0.1493 & -0.0407 & 0.0101 & -0.0013 \\
 8.0 & 0.1442 & -0.0458 & 0.0095 & -0.0011 \\
 9.0 & 0.1394 & -0.0506 & 0.0087 & -0.0005 \\
 10.0 & 0.1351 & -0.0549 & 0.0077 & 0.0007 \\
\hline
\end{tabular}
\caption{The yields $Y(t,x;T)$ on the defaultable bond described in Section \ref{sec:JDCEV}: exact
($Y$) and $n$th order approximation ($Y^{(n)}$).  We use the following parameters: $x = \log(1)$,
$\beta = -1/3$, $b=0.01$, $c=2$ and $\del = 0.3$.} \label{tab:creditspread}
\end{table*}

\begin{figure}
\centering
\begin{tabular}{cc}
$n=0$ & $n=0$\\
\includegraphics[height=.145\textheight]{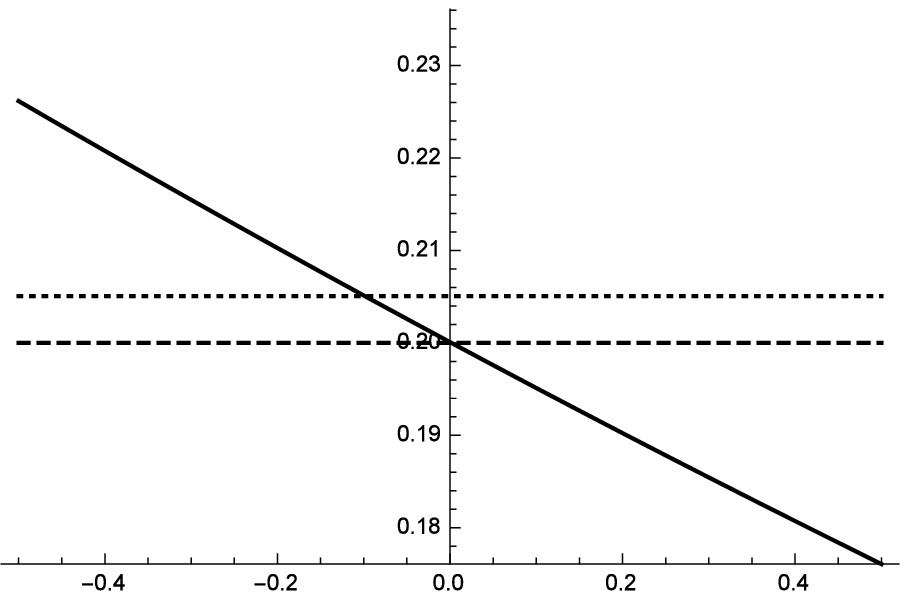} &
\includegraphics[height=.145\textheight]{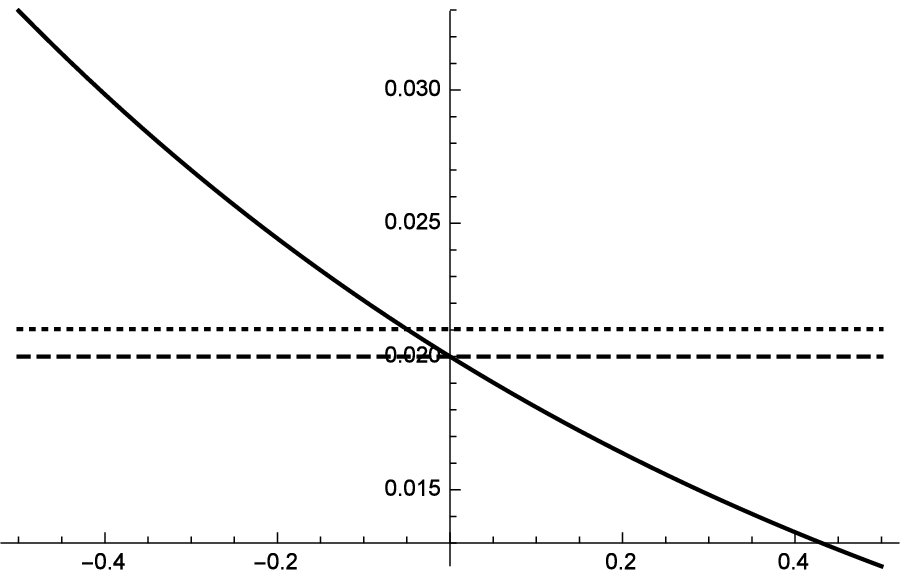} \\
$n=1$ & $n=1$\\
\includegraphics[height=.145\textheight]{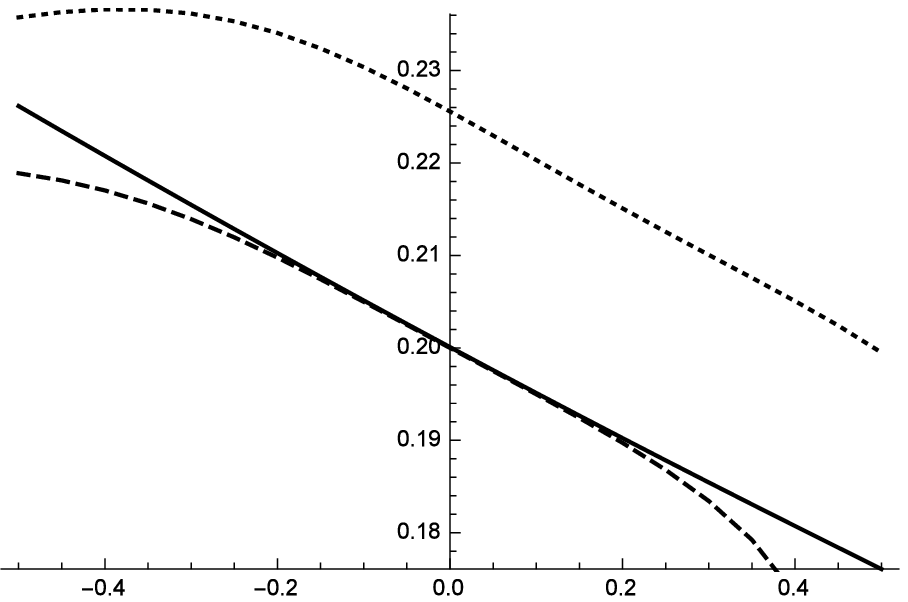} &
\includegraphics[height=.145\textheight]{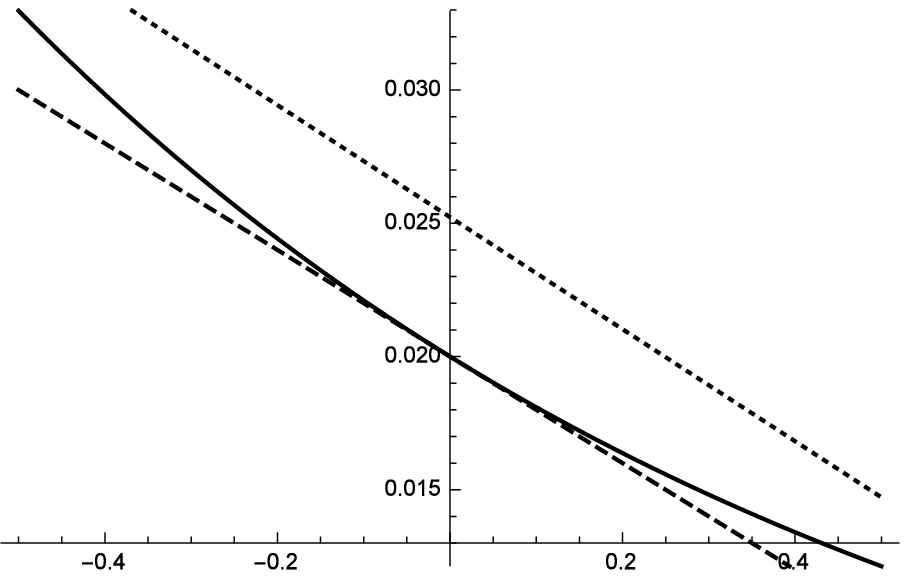} \\
$n=2$ & $n=2$\\
\includegraphics[height=.145\textheight]{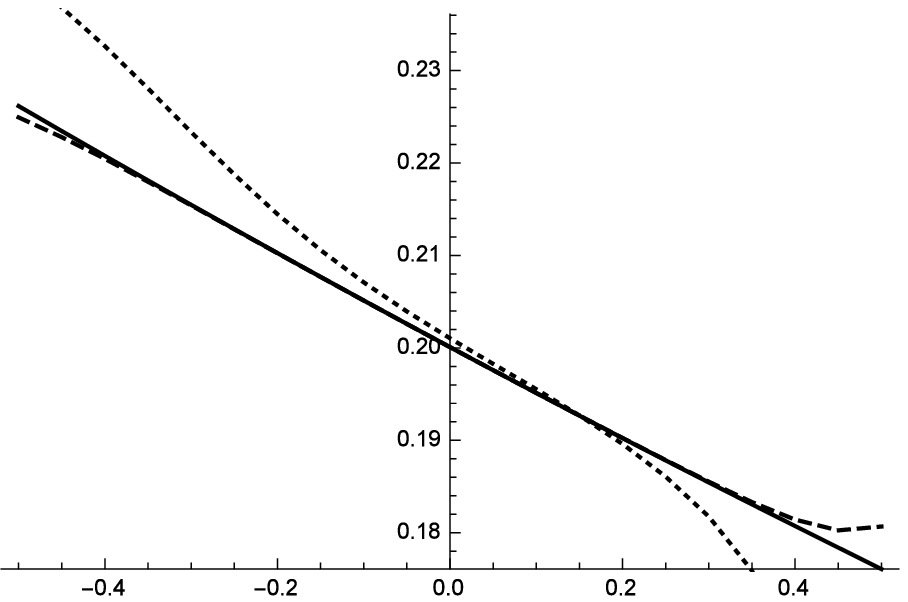} &
\includegraphics[height=.145\textheight]{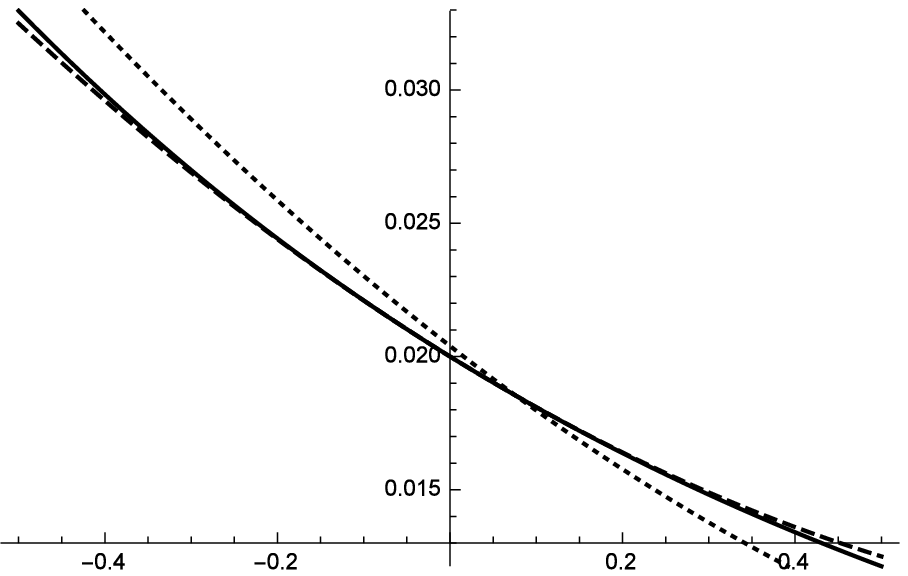} \\
$n=3$ & $n=3$\\
\includegraphics[height=.145\textheight]{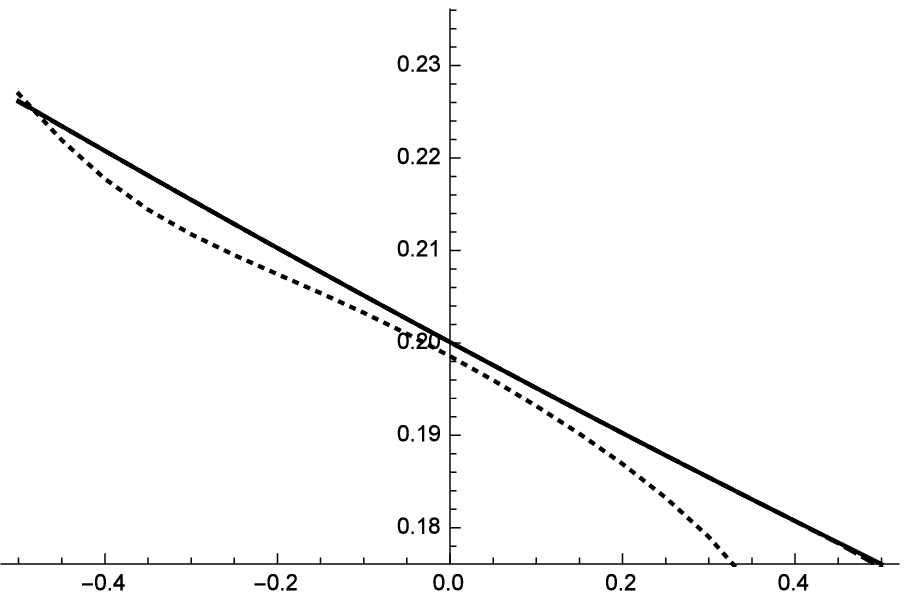} &
\includegraphics[height=.145\textheight]{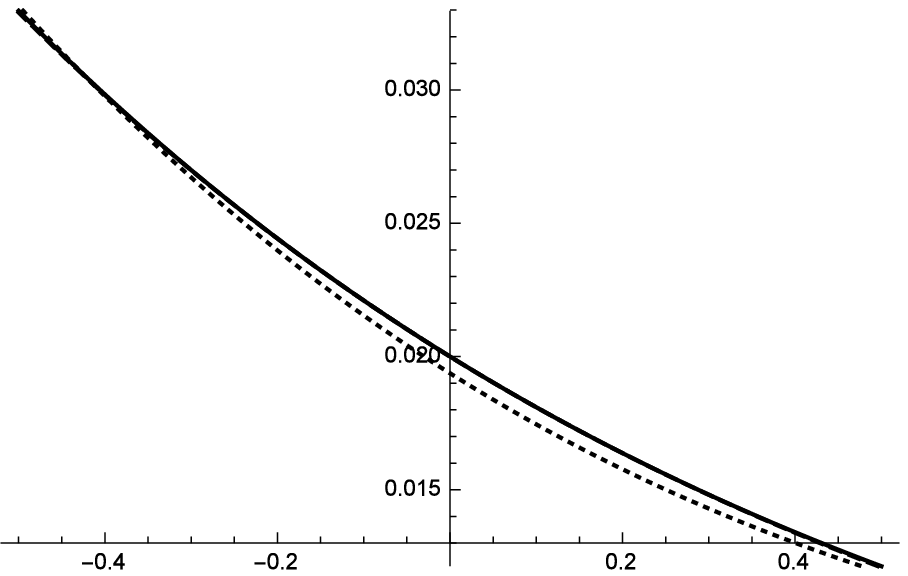} \\
$n=4$ & $n=4$\\
\includegraphics[height=.145\textheight]{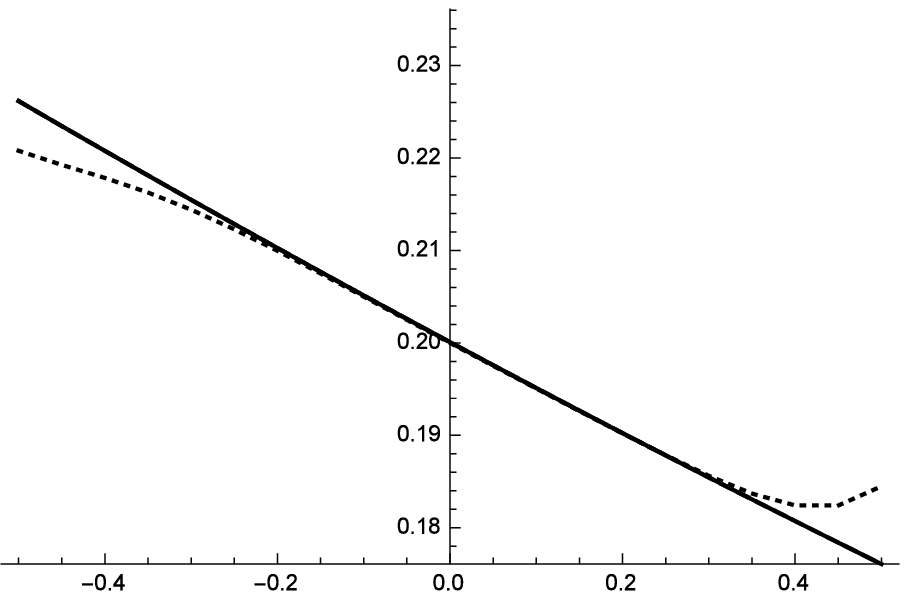} &
\includegraphics[height=.145\textheight]{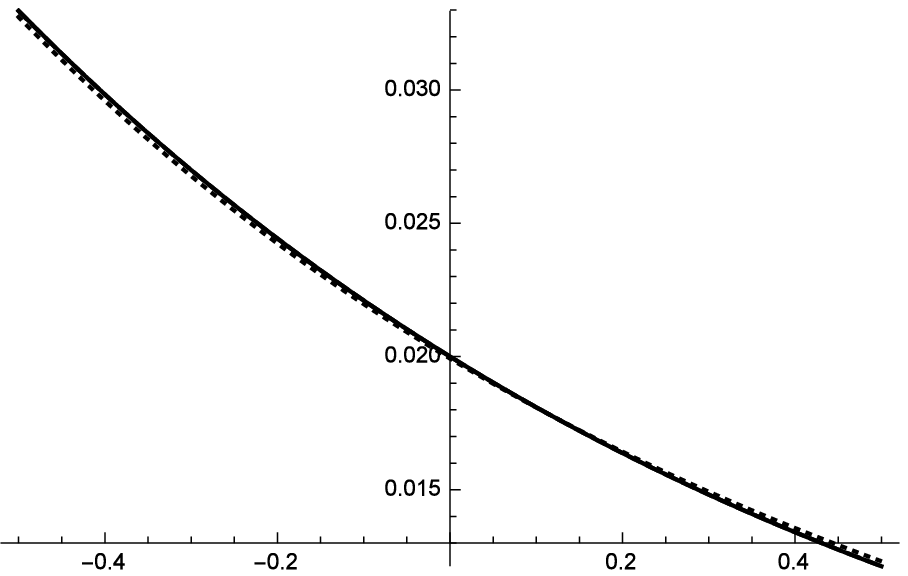} \\
\end{tabular}
\caption{
We consider the CEV model described in Section \ref{sec:hermite} with $x=0$, $T-t=1$, $\del=0.2$ and $\beta=1/2$.
LEFT: We plot  as a function of $\log$ moneyness $(\log K - x)$ the exact implied volatility $\text{IV}[u(t,x;K)]$ (solid) as well as the Taylor and Hermite approximations: $\text{IV}[u_\text{T}^{(n)}(t,x;K)]$ (dashed) and $\text{IV}[u_\text{H}^{(n)}(t,x;K)]$ (dotted).
RIGHT: We plot as a function of $x$ the exact diffusion coefficient $a(x)$ (solid) as well as the $n$th order Taylor and Hermite approximations: $a_\text{T}^{(n)}(x)$ (dashed) and $a_\text{H}^{(n)}(x)$ (dotted).
}
\label{fig:hermite}
\end{figure}
% Taylor-vs-Hermite-CEV-1.01

%%%%%%%%%%%%%%%%%%%%%%% Normal Inverse Gaussian   %%%%%%%%%%%%%%%%%%%%%%%%%%%

\clearpage
\begin{figure}
\centering
\includegraphics[width=\textwidth,height=0.55\textheight]{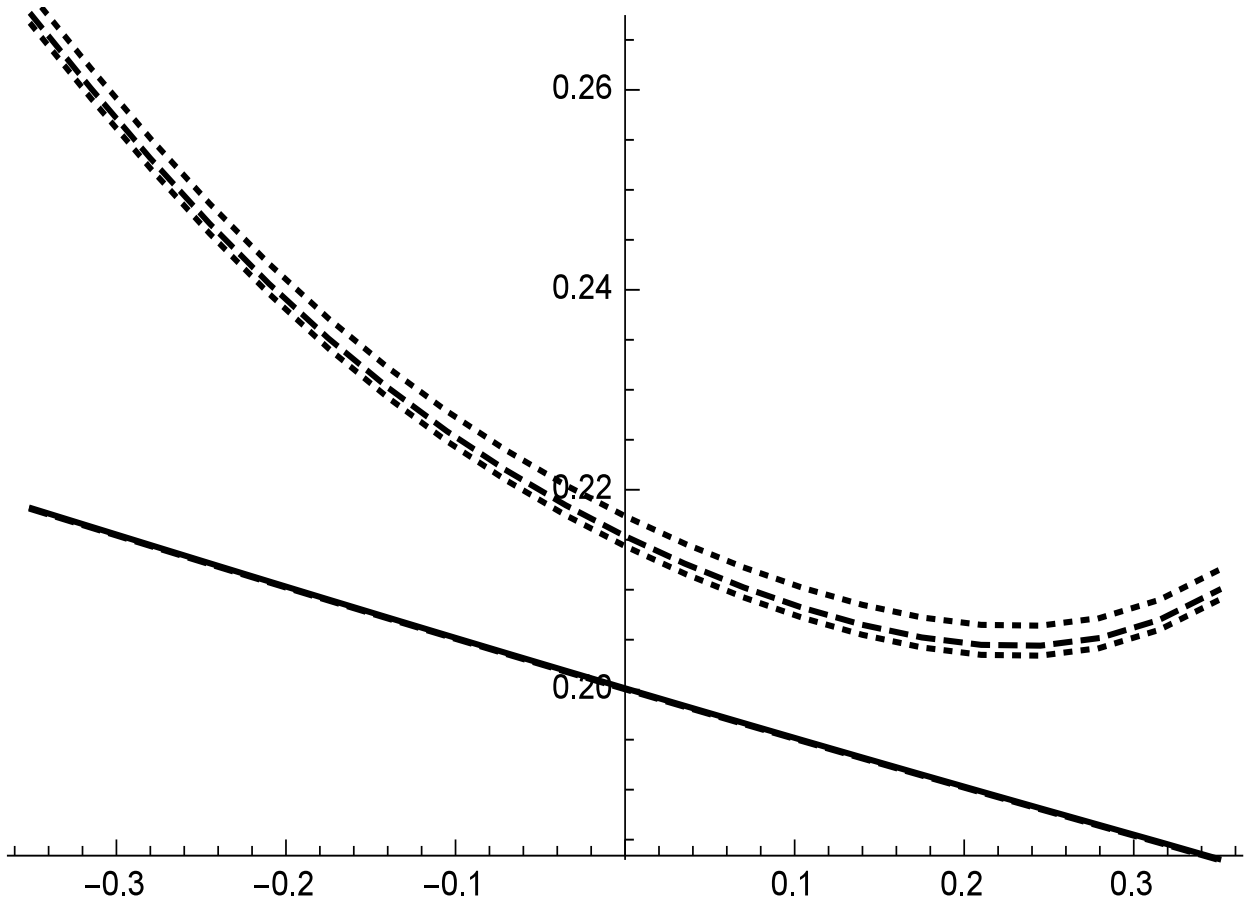}
\caption{Implied volatility (IV) is plotted as a function of $\log$-strike $k := \log K$ for the
model of Section \ref{sec:jump-vs-nojump}.
The solid line corresponds to the implied volatility induced by the exact call price in the case of no jumps.  The dotted lines indicate the $95\%$ confidence interval of implied volatility, computed via Monte Carlo simulation, for the model with jumps.  The dashed lines correspond to the implied volatility induced by our 3rd order Taylor series approximation: $\text{IV}[u^{(3)}(t,x;K)]$.  Note that the bottom dashed line and the solid line are nearly indistinguishable, while the top dashed line falls strictly within the two dotted lines.}
\label{fig:jump-vs-nojump}
\end{figure}

\end{document}